\documentclass[12pt]{article}
\usepackage{amsfonts}
\usepackage{mathrsfs}
\usepackage{amsthm}
\usepackage{multirow}
\usepackage{amssymb}
\usepackage{amsmath}
\usepackage{graphicx,color}
\usepackage{enumerate}
\usepackage{comment}
\usepackage{array}
\usepackage{bm}
\usepackage{booktabs}
\usepackage{indentfirst}
\usepackage{apacite}
\usepackage{subfigure}
\usepackage[authoryear,round,]{natbib}
\usepackage{graphicx} 
\usepackage{epstopdf}
\numberwithin{equation}{section}
\newcommand{\bea}{\begin{eqnarray}}
\newcommand{\eea}{\end{eqnarray}}
\newcommand{\Bea}{\begin{eqnarray*}}
\newcommand{\Eea}{\end{eqnarray*}}
\newcommand{\ba}{\begin{array}}
\newcommand{\ea}{\end{array}}
\newcommand{\bt}{\begin{tabular}}
\newcommand{\et}{\end{tabular}}
\newcommand{\btb}{\begin{table}}
\newcommand{\etb}{\end{table}}
\newcommand{\bc}{\begin{center}}
\newcommand{\ec}{\end{center}}
\newcommand{\beq}{\begin{equation}}
\newcommand{\eeq}{\end{equation}}

\newtheorem{theo}{\bf Theorem}[section]
\newtheorem{lemm}{\bf Lemma}
\newtheorem{prop}{\bf Proposition}[section]

\newtheorem{theorem}{Theorem}[section]

\newtheorem{example}{Example}

\begin{document}

\title{ Interquantile Shrinkage in Spatial Quantile  Autoregressive Regression models }

\author{
Ping  Dong$^{1}$ \ Jiawei Hou$^{2}$\  Yunquan Song$^{2}$\footnote{
The corresponding  author: Yunquan Song. Email: syqfly1980@163.com. This research was supported by NNSF project (61503412) of China, NSF project (ZR2019MA016) of Shandong Province of China.
 }\\
$^1$School of Statistics and Information, \\
Shanghai University of International Business and Economics\\
$^2$School of Science, China University of Petroleum
}
\date{}
\maketitle

\begin{abstract} \baselineskip=18pt
Spatial dependent data  frequently occur in many fields such as spatial econometrics and epidemiology. To deal with the dependence of variables and  estimate quantile-specific effects by covariates, spatial quantile autoregressive models (SQAR models) are introduced. Conventional quantile regression only focuses on the fitting models but ignores the examination of multiple conditional quantile functions, which provides a comprehensive view of the relationship between the response and covariates. Thus, it is necessary to study the different regression slopes at different quantiles, especially in situations where the quantile coefficients share some common feature. However, traditional Wald multiple tests not only increase the burden of computation but also bring greater FDR. In this paper, we transform the estimation and examination problem into a penalization problem, which estimates the parameters at different quantiles and identifies the interquantile commonality at the same time. To avoid the endogeneity caused by the spatial lag variables in SQAR models, we also introduce instrumental variables before estimation and propose two-stage estimation methods based on fused adaptive LASSO and fused adaptive sup-norm penalty approaches. The oracle properties of the proposed estimation methods are established. Through numerical investigations, it is demonstrated that the proposed methods lead to higher estimation efficiency than the traditional quantile regression.

\baselineskip=20pt {\it Key words:} 
Spatial quantile autoregressive models; Quantile regression; Instrumental variables; Fused Adaptive LASSO; Fused Adaptive Sup-norm


\end{abstract}

\baselineskip=18pt

\newpage

\section{Introduction}
In many fields including spatial econometrics, epidemiology and regional science, spatial dependent data is a common data type to be observed. Various spatial regression models are suggested for it.
Among them, a kind of simple and intuitive spatial autoregressive (SAR) model caught many academics' attention, which was first proposed by \cite{Cliff}.
Although a lot of studies have been done for models with independent and identically distributed random (i.i.d) disturbances, researchers has realized the heteroscedasticity exists in modelling the spatial data such as for unemployment data, crime rates data, housing prices, etc, see \cite{Anselin-B}, \cite{LeSage}, \cite{Zhang-S}.
\cite{Lin-Lee} extended the generalized method of moments (GMM) method to SAR model which allowed for heteroscedasticity. Then \cite{Kelejian} considered the GMM estimation with heteroscedasticity for a more general spatial model.
Clearly, all these models focused on the conditional mean function and were easily  affected by heteroscedasticity.

Quantile regression (QR) has attracted an increasing amount of attention after being introduced by \cite{Koenker}. It could provide more comprehensive statistical views than traditional regression which was only in the aspect of the condition mean, while quantile regression was studied at multiple quantiles. Quantile regression allows researches to explain the heteroscedasticity through the quantile coefficients not the disturbances.
Combining the SAR model and QR model opened up a new and exciting research direction, spatial quantile autoregressive (SQAR) model, which was been put forward by \cite{Su-Y}. \cite{Su-Y} offered an alternative view for allowing unknown heteroscedasticity in the SAR model, taking into account of both unobserved heterogeneity and spatial dependence.

Conventional multiple-quantile regression methods often carry out analysis at each quantile level separately. However, if the quantile coefficients share some common features across quantile levels, the slope coefficients may appear constant only at a certain quantile region. It's hard for the traditional quantile regression estimation to identify. Thus, distinguishing the quantile regions between constant coefficients and non-constant coefficients is the key problem. A natural idea is to take a hypothesis test, so \cite{Koenker-2005} brought the Wald test to distinguish the commonality of quantile slopes, but this method becomes infeasible for large number of quantiles or predictors by greater False Discovery Rate (FDR).

Fortunately, penalization methods are useful tools to deal with the nonsignificant differences. They are often applied for variable selection to choose and estimate the sparse regression models. Recently, the idea of penalization has been more popular and widely used in machine learning and artificial intelligence techniques.
The earliest the penalization was proposed by the \cite{Hoerl}, where ridge regression was a $L_{2}$-norm regularization method for nonorthogonal problems.
\cite{Tibshirani-1996} suggested $L_{1}$-norm regularization to linear regression to shrinkage the insignificant coefficients to zero, namely least absolute shrinkage and selection operator (LASSO), indicating the beginning of variable selection. For parametric models, \cite{Fan} proposed a class of variable selection procedures based on nonconcave and penalized likelihood method, named the smoothly clipped absolute deviation (SCAD) penalty. Then \cite{Fan-2004} extended the variable selection to partially linear models for longitudinal data.
\cite{Tibshirani-S} introduced the fused LASSO, where pairwise differences between variables were penalized by the $L_{1}$-norm. To obtain a convex objective function, \cite{Zou} and \cite{Friedman} proposed adaptive weights  $L_{1}$-penalty, namely adaptive LASSO. \cite{Rinaldo} modified it to fused adaptive LASSO with better properties.
For grouping structure model, \cite{Yuan-L} introduced group LASSO to identify significant groups of predictors. \cite{Zhang-X} studied the oracle properties of adaptive group LASSO in high-dimensional linear models.

With the development of variable selection methods, \cite{Zou-Y-a} estimated the common slopes by a composite quantile regression method and selected nonzero slopes by adaptive Lasso. \cite{Jiang-2013} used fused adaptive Lasso and fused adaptive sup-norm to smooth neighboring quantiles.
\cite{Ciuperca} extended adaptive group LASSO to quantile model with grouped predictors.

However, if the spatial lag parameter is nonzero, there exists endogeneity caused by the spatial lag variables. This causes the increasing difficulties of estimating the coefficients of SQAR model including the spatial lag parameter and the regression slopes.
\cite{Kim-M} suggested a double stage quantile regression (DSQR). They use quantile regression with random regressors after the endogeneity.
\cite{Su-Y} applied an instrumental variable quantile regression (IVQR) estimator to SQAR model. \cite{Kostov} suggested empirical likelihood quantile regression (ELQR) estimation, which is a non-parametric analogue of likelihood estimation for the linear model.
\cite{Xu-L} considered the instrumental variable (IV) combined MLE for estimation for SAR model with a nonlinear transform of dependent variable.

To avoid the endogeneity and identify the interquantile commonality, we combine instrumental variable and penalization approaches together, and suggest a novel estimation method for SQAR model. In the method, we adopt the fusion idea to shrink the differences of quantile slopes at two adjacent quantile levels toward zero, which could help to employ automatic estimation, and detect quantile regions with constant slope coefficients in spatial quantile autoregressive models at the same time. Two types of fusion penalties were applied in the multiple-quantile regression model: Fused Adaptive LASSO (FAL) and Fused Adaptive sup-norm (FAS).

The remainder of this article is organized as follows. In Section 2, we illustrate the proposed methods and give the asymptotic properties of the proposed penalization FAL and FAS estimators. In Section 3, we discuss the computation issues. A simulation study is conducted to assess the numerical performance of our proposed estimators in Section 4. We apply the proposed methods to analyze international economic growth data in Section 5. All technical details are added in the Appendix.

\section{Methodology}\label{Methodology}
\subsection{Spatial Autoregressive Model}\label{SAR}

SAR model has the following form,
\begin{equation}
	  {\bm{Y}} = \alpha + \lambda {\bm{W}} {\bm{Y}} +   {\bm{X}} {\bm{\beta}} + {\bm{\varepsilon}},
\label{SARModel0}
\end{equation}
where $ {\bm{Y}} = (Y_1,\cdots,Y_n)^T$ is an $n \times 1$ vector response value, and $ {\bm{X}} = ({\bm{X}}_1, \cdots, {\bm{X}}_n)^T$ 
is a matrix of $n$ observations on $p$ exogenous covariates, ${\bm{W}} = \{w_{ij}\}$ is a known $n \times n$ spatial weights matrix, ${\bm{\varepsilon}} = (\varepsilon_1, \cdots, {\varepsilon}_n)^{T}$ denotes an $n$-vector of i.i.d random disturbances with zero mean and finite variance, see \cite{Anselin}. In SAR model there are three type model parameters. The first one is the intercept term $\alpha \in \mathbb R$. The second is nonstochastic spatial lag parameter $\lambda \in \mathbb R$, representing the autocorrelation of response variable. The last one is the the regression coefficient vector ${\bm{\beta}} = (\beta_1, \cdots,\beta_p)^T \in \mathbb R^{p \times 1}$.

The SAR model in equation (\ref{SARModel0}) is also recast as
\[
Y_i = \alpha + \lambda \sum_{j = 1}^n w_{ij}Y_j + {\bm{\beta}}^{T} {\bm{X}}_i + \varepsilon_i,
\]
which intuitively implies that the response of the $i$-th subject is linearly depends on its neighbors and covariates. 
It is usually supposed that the noise of $\varepsilon_i$s are independent and identically  distributed with zero mean and finite variance.  
It means that the covariance  $\textnormal{Cov}({\bm{\varepsilon}}) =\sigma^2{\bm{I}}_n$, with ${\bm{I}}_n \in \mathbb R^{n \times n}$ is an identity matrix.
Then the observations of ${\bm{Y}}$ could be formulated as

\begin{equation}\label{SARModel1}
 {\bm{Y}} = ({\bm{I}}_n - \lambda{\bm{W}})^{-1}(\alpha + {\bm{X}}{\bm{\beta}} + {\bm{\varepsilon}}),
\end{equation}
where $({\bm{I}}_n - \lambda{\bm{W}})$ is needed to guaranteed the invertibility.
According to the work of \cite{Banerjee}, the matrix ${\bm{W}}$ has its largest singular value of 1 under certain normalization operations.
Therefore, $|\lambda| < 1$ is a sufficient condition to ensure the invertibility of $({\bm{I}}_n - \lambda{\bm{W}})$. Based on this, we add one constraint of $\lambda$, that is $|\lambda| < 1$.

Assuming we have weights matrix  ${\bm{W}} $,  let ${\bm{\theta}}= (\alpha, \lambda, {\bm{\beta}}^{T})^{T}$, the log likelihood function can be written as
\begin{equation}\label{Logfunction}
l({\bm{\theta}}, \sigma^2) = -\frac{n}{2}\log(2 \pi \sigma^2) + \log|{\bm{I}}_n - \lambda{\bm{W}}| - \frac{1}{2\sigma^2}||({\bm{I}}_n - \lambda {\bm{W}}){\bm{Y}} - \alpha - {\bm{X}}{\bm{\beta}}||^2
\end{equation}
To solve the maximum optimization problem, we first fix $\bm{\theta}$ and get the sandwiched estimator of $\sigma^2$, that is
\begin{align*}
\tilde{\sigma}^2
= \left[{\bm{Y}} - {\bm{A}}(\lambda)^{-1}(\alpha + {\bm{X}}{\bm{\beta}})\right]^{T}\left[{\bm{A}}(\lambda)^{-1}{\bm{A}}(\lambda)^{T}\right]^{-1}\left[{\bm{Y}} - {\bm{A}}(\lambda)^{-1}(\alpha + {\bm{X}}{\bm{\beta}})\right],
\end{align*}
where ${\bm{A}}(\lambda)= {\bm{I}}_n - \lambda {\bm{W}}$ is a $n \times n$ matrix related to $\lambda$.
Then apply the estimator $\tilde{\sigma}^2$ back to $l({\bm{\theta}}, \sigma^2)$ and obtain $l({\bm{\theta}}) = l({\bm{\theta}}, \tilde{\sigma}^2)$. Take $l({\bm{\theta}})$ as the objective function and the MLE estimator of ${\bm{\theta}}$ could be derived by $\hat{{\bm{\theta}}}_{\rm{MLE}} = \arg \max_{{\bm{\theta}}}l({\bm{\theta}})$. The parameters could be estimated as $\hat{{\bm{\theta}}}_{\rm{MLE}} = ({\hat{\alpha}_{\rm{MLE}}}, {\hat{\lambda}}_{\rm{MLE}}, {\hat{\bm{\beta}}_{\rm{MLE}}}^T )^T$ and the detailed asymptotic properties were provided in \cite{Anselin-B}.

Although the MLE method has many excellent theoretical properties, its computational cost is huge. When $n$ is large, the classical MLE approach becomes computationally expensive, mainly due to the high cost of computing $\log |{\bm{A}}(\lambda)|$ and the relative computational complexity is $O(n^3)$, see \cite{Trefethen-B}, \cite{Barry-P}, \cite{Smirnov-A}). In order to reduce the computational complexity, \cite{Ma-P} proposed a naive least squares method for SAR models, which was simple and effective. The least squares objective function is
\begin{equation}
Q({\bm{\theta}}) = ||{\bm{Y}} - \alpha - \lambda {\bm{W}}{\bm{Y}}  - {\bm{X}}{\bm{\beta}}||^2,
\label{LS-object}
\end{equation}

Equation (\ref{LS-object}) was designated a ordinary least squares problem, where ${\bm{Y}} $ is the response value and $({\bm{W}}{\bm{Y}}, {\bm{X}})$ are the covariates. The OLS estimator of ${\bm{\theta}}$ could be obtained by $\hat{{\bm{\theta}}}_{\rm{OLS}} = \arg \min_{{\bm{\theta}}}Q({\bm{\theta}})$, that is $\hat{{\bm{\theta}}}_{\rm{OLS}} = ({\hat{\alpha}}_{\rm{OLS}}, {\hat{\lambda}}_{\rm{OLS}}, {\hat{\bm{\beta}}_{\rm{OLS}}}^T )^T$. The asymptotic properties of $\hat{{\bm{\theta}}}_{\rm{OLS}}$ under certain conditions were showed in \cite{Huang-L} and \cite{Ma-P}.

However, there remains the problem of endogeneity caused by the spatial lag variables, leading covariates $({\bm{W}}{\bm{Y}}, {\bm{X}})$ of interest to be correlated with the error term ${\bm{\varepsilon}}$.
Intuitively, instrumental variable (IV) could be adopted, which was first proposed by \cite{Wright}. A valid instrument induces changes in the explanatory variable ${\bm{X}}$, but has no independent effect on the dependent variable ${\bm{Y}}$.
\cite{Xu-L} considered the instrumental variable combined MLE for estimation for SAR model with a nonlinear transform of dependent variable.

One computational method of calculating IV estimators is two-stage least squares. Inspired by it, we combine the naive least squares method and instrumental variable together and consider a two stage least squares estimation method, similar to \cite{Xie}.
Denote the ${\bm{U}} = {\bm{W}}{\bm{Y}}$,
we hope to find a replacement ${\bm{\hat{U}}}$ such that ${\text{Cov}}({\bm{\hat{U}}}, {\bm{\varepsilon}}) = 0$ and ${\text{Cov}}({\bm{\hat{U}}}, {\bm{X}}) \neq 0$.
In the first stage, regress each explanatory variable on all of the exogenous variables in the model.
Naturally, exogenous variables ${\bm{V}}$ could be selected as ${\bm{W}}({\bm{I}}_n - \lambda {\bm{W}})^{-1}{\bm{X}}$, according to the deformation of equation (\ref{SARModel0}),
\begin{equation}
 {\bm{Y}} = ({\bm{I}}_n - \lambda {\bm{W}})^{-1} {\alpha} + ({\bm{I}}_n - \lambda {\bm{W}})^{-1}{\bm{X}}{\bm{\beta}}+ ({\bm{I}}_n - \lambda {\bm{W}})^{-1} {\bm{\varepsilon}}.
 \label{SARModel2}
\end{equation}
But the spatial parameter $\lambda$ is unknown, so we choose the first term of ${\bm{W}}({\bm{I}}_n - \lambda{\bm{W}})^{-1}{\bm{X}}$ and let
${\bm{V}} = [{\bm{1}}_n, {\bm{X}}, \bm{WX}]$, where ${\bf{1}}_n$ is a $n \times 1$ vector with all $1$'s,. The first stage is regressing ${\bm{U}}$ on the exogenous variables ${\bm{V}}$, and assume the regression equation has the following representation
\begin{equation}
{\bm{U}} = {\bm{V}} {\bm{\Pi}} + {\bm{v}},
\label{SAR-IV}
\end{equation}
where ${\bm{v}}$ is the unknown error term different from ${\bm{\varepsilon}}$. The estimator of ${\bm{\Pi}}$ could be derived as $\hat{{\bm{\Pi}}} = ({\bm{V}}^T{\bm{V}} )^{-1}{\bm{V}}^T {\bm{U}}$ and the the predicted value is $\hat{{\bm{U}}} = {\bm{V}} \hat{{\bm{\Pi}}}$.

Then replace ${\bm{W}}{\bm{Y}}$ by the predictor ${\bm{\hat{U}}}$. Thus, the correlation between the spatially lagged endogenous variable and the error term is eliminated. The second stage is taking a least squares regression for
\begin{equation}
{\bm{Y}} = \alpha + \lambda {\bm{\hat{U}}} +   {\bm{X}} {\bm{\beta}} + {\bm{\varepsilon}}.
\label{SARModel3}
\end{equation}
Combining equation (\ref{SARModel2}) and (\ref{SARModel3}), the estimators of $\bm{\theta}$ could be obtained. The same strategy for choosing instrumental variables can be found in \cite{McMillen}.

\subsection{Spatial Quantile Autoregressive Model} \label{SQARM}
Denote $U_i = \sum_{j=1}^{n}w_{ij}Y_j$ as the $i$th element of ${\bm{U}} $, which is a $p \times 1$ vector of strictly exogenous regressors, and $\varepsilon_i$ are the iid error terms that are independent of ${\bm{X}}_i$. We consider the following location-scale model,
\[
Y_i = \alpha + \lambda U_i + {\bm{\beta}}^{T} {\bm{X}}_i  + \varepsilon_i.
\]

Suppose we are interested in regression at $K$ quantile levels $0<\tau_{1}<\cdots<\tau_{K}<1$, where $K$ is a finite integer. Denote $Q_{\tau_k}(Y_i|\mathcal{F}_{-i}, {\bm{X}}_i)$ as the $\tau_{k}$th conditional quantile function of given $\mathcal{F}_{-i}$ and ${\bm{X}}_i$, $k = 1, 2, \cdots,K$. $\mathcal{F}_{-i}$ is the $\sigma$-field of $\{Y_j: j \neq i\}$. $\tau_k$th quantile of $\epsilon_i$ is zero, that is $Q_{\tau_k}(\epsilon_i|\mathcal{F}_{-i}, {\bm{X}}_i) = 0 $ for $i = 1, 2, \cdots, n$. Then the spatial quantile autoregressive (SQAR) model is represented as following,

\begin{equation}\label{SQARmodel1}
Q_{\tau_k}(Y_i|\mathcal{F}_{-i}, {\bm{X}}_i)  = \alpha_k + \lambda_{k} {U_i}  +  {\bm{\beta}}^{T}_k {\bm{X}}_i,
\end{equation}
where  $k$ is corresponding to the $\tau_k$th quantile of $\varepsilon_i$, $\alpha_k$ is the intercept term that is $\tau_k$-independent, $\lambda_k$ is the scalar spatial lag parameter that is $\tau_k$-dependent, and ${\bm{\beta}}_k$ is a $p$-vector regression parameters that is also $\tau_k$-dependent.
In matrix form, (\ref{SQARmodel1}) equals
\begin{equation}\label{SQARmodel2}
{\bm{Y}}  = \alpha_k + \lambda_{k} {\bm{W}}{\bm{Y}}  +  {\bm{X}} {\bm{\beta}}_k + {\bm{\varepsilon}}
\end{equation}

where we are interested in the parameter ${\bm{\theta}}_{(k)} = \left(\alpha_k, \lambda_k, {\bm{\beta}}_{k}^{T} \right)^{T}$.

For SQAR models, \cite{Su-Y} applied an instrumental variable quantile regression(IVQR) estimator. But it needs grid searching and costs a lot of computation time. So, we adopt a simpler approach, which is the straight forward extension of two stage least squares estimation method, see \cite{Kim-M}, \cite{Zietz}, \cite{Liao-W}. This is similar to the procedure in Section \ref{SAR}.

At $\tau_k$ quantile, we assume that ${\bm U}={\bm{W}}{\bm{Y}}$ has the following representation:
\begin{equation}
{\bm{U}} = {\bm V} {\bm\Pi}_k + {\bm v},
\end{equation}
where ${\bm V} = [{\bm {1}}_n, {\bm X}, {\bm W}{\bm X}]$ is a $n \times (2p + 1)$ matrix,  ${\bm\Pi}_k$ is a $\tau_k$-dependent $(2p+1) \times 1 $ matrix of unknown parameters and ${\bm v}$ is a $n \times 1$ matrix of error terms. In the first stage, we estimate quantile regressions for ${\bm{U}}$ using the instruments variables ${\bm{V}}$ as explanatory variables by minimizing objective function as following,
\begin{equation}\label{Objfunction-00}
    R_{IV}^0(\tau_k, {\bm\Pi}_k)  = \sum_{i = 1}^{n} \rho_{\tau_k} (U_i -  {\bm V}_i {\bm\Pi}_k),
\end{equation}
where $\rho_{\tau}(r)= r \psi_{\tau}(r) = \tau r I(r>0)+(\tau-1) r I(r \leq 0)$ is the quantile check function and $I(\cdot)$ is the indicator function, where $\psi_{\tau}(r) = \tau - 1_{[r \leq 0]}$ and $1_{[\cdot]}$ is the Kronecker index, see Koenker and Bassett (1978).
Denote the predictor of ${\bm{U}}$  as $\hat{\bm{U}}_k$ at $\tau_k$ quantile.

The reduced-form representation of $\bm Y$ is
\begin{equation}
{\bm{Y}} = {\bm V} {\bm\pi}_k + {\bm u},
{\label{Y_represent}}
\end{equation}
where ${\bm \pi}_k = \begin{bmatrix}
\begin{pmatrix}
&{ 1} \\
&{\bm 0}_{p}\\
&{\bm 0}_{p}\\
\end{pmatrix},
{{\bm \Pi}_{k}},
\begin{pmatrix}
& {0} \\
&{\bm {1}}_p \\
&{\bm 0}_{p}
\end{pmatrix}
\end{bmatrix}{\bm {\theta}}_{(k)} = H({\bm \Pi}_k) {\bm {\theta}}_{(k)}$, ${\bf{0}}_p$ is a $p \times 1$ vector with all $0$'s, ${\bf{1}}_p$ is a $p \times 1$ vector with all $1$'s and ${\bm u} = \lambda_k{\bm v} + {\bm \varepsilon}$.
Another representation of (\ref{Y_represent}) is
\begin{equation}
 {\bm{Y}} = \alpha_k + \lambda_k {\bm V} {\bm\Pi}_k + {\bm{X}} {\bm{\beta}}_k + {\bm u}.
\end{equation}
In the second stage, we replace ${\bm V} {\bm\Pi}_k$ by estimator, and apply quantile regressions of ${\bm{Y}}$ on ${\bm{X}}$ and $\hat{\bm{U}}_k$.
Thus, the correlation between the spatially lagged endogenous variables and the error terms are eliminated.

The objective function could be defined as,
\begin{equation}\label{Objfunction-0}
    R_{IV}^1(\tau_k, \alpha_k, \lambda_k, {\bm{\beta}}_k)  = \sum_{i = 1}^{n} \rho_{\tau_k} (Y_i - \alpha_k - \lambda_k {\hat{U}}_{ki} - {\bm{\beta}}_k^T{\bm{X}}_{i})
\end{equation}
where ${\hat{U}}_{ki}$ is the $i$th element of $\hat{\bm{U}}_k$.

Minimizing the quantile loss function at each quantile level separately is equivalent to minimizing the following combined loss function,

\begin{equation}\label{Objfunction-sum}
    \sum_{k=1}^K R_{IV}^1 (\tau_k, {\bm{\theta}}_{(k)}) = \sum_{k=1}^K \sum_{i = 1}^{n} \rho_{\tau_k} (Y_i - \alpha_k - \lambda_k {\hat{U}}_{ki} - {\bm{\beta}}_{k}^{T}{\bm{X}}_i).
\end{equation}

Denote $\mathbb{X}_i = (1, {\bm X}_i)^T$, $F_{i}$ as the conditional cumulative distribution function of $Y$ given $\mathcal{F}_{-i}$ and ${\bm{X}}_i$.
${\bm V}_i$ is the $i$th row vector of ${\bm V}$. To establish the asymptotic properties of the estimators $\hat{\bm \alpha}_{k}$ obtained by minimizing equation (\ref{Objfunction-sum}), we assume the following regularity conditions:

\begin{itemize}
    \item [(A1)] For $k=1, \cdots, K, i=1, \cdots, n,$ the conditional density function of $Y$ given $\mathcal{F}_{-i}$ and ${\bm{X}}_i$,  denoted as $f_{i},$ is continuous and has a bounded first derivative, and $f_{i}\left\{Q_{\tau_k}(Y_i|\mathcal{F}_{-i}, {\bm{X}}_i)\right\}$ is uniformly bounded away from zero and infinity.
    \item [(A2)] The row and column sums of the matrices ${\bm{W}}$ and ${\bm{I}}_n - \lambda {\bm{W}}$ are bounded uniformly in absolute value.
    \item [(A3)] For matrix ${\bm{G}} = {\bm{W}}({\bm{I}}_n - \lambda {\bm{W}})^{-1}$, there exists a constant $c$ such that $c{\bm{I}}_n - {\bm{G}}{\bm{G}}^{T}$ is positive semidefinite for all $n$.
    \item [(A4)] For all $\tau_k, k=1, 2, \cdots, K$, ${\bm{\theta}}_{(k)}$ is in the interior of the set $\mathcal{R} \times \mathcal{D}$, and $\mathcal{R} \times \mathcal{D}$ is compact and convex.
    \item[(A5)] The sequence $\{(u_i, v_i , \mathbb{X}_i)\}$ is independent and identically distributed, where $u_i$ and $v_i$ are the $i$th elements in $\bm u$ and $\bm v$ respectively.
    \item[(A6)] The second geometric moment of ${\bm X}$ and the third geometric moment of ${\mathbb X}_i$ are finite, that is $E(\|{\bm X}\|^2) < \infty$ and $E(\|\mathbb{X}_i\|^3) < \infty$.
    \item[(A7)] $H({\bm \Pi}_k)$ is full column rank, for $k = 1, 2, \cdots, K$.
    \item[(A8)] The condition densities $g_1(\cdot|x)$ and $g_2(\cdot|x)$, respectively for $u_i$ and $v_i$, are Lipschitz continuous for all $x$. Moreover, $B_1 = E\{g_1(0|{\mathbb{X}}_i){\mathbb{X}}_i {\mathbb{X}}_i^{T}\}$ and $B_2 = E\{g_2(0|{\mathbb{X}}_i){\mathbb{X}}_i \mathbb{X}_i^{T}\}$ are finite and positive definite.
    \item[(A9)] $E(\psi_{\tau}(u_i)|{\mathbb{X}}_i) = 0$ and $E(\psi_{\tau}(v_i)|{\mathbb{X}}_i) = 0$.

\end{itemize}
Assumption (A1) is the basic assumption for quantile regression.  Assumptions (A2)-(A4) are required in the setting of SAR model, see \cite{Lee}; \cite{Zhang-S}. Assumptions (A5)-(A9) are needed for two stage IV estimators, see \cite{Kim-M}.

\begin{theo}
 Assuming conditions $(A1)$-$(A9)$ hold, we have
\[n^{1/2}(\hat{\bm{\theta}}_{(k)} - {\bm{\theta}}_{(k)}) \stackrel{d}\longrightarrow N({\bm 0},  {\bm \Sigma_{k}}), ~~~\mbox{as}~~~ n \rightarrow \infty,\]
where ${\bm \Sigma_{k}} = D_{k} \Omega D_{k}^{T}$, $D_{k} = (H({\bm \Pi}_k)^{T} A_1 H({\bm \Pi}_k))^{-1} H({\bm \Pi}_k)^{T}[{\bm I} , A_1 A_2^{-1}\lambda_k]$, $\Omega = E({\bm{\Sigma}}_u \otimes {\bm V}_i {\bm V}_i^{T})$, $A_1 = E\{g_1(0|{\bm{V}}_i){\bm{V}}_i {\bm{V}}_i^{T}\}$, $A_2 = E\{g_2(0|{\bm{V}}_i){\bm{V}}_i {\bm{V}}_i^{T}\}$, ${\bm{\Sigma}}_u$ is the matrix of general term $\psi_{\tau}(u_i)\psi_{\tau}(v_i)$.
\end{theo}

However, in some applications, the quantile slope may be constant in certain quantile regions for some predictors. If we still take estimations at each quantile level, the information of common features will be ignored, leading to the reduction of the efficiency. The best strategy is to borrow information from neighboring quantiles, see \cite{Zou-Y-b}; \cite{Jiang-2014}.

In the following sections, we denote $\beta_{k, 0} = \lambda_k$, and $\beta_{k, l}$ as the slope corresponding to the $l$th predictor at the quantile level $\tau_{k}$ where $l = 1, 2, \cdots, p$ and $k = 1, 2, \cdots, K$.
Denote $d_{k, l}=\beta_{k, l}-\beta_{k-1, l}$ as the slope difference at two neighboring quantiles $\tau_{k-1}$ and $\tau_{k},$ with $k=2, \cdots, K$ and $d_{1, l}=\beta_{1, l}$ for $l=0, 1, 2,  \cdots, p$. The parameter vector $\bm{\theta}=\left({\bm{\alpha}}^T,  \boldsymbol{d}_{1}^{T}, \cdots, \boldsymbol{d}_{K}^{T}\right)^{T} \in \mathbb{R}^{(p + 1)K}$ denote the collection of unknown parameters, where $\bm{\alpha} = \left(\alpha_{1}, \cdots, \alpha_{K}\right)^{T}$, and $\boldsymbol{d}_{k}=\left(d_{k, 0}, d_{k, 1}, \cdots, d_{k, p}\right)^{T}$.
Therefore, the $\tau_{k}$ th quantile
coefficient vector can be written as
\[
{\bm{\alpha}}_{k} = \left(\alpha_k, \lambda_k, {\bm{\beta}}_{k}^{T} \right)^{T}=T_{k} \bm{\theta},
\]
where ${\bm{T}}_{k}=\left(\bm{D}_{k, 0}, \bm{D}_{k, 1}, \bm{D}_{k, 2} \right)  \in \mathbb{R}^{(p+2) \times(p+2) K}$. $\bm{D}_{k, 0}$ is a $(p+2) \times K$ matrix with 1 in the
first row and the $k$ th column, but zero elsewhere, that is
\[
\boldsymbol{D}_{k, 0} =
    \begin{pmatrix}
        0 & \cdots & 1 & \cdots & 0 \\
        0 & \cdots & 0 & \cdots & 0 \\
        \vdots &   & \vdots &   & \vdots \\
         0 & \cdots & 0 & \cdots & 0 \\
    \end{pmatrix}_{(p+2) \times K}
\]
$\boldsymbol{D}_{k, 1} = {\bf{1}}_k^T \otimes ({\bf{0}}_{p+1}, {\bm{I}}_{p+1})^T$ is a $( p+2) \times k(p+1)$ matrix, where ${\bf{1}}_k$ is a $k \times 1$ vector with all $1$'s, ${\bf{0}}_{p+1} \in \mathbb{R}^{p+1}$ is a $(p+1) \times 1$ zero vector,  $\bm{I}_{p+1}$ is the $(p+1) \times (p+1)$ identity matrix of dimension $p+1$.
That is

\[
\boldsymbol{D}_{k, 1} =
    \addtocounter{MaxMatrixCols}{20}
    \begin{pmatrix}
        0 &  0 & \cdots & 0 & \cdots \cdots  & 0 &  0 &  \cdots& 0  \\
        1 &  0 & \cdots & 0 & \cdots \cdots  & 1 &  0 &   \cdots& 0   \\
        0 &  1 & \cdots & 0 & \cdots \cdots  & 0 &  1 &   \cdots & 0   \\
        \vdots & \vdots &   & \vdots &   & \vdots &   \vdots  & &  \vdots  \\
        0 & 0 & \cdots  & 1 & \cdots \cdots  & 0 &  0 &   \cdots& 1    \\
    \end{pmatrix}_{( p+2) \times k(p+1)}
\]
$\boldsymbol{D}_{k, 2}$ is a $( p+2) \times (K-k)(p+1)$ zero matrix.

Define ${\mathbf{Z}}_{ k i}^{T}=\left({{1}}, \hat{U}_{k i},  {\bm{X}}_{i}^{T}\right)
\boldsymbol{T}_{k} \in \mathbb{R}^{1 \times(p+2) {K}}$.
With these reparameterizations, the combined quantile objective function (\ref{Objfunction-sum}) can be rewritten as

\[
Q_0({\bm{\theta}}) = \sum_{k=1}^{K} \sum_{i=1}^{n} \rho_{\tau_{k}}\left(Y_{i}-\mathbf{Z}_{k i}^{T} \boldsymbol{\theta}\right) .
\]

\subsection{Penalized Joint Spatial Quantile Estimators}
In order to detect the insignificant and the constant quantile slope coefficients, we propose to shrink the interquantile slope differences $\{\beta_{k,l} - \beta_{k-1,l}: k = 2,\cdots,K,l =0, 1, 2, \cdots,p\}$ towards zero simultaneously, resulting in a simpler model structure and inducing the smoothness across quantiles.

\subsubsection{Penalised Fused Adaptive Lasso Estimator }
In this section, we first present the penalized spatial quantile estimator 
to shrink interquantile slope differences towards zero. The estimator could be obtained by minimizing the following objective function
\begin{equation}\label{Objfunc_p1}
    Q_1(\bm{\theta}) = \sum_{k=1}^{K} \sum_{i=1}^{n} \rho_{\tau_{k}}\left(Y_{i}-\mathbf{Z}_{k i}^{T} {\bm{\theta}}\right) + \tilde{\gamma}_{1n} \sum_{k=2}^{K} \sum_{l=1}^{p}\tilde{\omega}_{k, l} |d_{k,l}|,
\end{equation}
where $\tilde{\gamma}_{1n} \geq 0$ is a tuning parameter controlling the degree of penalization. The adaptive weight for $d_{k, l} = \beta_{k,l} - \beta_{k-1,l}$ is $\tilde{\omega}_{k, l}$ and we set
$\tilde{\omega}_{k, l} = \left|\tilde{d}_{k, l}\right|^{-1} = |\tilde{\beta}_{k,l} - \tilde{\beta}_{k-1,l}|^{-1}$, $k = 2, 3, \cdots, K$, $l = 0, 1, 2, \cdots, p$. Notice that $\tilde{\lambda}_k$ and $\tilde{\bm{\beta}}_k$ are the estimators obtained by minimizing $R_{IV}^1(\tau_k, \alpha_k, \lambda_k, {\bm{\beta}}_k)$ in equation (\ref{Objfunction-0}).
Write ${\bm{\beta}}_{(l)} = (\beta_{1,l}, \beta_{2, l},\cdots, \beta_{K, l})^T$.
If $d_{k, l}$ is shrunk to $0$, this implies the parameter ${\bm{\beta}}_{(l)}$ remain the same for the $(k-1)$th and the $k$th quantile levels. By employing this penalty, we can identify the quantile regions where each ${\bm{\beta}}_{(l)}$ varies or remains unvarying.

To establish asymptotic properites of penalised Lasso estimator, we give two more assumptions.

\begin{itemize}

    \item [(A10)] $\max _{1 \leq i \leq n}\left\|{\bm{X}}_{i}\right\|=o\left(n^{1 / 2}\right)$.
    \item [(A11)] For $1 \leq k \leq K$, there exist some positive definite matrices ${\bm\Gamma}_{k}$ and ${\bm\Omega}_{k}$ such that $\lim _{n \rightarrow \infty} n^{-1} \sum_{i=1}^{n} \mathbf{Z}_{i k} \mathbf{Z}_{i k}^{T}={\bm\Gamma}_{k}$ and $\lim _{n \rightarrow \infty} n^{-1} \sum_{i=1}^{n} f_{i}\{Q_{\tau_{k}}({Y_i}|\mathcal{F}_{-i}, {\bm{X}}_{i})
        \} \mathbf{Z}_{i k} \mathbf{Z}_{i k}^{T} = {\bm\Omega}_{k}$.
\end{itemize}

Before theoretical properties, we first denote ${\bm{\theta}}_{0}=(\theta_{j, 0}, j=1, \cdots, (p+2) K)$ as the true value of $\bm \theta$. Let the index sets $\mathcal{A}_{1}=\{1, \cdots, K\}$,
$\mathcal{A}_{2}= \{j:\theta_{j, 0} \neq 0, j=K+1, \cdots, (p+2) K\},$ and $\mathcal{A}=\mathcal{A}_{1} \cup \mathcal{A}_{2} .$ We write ${\bm \theta}_{\mathcal{A}}=\left(\theta_{j}: j \in \mathcal{A}\right)^{T},$ and its truth as ${\bm \theta}_{\mathcal{A}, 0}=\left(\theta_{j, 0}: j \in \mathcal{A}\right)^{T}$.

Without loss of generality, we assume that the quantile slopes ${\bm{\beta}}_{(l)}$ vary for the first $s_l$ ($s_l<K$ )quantiles, but remain constant for the remaining $(K-s_l)$ quantile levels.
Suppose the model structure is known, the oracle estimator $\hat{\bm \theta}_{\mathcal{A}} \in {\mathbb{R}}^{K+\sum_{l=0}^ps_l}$ can be obtained by
\[
\hat{\boldsymbol{\theta}}_{\mathcal{A}}=\arg \min _{\boldsymbol{\theta}_{\mathcal{A}}} \sum_{k=1}^{K} \sum_{i=1}^{n} \rho_{\tau_{k}}\left(Y_{i}-{\bm Z}_{i k, \mathcal{A}}^{T} \boldsymbol{\theta}_{\mathcal{A}}\right),
\]
where ${\bm Z}_{i k, A} \in \mathbb{R}^{{K+\sum_{l=0}^ps_l}}$ contains the first ${K+\sum_{l=0}^ps_l}$ elements of ${\bm Z}_{i k}$. The properties of the oracle and fused adaptive LASSO estimators are as following.

\begin{prop}
 Under conditions $(\mathrm{A} 1)-(\mathrm{A} 3),$ we have
\[
n^{1 / 2}\left(\hat{\boldsymbol{\theta}}_{\mathcal{A}}-\boldsymbol{\theta}_{\mathcal{A}, 0}\right) \stackrel{d}{\longrightarrow} N\left({\mathbf{0}}, {\bm{\Sigma}}_{\mathcal{A}}\right), \text { as } n \rightarrow \infty
\]
where ${\bm{\Sigma}}_{\mathcal{A}}=\left(\sum_{k=1}^{K} {\bm{\Omega}}_{k, \mathcal{A}}\right)^{-1}\left\{\sum_{k=1}^{K} \tau_{k}\left(1-\tau_{k}\right) {\bm{\Gamma}}_{k, \mathcal{A}}\right\}\left(\sum_{k=1}^{K} {\bm{\Omega}}_{k, \mathcal{A}}\right)^{-1}, {\bm{\Omega}}_{k, \mathcal{A}}$ and ${\bm{\Gamma}}_{k, \mathcal{A}}$ are the top-left $({K+\sum_{l=0}^ps_l}) \times({K+\sum_{l=0}^ps_l})$ submatrices of ${\bm{\Omega}}_{k}$ and ${\bm{\Gamma}}_{k},$ respectively.
\end{prop}

However, in practice, the true model structure is usually unknown beforehand. Hence we take into account the full parameter vector $\bm \theta$ and estimate it as $\hat{\bm{\theta}}_{\mathrm{FAL}}=\arg \min _{\bm{\theta}} Q_1({\bm{\theta}})$, where $Q_1({\bm{\theta}})$ is defined in (\ref{Objfunc_p1}). We show that $\hat{\bm{\theta}}_{\mathrm{FAL}}$ has the following oracle property.

\begin{theo}
Suppose that conditions (A1)-(A11) hold. If $n^{1 / 2} \tilde{\gamma}_{1n} \rightarrow 0$ and $n \tilde{\gamma}_{1n} \rightarrow \infty$ as $n \rightarrow \infty,$ we have
\begin{itemize}
\item[1.] {\textbf{Sparsity:}}
$$\operatorname{Pr}\left(\left\{j: \hat{\theta}_{j, \mathrm{FAL}} \neq 0, j=K+1, \cdots, (p+2) K\right\}=\mathcal{A}_{2}\right) \rightarrow 1.$$
\item[2.] {\textbf{Asymptotic normality:}}
$$n^{1 / 2}\left(\hat{\bm\theta}_{\mathcal{A}, \mathrm{FAL}}-{\bm\theta}_{\mathcal{A}, 0}\right) \stackrel{d}{\longrightarrow} N\left({\bm{0}}, {\bm\Sigma}_{\mathcal{A}}\right),$$
where ${\bm\Sigma}_{\mathcal{A}}$ is the covariance matrix of the oracle estimator given in Proposition 1.
\end{itemize}
\end{theo}

\subsubsection{Penalised Fused Adaptive Sup-norm Estimator }
In fused adaptive Lasso estimation, the slope coeffients and the interquantile slope differences are penalized individually.
In this section, we first present another penalty to shrink the different quantile levels associated with each predictor as one group.
Denote ${\bm{d}}_{(-2)} = \left(\beta_{1,0}, \beta_{1,1}, \cdots, \beta_{1, p}\right)^{T} \in \mathbb{R}^{p+1}$ as the slope coefficients at $\tau_{1}$, ${\bm{ d}}_{(-1)}=\left(\alpha_{1}, \cdots, \alpha_{K}\right)^{T} \in {\mathbb{R}}^{K}$ as the vector of the intercept terms, ${\bm{d}}_{(0)} = (\lambda_{2} - \lambda_{1}, \lambda_{3} - \lambda_{2}, \cdots, \lambda_{K} - \lambda_{K-1})^T$ as a vector of interquantile slope differences corresponding to the spatial lag terms,
${\bm{d}}_{(l)} = (\beta_{2,l} - \beta_{1,l}, \beta_{3,l} - \beta_{2,l}, \cdots, \beta_{K,l} - \beta_{K-1, l})^T$, $l = 1, 2, \cdots, p$, as a vector of interquantile slope differences corresponding to the $l$the predictor. The new parameter vector $\bm{\theta}$ could be recorded as $\bm{\theta}=\left({\bm d}_{(-2)}^{T}, {\bm d}_{(-1)}^{T}, {\bm d}_{(0)}^{T}, \cdots, {\bm d}_{(p)}^{T}\right)^{T}$ and update the covariates vector $\mathbf{Z}_{i k}$ with the order of elements of the new parameter vector $\bm \theta$.

The penalised estimator could be obtained by minimizing the following objective function
\begin{equation}\label{Objfunc_p2}
    Q_2(\bm{\theta})=\sum_{k=1}^{K} \sum_{i=1}^{n} \rho_{\tau_{k}}\left(Y_{i}-\mathbf{Z}_{k i}^{T} {\bm{\theta}}\right) + \tilde{\gamma}_{2n} \sum_{l=0}^{p}  \tilde{\omega}_{(l)} \| {\bm{d}}_{(l)}\|_{\infty} ,
\end{equation}
where $\tilde{\omega}_{(l)} = (\|{\tilde{\bm{d}}}_{(l)}\|_{\infty})^{-1} = (\max_{k}|\beta_{k,l} - \beta_{k-1,l}|)^{-1}$, $l = 0, 1, 2, \cdots, p$ are the group-wise adaptive weights. ${\tilde{\bm{d}}}_{(l)}$ are the initial estimators calculated from the quantile regression method with instrument variables by minimizing equation (\ref{Objfunction-0}). $\tilde{\gamma}_{2n} > 0$ is the tuning parameter controls the degree of the group-wise penalization on the interquantile coefficients differences.

To derive the asymptotic property of the Fused Adaptive Sup-norm method, we define the index sets $\mathcal{B}_{1}=\{-2,-1\}$, $\mathcal{B}_{2}=\left\{l:\left\|{\bm d}_{(l)}\right\| \neq 0, l=0, 1, \cdots, p\right\}$, and $\mathcal{B}=\mathcal{B}_{1} \cup \mathcal{B}_{2}$. Assume $\bm{\theta}_{\mathcal{B}}=\left({\bm d}_{(l)}^{T}, l \in \mathcal{B}\right)^{T}$ is the nonnull subset of ${\bm{\theta}}$ and the true parameter vector ${\bm \theta}_{\mathcal{B}, {0}}=\left({\bm d}_{(l), 0}^{T}: l \in \mathcal{B}\right)^{T}$.
Without loss of generality, we assume $\left\|{\bm d}_{(l)}\right\|\neq 0$ for $l<g~(g \geq 0)$ and $\left\|{\bm d}_{(l)}\right\|=0$ for $l=g+1, \cdots, p$, that is, ${\bm{\beta}}_{(l)} $ vary across quantiles for $l<g$, the remaining ${\bm{\beta}}_{(l)}$ are constant for $l\geq g$.

\begin{prop}
 Let $\hat{\bm {\theta}}_{\mathcal{B}}$ be the oracle estimator of ${\bm {\theta}}_{\mathcal{B}, 0}$ obtained by knowing the true structure. Assuming that conditions (A1)-(A11) hold, we have
\[
n^{1 / 2}\left(\hat{\bm{\theta}}_{\mathcal{B}}-\bm{\theta}_{\mathcal{B}, 0}\right) \stackrel{d}{\longrightarrow} N\left({\bm{0}}, \bm{\Sigma}_{\mathcal{B}}\right), \text { as } n \rightarrow \infty
\]
where ${\bm{\Sigma}}_{\mathcal{B}}=\left(\sum_{k=1}^{K} {\bm{\Omega}}_{k, \mathcal{B}}\right)^{-1}\left\{\sum_{k=1}^{K} \tau_{k}\left(1-\tau_{k}\right) {\bm{\Gamma}}_{k, \mathcal{B}}\right\}\left(\sum_{k=1}^{K} {\bm{\Omega}}_{k, \mathcal{B}}\right)^{-1}, {\bm{\Omega}}_{k, \mathcal{B}}$ and ${\bm{\Gamma}}_{k, \mathcal{B}}$ are the top-left $m \times m$ submatrices of ${\bm{\Omega}}_{k}$ and ${\bm{\Gamma}}_{k},$ respectively, where $m=p+1 + K+g(k-1)$.
\end{prop}
Theorem 2.3 shows that when the true model structure is unknown, the fused adaptive sup-norm penalized estimator of $\bm \theta$ has the following oracle property. The estimator could be obtain by minimizing the objective funciton (\ref{Objfunc_p2}), $\hat{\bm \theta}_{\mathrm{FAS}} = \arg \min _{\bm{\theta}} Q_2({\bm{\theta}})$.

\begin{theorem}
Suppose that conditions (A1)-(A11) hold. If $n^{1 / 2} \tilde{\gamma}_{2n} \rightarrow 0$ and $n \tilde{\gamma}_{2n} \rightarrow \infty$ as $n \rightarrow \infty,$ we have

\begin{itemize}
\item[1.] \textbf{Sparsity:} $$\operatorname{Pr}\left(\left\{l:\left\|\hat{\bm{d}}_{(l), \mathrm{FAS} }\right\| \neq \boldsymbol{0}, l=0, 1, \cdots, p\right\}=\mathcal{B}_{2}\right) \rightarrow 1$$
\item[2.] \textbf{Asymptotic normality:}
$$n^{1 / 2}\left(\hat{\bm{\theta}}_{\mathcal{B}, \mathrm{FAS}}-{\bm{\theta}}_{\mathcal{B}}\right) \stackrel{d}{\longrightarrow} N\left(0, {\bm{\Sigma}}_{\mathcal{B}}\right),$$
where ${\bm{\Sigma}}_{\mathcal{B}}$ is the covariance matrix of the oracle estimator given in Proposition 2.3.

\end{itemize}

\end{theorem}

\subsection{Estimation of the variance of the noise}

Similar to section \ref{SAR}, we give the estimator of the variance of the noise $\hat{\sigma}_k^2$,
\begin{align}
\hat{\sigma}^2_k
= \left[{\bm{Y}} - {\bm{A}}({\lambda}_k)^{-1}(\alpha_k + {\bm{X}}{\bm{\beta}}_k)\right]^{T}\left[{\bm{A}}(\lambda)^{-1}{\bm{A}}(\lambda)^{T}\right]^{-1}\left[{\bm{Y}} - {\bm{A}}({\lambda}_k)^{-1}(\alpha_k + {\bm{X}}{\bm{\beta}}_k)\right],
\label{eq_sigma_2}
\end{align}
where ${\bm{A}}(\lambda_k)= {\bm{I}}_n - \lambda_k {\bm{W}}$ is related to $\lambda_k$, $\lambda_k$ and ${\bm{\beta}}_k$ could be  estimated by the   solutions of (\ref{Objfunc_p1}) or (\ref{Objfunc_p2}).

It is noticed that ${\bm{A}}(\lambda_k) $ is a nonsingular matrix, then
$
\left[{\bm{A}}(\lambda_k) {\bm{A}}(\lambda_k) ^{T}\right]^{-1}
= ({\bm{I}}_n - \lambda _k {\bm{W}})^T  ({\bm{I}}_n - \lambda_k {\bm{W}}).
$
Let ${\bm{P}}_k = {\bm{A}}(\lambda_k)(\alpha_k + {\bm{X}} {\bm{\beta}}_k)$, then
${\bm{P}}_k$ could be calculated by solving a linear system,
\begin{align*}
{\bm{P}}_k
= ({\bm{I}}_n - \lambda_k  {\bm{W}})^{-1} (\alpha_k + {\bm{X}} {\bm{\beta}}_k)  .
\end{align*}
Thus $\hat{\sigma}^2_k$ defined by (\ref{eq_sigma_2}) can be calculated by
\begin{equation}\label{eq_sigma_2_2}
    \hat{\sigma}_k^2 = \frac{1}{n}\|({\bm{I}}_n - \lambda_k {\bm{W}})\cdot({\bm{Y}} - {\bm{P}}_k)\|_2^2.
\end{equation}

\section{Computation}

In our work, we focus on shrink the interquantile slope differences towards zero. The above minimization (\ref{Objfunc_p1}) and (\ref{Objfunc_p2}) are equivalent to linearly constrained minimization problems, which can be formulated as a linear programming problem with linear constraints.
Minimizing (\ref{Objfunc_p1}) is equivalent to solving
\begin{equation}\label{Linear-constrained1}
    \hat{\bm{\theta}} = \arg \min _{\bm{\theta}} \sum_{k=1}^{K} \sum_{i=1}^{n} \rho_{\tau_{k}}\left(Y_{i}-\bm{Z}_{k i}^{T} \bm{\theta}\right), ~~~ s.t.  \sum_{k=2}^{K} \sum_{l=0}^{p}\tilde{\omega}_{k, l} |d_{k,l}|  \leq t,
\end{equation}
where $t > 0$ is a tuning parameter that plays a similar role as $\tilde{\gamma}_{1n}$. Adopting this constrained minimization in (\ref{Linear-constrained1}) gives us a natural range of the tuning parameter, that is, $t \in [0, t_1^0]$, where $t_1^0 = (K-1)(p+1)$, see \cite{Jiang-2014}. Similarly, minimizing (\ref{Objfunc_p2}) is equivalent to solving
\begin{equation}\label{Linear-constrained2}
    \hat{\bm{\theta}} = \arg \min _{\bm{\theta}} \sum_{k=1}^{K} \sum_{i=1}^{n} \rho_{\tau_{k}}\left(Y_{i}-\bm{Z}_{k i}^{T} \bm{\theta}\right), ~~~ s.t.  \sum_{l=0}^{p}  \tilde{\omega}_{(l)} \| {\bm{d}}_{(l)}\|_{\infty} \leq t,
\end{equation}
where the tuning parameter $t \in [0, t_2^0]$ with $t_2^0 = p+1$, see \cite{Jiang-2014}.

The choose of tuning parameter $t$ is important for the solution to (\ref{Linear-constrained1}) and (\ref{Linear-constrained2}). We consider Akaike information criterion (AIC) 
to choose tuning parameter $t$, that is,
\[
\operatorname{AIC}(t)= \operatorname{Loss}(t) +\frac{1}{n} \operatorname{edf}(t),
\]
where $\operatorname{Loss}(t) = \sum_{k=1}^{K} \log \left[\sum_{i=1}^{n} \rho_{\tau_{k}}\left(Y_{i}-\bm{Z}_{k i}^{T} \hat{\bm{\theta}}(t)\right) \right]$ denotes the quantile logarithmic loss and measures the goodness of fit; see (\cite{Akaike}), \cite{Bondell}.
The the second term focuses on the complexity of model by degree of freedom $\operatorname{edf}(t)$ with a multiplier $1/n$. 
The effective degree of freedom is also associated with the tuning parameter $t$ or $\tilde{\gamma}_{1n}$ and $\tilde{\gamma}_{2n}$. Here we set $\operatorname{edf}(t)$ as the number of nonzero unique quantile slope coefficients over predictors in both FAL and FAS approaches, see \cite{Jiang-2013}, \cite{Zhao-R}.
Also, we could apply Bayesian Information Criterion (BIC) to choose the tuning parameter, see \cite{Schwarz}. Different from AIC, the multiplier on $\operatorname{edf}(t)$ of BIC is $\log(n)/2n$.
The detailed procedures are in the following.
\begin{itemize}
\item[] {\emph{Step 1.}} Given $\tau_k, k=1, 2, \cdots, K$, run ordinary quantile regression of ${\bm{U}}$ on ${\bm{V}}$ by minimizing equation (\ref{Objfunction-00}), and get the predictor $\hat{\bm{U}}_k$.
\item[] {\emph{Step 2.}} Take ${\bm{Z}}_{ k i}^{T}=\left(1, \hat{U}_{k i}, {\bm{X}}_{i}^{T}\right)\boldsymbol{T}_{k}$ as the covariates and $Y_i$ as the response, where $\hat{U}_{ki}$ is the $i$th element of $\hat{\bm{U}}_k$  and $\boldsymbol{T}_{k}$ is defined in Section \ref{SQARM}.
\item[] {\emph{Step 3.}} For each given $t $, estimator $\hat{\bm{\theta}}(t)$ could be obtained by solving (\ref{Linear-constrained1}) or (\ref{Linear-constrained2}).
\item[] {\emph{Step 4.}} Choose tuning parameter $t$ by AIC or BIC above, and denote $t^*$.
\item[] {\emph{Step 5.}} Solve (\ref{Linear-constrained1}) and (\ref{Linear-constrained2}) at $t = t^*$, then the FAL or FAS estimators could be obtained.
\end{itemize}

\section{ Simulation Study}

In this section, we conduct Monte Carlo simulations to evaluate the performance of our proposed methods for SQAR models.
In each example, the simulation is repeated 500 times with 9 quantile levels $\tau \in \mathbb{S}_{\tau}$, where $\mathbb{S}_{\tau}=\{0.1,0.2, \ldots, 0.9\}$. We compare the following approaches: the conventional quantile regression method with instrument variables (RQ), the Fused LASSO method without adaptive weights (FL), the Fused Adaptive LASSO (FAL) method, the Fused Sup-norm method without adaptive weights (FS), and the Fused Adaptive Sup-norm (FAS) method.
To evaluate various approaches, we examine the median of squared error (MedSE), that is the median of  $\|\hat{\bm{\theta}}_{(k)} - {\bm{\theta}}_{(k)} \|^2$ over 500 simulations, which has been used in \cite{Liu-C} and \cite{Liang-L}.

The data generating process is based on the model  (\ref{SQARmodel1}). Let the spatial weight matrix ${\bm {W}}_n = {\bm{I}}_{m_1} \otimes {\bm{B}}_{m_2}$, where ${\bm{B}}_{m_2} = (1/(m_2 - 1)) ({\bm{1}}_{m_2}{\bm{1}}_{m_2} - {\bm{I}}_{m_2})$, $\otimes$ is the
Kronecker product and ${\bm{1}}_{m_2}$ is an $m_2$-dimensional column vector of ones, see \cite{Case} and \cite{Lee}.

\begin{example}
The model in this example only has a univariate predictor. The data are generated from
\begin{equation}
\label{exam:1}
Y_i = \alpha(\tau_{n,i}) + \lambda(\tau_{n,i}) U_i + {{\beta}(\tau_{n,i})} {X}_{i} + e_i, ~~i = 1, \cdots, n,
\end{equation}
where $\tau_{n,i}$ is randomly from $\mathbb{S}_{\tau}$, $\alpha(\tau_{n,i}) = \alpha +  b F_n^{-1}(\tau_{n,i})$, $\lambda(\tau_{n,i}) = \lambda + c_0 F_n^{-1}(\tau_{n,i})$, $\beta(\tau_{n,i}) = \beta + c_1 F_n^{-1}(\tau_{n,i})$. $F_n^{-1}(\tau_{n,i})$ is the $\tau_{n,i}$ quantile of distribution $F_n$. Two distribution $F_n$ are considered as $N(0,1)$. $X_{i}$ are generated from $U(0, 1)$.
The sample size $n$ is chosen as $n = m_1 \times m_2 = 80,~120, ~160$, where $m_2 = 4$ and $m_1 =  20,~30,~40$. $\lambda$ is chosen as $\lambda = 0.2,~0.5,~0.8$, implying the low autocorrelation, medium autocorrelation and high autocorrelation respectively. $\alpha$ is chosen as $\alpha = 3$, and $\beta$ is chosen as $\beta = 3$.
Four settings of regression coefficients are considered as
\begin{itemize}
\item[] {\textrm{I}:}  ~~$b = 0.5$, $c_0 = 0.1$ and $c_1 = 0.2$.  The data are from a heteroscedastic model, where all the coefficients vary across quantile
levels.
\item[] {\textrm{II}:} ~$b = 0.5$, $c_0 = 0$ and $c_1 = 0.2$.  The data are from a  heteroscedastic model, where intercept term $\alpha(\tau_{n,i}) $ and slope coefficients of predictor $\beta(\tau_{n,i})$ vary across quantile levels, the spatial lag parameter is a constant.
\item[] {\textrm{III}:} $b = 0.5$, $c_0 = 0.1$ and $c_1 = 0$.  The data are from a  heteroscedastic model, where intercept term $\alpha(\tau_{n,i})$ and spatial lag parameter $\lambda(\tau_{n,i})$ vary across quantile levels, the slope coefficient of predictor $\beta(\tau_{n,i})$ stays invariant for all quantiles.
\item[] {\textrm{IV}:} $b = 0.5$, $c_0 = 0$ and $c_1 =  0$.  The data are from a homoscedastic model, where only the intercept term $\alpha(\tau_{n,i}) $ varies across quantile levels, the spatial lag parameter $\lambda(\tau_{n,i})$ and slope coefficients of predictor $\beta(\tau_{n,i})$ stay invariant for all quantiles.
\end{itemize}
\end{example}

\begin{table}[htbp]
\tiny
\caption{\small{The MedSE of coefficients for four settings with $n=80$, where $F_n$ is $N(0,1)$ in Example 1.}}
\center
    \begin{tabular}{rlrrrrrrrrr} \hline \hline
          &       & \multicolumn{9}{c}{$\tau$ } \\
          &       & 0.1   & 0.2   & 0.3   & 0.4   & 0.5   & 0.6   & 0.7   & 0.8   & 0.9 \\ \hline
          &       & \multicolumn{9}{c}{I: $b = 0.5,~c_0 = 0.1,~c_1 = 0.2$} \\
          & RQ    & 0.7361  & 0.6044  & 0.4944  & 0.4103  & 0.3659  & 0.3810  & 0.4715  & 0.4479  & 0.6859  \\
          & FL    & 0.2756  & 0.2542  & 0.2548  & 0.2642  & 0.2571  & 0.2742  & 0.2832  & 0.2564  & 0.2720  \\
          & FAL   & 0.3082  & 0.2993  & 0.2655  & 0.2611  & 0.2556  & 0.2702  & 0.2778  & 0.2731  & 0.3177  \\
          & FS    & 0.3647  & 0.3356  & 0.3289  & 0.2931  & 0.2825  & 0.2622  & 0.2870  & 0.2950  & 0.3267  \\
          & FAS   & 0.3597  & 0.3283  & 0.3292  & 0.3003  & 0.2678  & 0.2707  & 0.2610  & 0.2767  & 0.3384  \\
          &       & \multicolumn{9}{c}{II: $b = 0.5,~c_0 = 0,~c_1 = 0.2$} \\
          & RQ    & 0.4863  & 0.3018  & 0.2704  & 0.2948  & 0.2433  & 0.2194  & 0.2599  & 0.3044  & 0.4875  \\
          & FL    & 0.1885  & 0.1625  & 0.1525  & 0.1613  & 0.1577  & 0.1516  & 0.1385  & 0.1528  & 0.1799  \\
          & FAL   & 0.1902  & 0.1552  & 0.1605  & 0.1557  & 0.1632  & 0.1484  & 0.1559  & 0.1815  & 0.2003  \\
          & FS    & 0.2509  & 0.2064  & 0.2001  & 0.1995  & 0.1959  & 0.1625  & 0.1783  & 0.1881  & 0.2706  \\
          & FAS   & 0.2645  & 0.2040  & 0.1886  & 0.1868  & 0.1774  & 0.1507  & 0.1630  & 0.1954  & 0.2367  \\
    \multicolumn{1}{l}{$\lambda = 0.2$} &       & \multicolumn{9}{c}{III: $b = 0.5,~c_0 = 0.1,~c_1 = 0$} \\
          & RQ    & 0.6082  & 0.5190  & 0.4522  & 0.3570  & 0.3298  & 0.3224  & 0.4073  & 0.4219  & 0.6336  \\
          & FL    & 0.2356  & 0.2410  & 0.2433  & 0.2404  & 0.2341  & 0.2361  & 0.2233  & 0.2246  & 0.2202  \\
          & FAL   & 0.2535  & 0.2382  & 0.2516  & 0.2496  & 0.2351  & 0.2299  & 0.2084  & 0.2224  & 0.2358  \\
          & FS    & 0.3287  & 0.3094  & 0.3015  & 0.2781  & 0.2522  & 0.2532  & 0.2412  & 0.2585  & 0.3374  \\
          & FAS   & 0.3338  & 0.3151  & 0.2991  & 0.2865  & 0.2684  & 0.2348  & 0.2412  & 0.2575  & 0.3193  \\
          &       & \multicolumn{9}{c}{IV: $b = 0.5,~c_0 = 0,~c_1 = 0$} \\
          & RQ    & 0.4604  & 0.2849  & 0.2614  & 0.2765  & 0.2221  & 0.2123  & 0.2599  & 0.2768  & 0.4917  \\
          & FL    & 0.1702  & 0.1568  & 0.1661  & 0.1520  & 0.1456  & 0.1311  & 0.1318  & 0.1251  & 0.1337  \\
          & FAL   & 0.1638  & 0.1553  & 0.1725  & 0.1491  & 0.1390  & 0.1316  & 0.1419  & 0.1272  & 0.1403  \\
          & FS    & 0.2113  & 0.1920  & 0.1980  & 0.1818  & 0.1717  & 0.1464  & 0.1600  & 0.1896  & 0.2529  \\
          & FAS   & 0.2014  & 0.1863  & 0.1967  & 0.1746  & 0.1630  & 0.1468  & 0.1401  & 0.1644  & 0.2035  \\
          &       &       &       &       &       &       &       &       &       &  \\
          &       & \multicolumn{9}{c}{I: $b = 0.5,~c_0 = 0.1,~c_1 = 0.2$} \\
          & RQ    & 0.8619  & 0.6545  & 0.6133  & 0.6308  & 0.5545  & 0.6513  & 0.6871  & 0.8136  & 1.2680  \\
          & FL    & 0.3197  & 0.3127  & 0.3020  & 0.3299  & 0.3288  & 0.3410  & 0.3485  & 0.3757  & 0.4075  \\
          & FAL   & 0.3056  & 0.2935  & 0.3002  & 0.3272  & 0.3325  & 0.3646  & 0.3793  & 0.3822  & 0.3930  \\
          & FS    & 0.3626  & 0.3233  & 0.3508  & 0.3703  & 0.3601  & 0.3523  & 0.3813  & 0.4055  & 0.4183  \\
          & FAS   & 0.3822  & 0.3284  & 0.3431  & 0.3740  & 0.3676  & 0.3388  & 0.4052  & 0.3930  & 0.4090  \\
          &       & \multicolumn{9}{c}{II: $b = 0.5,~c_0 = 0,~c_1 = 0.2$} \\
          & RQ    & 0.5216  & 0.3391  & 0.2811  & 0.2674  & 0.1991  & 0.2374  & 0.3083  & 0.3619  & 0.5391  \\
          & FL    & 0.1634  & 0.1378  & 0.1383  & 0.1276  & 0.1257  & 0.1463  & 0.1423  & 0.1600  & 0.1860  \\
          & FAL   & 0.1600  & 0.1467  & 0.1465  & 0.1417  & 0.1321  & 0.1530  & 0.1748  & 0.1659  & 0.1995  \\
          & FS    & 0.2198  & 0.2033  & 0.1665  & 0.1686  & 0.1426  & 0.1513  & 0.1499  & 0.1762  & 0.2564  \\
          & FAS   & 0.2373  & 0.1916  & 0.1661  & 0.1537  & 0.1343  & 0.1355  & 0.1567  & 0.1698  & 0.2459  \\
    \multicolumn{1}{l}{$\lambda = 0.5$} &       & \multicolumn{9}{c}{III: $b = 0.5,~c_0 = 0.1,~c_1 = 0$} \\
          & RQ    & 0.7031  & 0.6020  & 0.5903  & 0.5547  & 0.5424  & 0.5815  & 0.6076  & 0.7203  & 1.1287  \\
          & FL    & 0.2948  & 0.2698  & 0.2828  & 0.3058  & 0.3124  & 0.3176  & 0.3219  & 0.3063  & 0.3206  \\
          & FAL   & 0.2652  & 0.2494  & 0.2784  & 0.2868  & 0.3075  & 0.3070  & 0.3200  & 0.3009  & 0.3505  \\
          & FS    & 0.3624  & 0.2901  & 0.2986  & 0.3153  & 0.3089  & 0.3190  & 0.3979  & 0.3637  & 0.3692  \\
          & FAS   & 0.3450  & 0.3150  & 0.3225  & 0.3278  & 0.3305  & 0.3365  & 0.4314  & 0.3971  & 0.4202  \\
          &       & \multicolumn{9}{c}{IV: $b = 0.5,~c_0 = 0,~c_1 = 0$} \\
          & RQ    & 0.4436  & 0.3162  & 0.2644  & 0.2788  & 0.1857  & 0.2006  & 0.2682  & 0.3424  & 0.5146  \\
          & FL    & 0.1260  & 0.1298  & 0.1285  & 0.1286  & 0.1252  & 0.1160  & 0.1215  & 0.1188  & 0.1116  \\
          & FAL   & 0.1328  & 0.1361  & 0.1300  & 0.1278  & 0.1236  & 0.1119  & 0.1106  & 0.1018  & 0.1079  \\
          & FS    & 0.2063  & 0.1656  & 0.1694  & 0.1631  & 0.1238  & 0.1208  & 0.1283  & 0.1475  & 0.1619  \\
          & FAS   & 0.1849  & 0.1435  & 0.1666  & 0.1755  & 0.1341  & 0.1313  & 0.1276  & 0.1491  & 0.1681  \\
          &       &       &       &       &       &       &       &       &       &  \\
          &       & \multicolumn{9}{c}{I: $b = 0.5,~c_0 = 0.1,~c_1 = 0.2$} \\
          & RQ    & 2.7947  & 2.1128  & 2.9254  & 2.8017  & 2.6896  & 2.9395  & 5.0333  & 4.7663  & 7.4942  \\
          & FL    & 0.7116  & 0.7434  & 0.8296  & 0.8624  & 0.8993  & 0.9016  & 0.9921  & 0.9747  & 0.9964  \\
          & FAL   & 0.7940  & 0.8002  & 0.9057  & 0.9735  & 1.0403  & 1.0152  & 1.1523  & 1.0489  & 1.3374  \\
          & FS    & 0.6810  & 0.6872  & 0.7777  & 0.8705  & 0.9891  & 1.1163  & 1.2030  & 1.1275  & 1.2637  \\
          & FAS   & 0.9139  & 0.8063  & 1.0421  & 1.1484  & 1.2156  & 1.2924  & 1.6199  & 1.5367  & 2.1910  \\
          &       & \multicolumn{9}{c}{II: $b = 0.5,~c_0 = 0,~c_1 = 0.2$} \\
          & RQ    & 0.8598  & 0.7497  & 0.5430  & 0.4917  & 0.4068  & 0.3961  & 0.4861  & 0.6886  & 0.7544  \\
          & FL    & 0.1753  & 0.1675  & 0.1722  & 0.1753  & 0.1659  & 0.1627  & 0.1768  & 0.1804  & 0.2270  \\
          & FAL   & 0.1885  & 0.1684  & 0.1818  & 0.1754  & 0.1690  & 0.1612  & 0.1693  & 0.1889  & 0.2147  \\
          & FS    & 0.2630  & 0.1835  & 0.2023  & 0.1778  & 0.1905  & 0.1778  & 0.1899  & 0.2168  & 0.2278  \\
          & FAS   & 0.2948  & 0.2303  & 0.2243  & 0.1874  & 0.2018  & 0.1957  & 0.1873  & 0.2209  & 0.2587  \\
    \multicolumn{1}{l}{$\lambda = 0.8$} &       & \multicolumn{9}{c}{III: $b = 0.5,~c_0 = 0.1,~c_1 = 0$} \\
          & RQ    & 2.6980  & 1.8952  & 2.6252  & 2.2426  & 2.2561  & 2.9997  & 4.4895  & 4.5375  & 7.4040  \\
          & FL    & 0.6922  & 0.6815  & 0.6857  & 0.7192  & 0.7606  & 0.8405  & 0.9076  & 0.8970  & 0.9606  \\
          & FAL   & 0.7531  & 0.7480  & 0.8717  & 0.8814  & 0.9219  & 0.9079  & 1.0385  & 0.9878  & 1.3273  \\
          & FS    & 0.7161  & 0.6420  & 0.6382  & 0.8123  & 0.9543  & 1.0082  & 0.9729  & 0.9486  & 1.1019  \\
          & FAS   & 1.0054  & 0.7510  & 0.9247  & 0.9981  & 1.0529  & 1.2140  & 1.4839  & 1.3804  & 1.9347  \\
          &       & \multicolumn{9}{c}{IV: $b = 0.5,~c_0 = 0,~c_1 = 0$} \\
          & RQ    & 0.7631  & 0.5842  & 0.5760  & 0.4094  & 0.3876  & 0.3912  & 0.4474  & 0.5601  & 0.7211  \\
          & FL    & 0.1413  & 0.1615  & 0.1617  & 0.1731  & 0.1525  & 0.1486  & 0.1377  & 0.1320  & 0.1351  \\
          & FAL   & 0.1609  & 0.1509  & 0.1545  & 0.1463  & 0.1292  & 0.1396  & 0.1279  & 0.1377  & 0.1470  \\
          & FS    & 0.1937  & 0.2005  & 0.1762  & 0.1794  & 0.1474  & 0.1486  & 0.1481  & 0.1439  & 0.1607  \\
          & FAS   & 0.2411  & 0.2293  & 0.2084  & 0.2139  & 0.1892  & 0.1539  & 0.1452  & 0.1638  & 0.1612  \\ \hline
    \end{tabular}%
\label{tab:1}
\end{table}%

\begin{table}[htbp]
\tiny
\caption{\small{The MedSE of coefficients for four settings with $n=120$, where $F_n$ is $N(0,1)$ in Example 1.}}
\center
    \begin{tabular}{rlrrrrrrrrr} \hline \hline
          &       & \multicolumn{9}{c}{$\tau$ } \\
          &       & 0.1   & 0.2   & 0.3   & 0.4   & 0.5   & 0.6   & 0.7   & 0.8   & 0.9 \\ \hline
          &       & \multicolumn{9}{c}{I: $b = 0.5,~c_0 = 0.1,~c_1 = 0.2$} \\
          & RQ    & 0.3559  & 0.3513  & 0.2660  & 0.2685  & 0.2158  & 0.2848  & 0.3249  & 0.3858  & 0.4761  \\
          & FL    & 0.1843  & 0.1782  & 0.1672  & 0.1642  & 0.1666  & 0.1802  & 0.1771  & 0.1977  & 0.2607  \\
          & FAL   & 0.1967  & 0.1787  & 0.1711  & 0.1678  & 0.1617  & 0.1732  & 0.1786  & 0.1872  & 0.2421  \\
          & FS    & 0.2464  & 0.2085  & 0.1920  & 0.1888  & 0.1849  & 0.2163  & 0.2268  & 0.2576  & 0.2756  \\
          & FAS   & 0.2332  & 0.2090  & 0.1842  & 0.1855  & 0.1814  & 0.2087  & 0.2133  & 0.2416  & 0.2691  \\
          &       & \multicolumn{9}{c}{II: $b = 0.5,~c_0 = 0,~c_1 = 0.2$} \\
          & RQ    & 0.3533  & 0.2052  & 0.2014  & 0.1770  & 0.1706  & 0.1809  & 0.1973  & 0.1951  & 0.3334  \\
          & FL    & 0.1335  & 0.1184  & 0.1062  & 0.1274  & 0.1304  & 0.1198  & 0.1364  & 0.1366  & 0.1580  \\
          & FAL   & 0.1437  & 0.1108  & 0.1035  & 0.1274  & 0.1252  & 0.1178  & 0.1293  & 0.1361  & 0.1732  \\
          & FS    & 0.2327  & 0.1541  & 0.1356  & 0.1451  & 0.1404  & 0.1525  & 0.1479  & 0.1609  & 0.1879  \\
          & FAS   & 0.2098  & 0.1421  & 0.1287  & 0.1271  & 0.1321  & 0.1339  & 0.1267  & 0.1676  & 0.1852  \\
    \multicolumn{1}{l}{$\lambda = 0.2$} &       & \multicolumn{9}{c}{III: $b = 0.5,~c_0 = 0.1,~c_1 = 0$} \\
          & RQ    & 0.3350  & 0.2763  & 0.2539  & 0.2537  & 0.1949  & 0.2552  & 0.3135  & 0.3532  & 0.4313  \\
          & FL    & 0.1985  & 0.1680  & 0.1590  & 0.1593  & 0.1540  & 0.1562  & 0.1498  & 0.1703  & 0.1812  \\
          & FAL   & 0.1865  & 0.1712  & 0.1598  & 0.1710  & 0.1548  & 0.1512  & 0.1533  & 0.1737  & 0.2025  \\
          & FS    & 0.2329  & 0.1950  & 0.1785  & 0.1763  & 0.1631  & 0.1755  & 0.1875  & 0.2222  & 0.2546  \\
          & FAS   & 0.2114  & 0.1872  & 0.1728  & 0.1719  & 0.1588  & 0.1691  & 0.1825  & 0.1982  & 0.2368  \\
          &       & \multicolumn{9}{c}{IV: $b = 0.5,~c_0 = 0,~c_1 = 0$} \\
          & RQ    & 0.3119  & 0.1875  & 0.1676  & 0.1673  & 0.1464  & 0.1622  & 0.1683  & 0.1868  & 0.2931  \\
          & FL    & 0.1114  & 0.1036  & 0.1083  & 0.1129  & 0.1064  & 0.1049  & 0.1043  & 0.1054  & 0.1077  \\
          & FAL   & 0.1254  & 0.1064  & 0.1041  & 0.1063  & 0.1044  & 0.1033  & 0.1018  & 0.0977  & 0.1061  \\
          & FS    & 0.1921  & 0.1490  & 0.1384  & 0.1419  & 0.1351  & 0.1275  & 0.1207  & 0.1485  & 0.1499  \\
          & FAS   & 0.1775  & 0.1372  & 0.1309  & 0.1293  & 0.1244  & 0.1192  & 0.1137  & 0.1271  & 0.1504  \\
          &       &       &       &       &       &       &       &       &       &  \\
          &       & \multicolumn{9}{c}{I: $b = 0.5,~c_0 = 0.1,~c_1 = 0.2$} \\
          & RQ    & 0.5240  & 0.3687  & 0.3572  & 0.4004  & 0.3899  & 0.3904  & 0.4635  & 0.5392  & 0.7573  \\
          & FL    & 0.1722  & 0.1889  & 0.1659  & 0.1851  & 0.1926  & 0.2211  & 0.2342  & 0.2820  & 0.3321  \\
          & FAL   & 0.1893  & 0.1705  & 0.1711  & 0.1842  & 0.1871  & 0.2184  & 0.2305  & 0.2576  & 0.3221  \\
          & FS    & 0.2243  & 0.2286  & 0.2279  & 0.2273  & 0.2312  & 0.2485  & 0.2510  & 0.2818  & 0.3575  \\
          & FAS   & 0.2317  & 0.2246  & 0.2263  & 0.2281  & 0.2164  & 0.2497  & 0.2481  & 0.3052  & 0.3800  \\
          &       & \multicolumn{9}{c}{II: $b = 0.5,~c_0 = 0,~c_1 = 0.2$} \\
          & RQ    & 0.3315  & 0.2482  & 0.1902  & 0.1945  & 0.1700  & 0.1948  & 0.2191  & 0.2320  & 0.3915  \\
          & FL    & 0.1292  & 0.1097  & 0.0990  & 0.1151  & 0.1045  & 0.0978  & 0.1096  & 0.1164  & 0.1495  \\
          & FAL   & 0.1239  & 0.1134  & 0.1115  & 0.1120  & 0.1063  & 0.0999  & 0.1071  & 0.1210  & 0.1477  \\
          & FS    & 0.1698  & 0.1374  & 0.1096  & 0.1249  & 0.1290  & 0.1279  & 0.1408  & 0.1571  & 0.1871  \\
          & FAS   & 0.1699  & 0.1407  & 0.1135  & 0.1217  & 0.1250  & 0.1333  & 0.1384  & 0.1465  & 0.1838  \\
    \multicolumn{1}{l}{$\lambda = 0.5$} &       & \multicolumn{9}{c}{III: $b = 0.5,~c_0 = 0.1,~c_1 = 0$} \\
          & RQ    & 0.3315  & 0.2482  & 0.1902  & 0.1945  & 0.1700  & 0.1948  & 0.2191  & 0.2320  & 0.3915  \\
          & FL    & 0.1292  & 0.1097  & 0.0990  & 0.1151  & 0.1045  & 0.0978  & 0.1096  & 0.1164  & 0.1495  \\
          & FAL   & 0.1239  & 0.1134  & 0.1115  & 0.1120  & 0.1063  & 0.0999  & 0.1071  & 0.1210  & 0.1477  \\
          & FS    & 0.1698  & 0.1374  & 0.1096  & 0.1249  & 0.1290  & 0.1279  & 0.1408  & 0.1571  & 0.1871  \\
          & FAS   & 0.1699  & 0.1407  & 0.1135  & 0.1217  & 0.1250  & 0.1333  & 0.1384  & 0.1465  & 0.1838  \\
          &       & \multicolumn{9}{c}{IV: $b = 0.5,~c_0 = 0,~c_1 = 0$} \\
          & RQ    & 0.3433  & 0.2289  & 0.1612  & 0.1706  & 0.1620  & 0.1791  & 0.1514  & 0.2094  & 0.3569  \\
          & FL    & 0.1089  & 0.1053  & 0.1031  & 0.1016  & 0.0951  & 0.0903  & 0.0932  & 0.0880  & 0.0955  \\
          & FAL   & 0.1074  & 0.0952  & 0.0912  & 0.0901  & 0.0871  & 0.0876  & 0.0823  & 0.0853  & 0.0913  \\
          & FS    & 0.1463  & 0.1321  & 0.1218  & 0.1238  & 0.1291  & 0.1090  & 0.1139  & 0.1370  & 0.1634  \\
          & FAS   & 0.1457  & 0.1295  & 0.1200  & 0.1277  & 0.1187  & 0.1102  & 0.1039  & 0.1128  & 0.1380  \\
          &       &       &       &       &       &       &       &       &       &  \\
          &       & \multicolumn{9}{c}{I: $b = 0.5,~c_0 = 0.1,~c_1 = 0.2$} \\
          & RQ    & 1.5669  & 1.2887  & 1.5433  & 2.0641  & 1.6935  & 2.1308  & 2.2500  & 3.2445  & 5.9037  \\
          & FL    & 0.4692  & 0.4579  & 0.4548  & 0.5164  & 0.5354  & 0.5729  & 0.5716  & 0.6037  & 0.6327  \\
          & FAL   & 0.5101  & 0.4693  & 0.5164  & 0.4804  & 0.5447  & 0.6434  & 0.7339  & 0.8091  & 1.0014  \\
          & FS    & 0.4865  & 0.4321  & 0.5028  & 0.5512  & 0.5526  & 0.5393  & 0.5949  & 0.6180  & 0.7059  \\
          & FAS   & 0.7185  & 0.5167  & 0.6073  & 0.6802  & 0.7519  & 0.8259  & 0.8542  & 0.8941  & 1.3331  \\
          &       & \multicolumn{9}{c}{II: $b = 0.5,~c_0 = 0,~c_1 = 0.2$} \\
          & RQ    & 0.6257  & 0.4470  & 0.3910  & 0.3233  & 0.3265  & 0.3618  & 0.3657  & 0.4479  & 0.5828  \\
          & FL    & 0.1406  & 0.1187  & 0.1095  & 0.1188  & 0.1087  & 0.1208  & 0.1145  & 0.1191  & 0.1525  \\
          & FAL   & 0.1440  & 0.1228  & 0.1184  & 0.1160  & 0.1027  & 0.1030  & 0.1027  & 0.1184  & 0.1604  \\
          & FS    & 0.1870  & 0.1806  & 0.1540  & 0.1364  & 0.1248  & 0.1224  & 0.1358  & 0.1453  & 0.1754  \\
          & FAS   & 0.1788  & 0.1878  & 0.1607  & 0.1494  & 0.1404  & 0.1639  & 0.1647  & 0.1623  & 0.2275  \\
    \multicolumn{1}{l}{$\lambda = 0.8$} &       & \multicolumn{9}{c}{III: $b = 0.5,~c_0 = 0.1,~c_1 = 0$} \\
          & RQ    & 1.6366  & 1.1917  & 1.4157  & 1.7806  & 1.8727  & 1.7215  & 2.2950  & 3.1020  & 5.3243  \\
          & FL    & 0.4367  & 0.3893  & 0.4051  & 0.4420  & 0.4739  & 0.5236  & 0.4881  & 0.5352  & 0.5701  \\
          & FAL   & 0.3723  & 0.3819  & 0.4048  & 0.4543  & 0.4851  & 0.5592  & 0.6055  & 0.6587  & 0.8423  \\
          & FS    & 0.4462  & 0.4234  & 0.4423  & 0.4898  & 0.5182  & 0.5163  & 0.5933  & 0.6462  & 0.8291  \\
          & FAS   & 0.6126  & 0.5321  & 0.6437  & 0.6857  & 0.7725  & 0.7167  & 0.8832  & 0.8884  & 1.3151  \\
          &       & \multicolumn{9}{c}{IV: $b = 0.5,~c_0 = 0,~c_1 = 0$} \\
          & RQ    & 0.5818  & 0.4403  & 0.3215  & 0.3192  & 0.3845  & 0.4038  & 0.3432  & 0.4179  & 0.5668  \\
          & FL    & 0.1227  & 0.1202  & 0.1136  & 0.1172  & 0.1126  & 0.1268  & 0.0984  & 0.1042  & 0.1371  \\
          & FAL   & 0.1345  & 0.1229  & 0.1190  & 0.1201  & 0.1190  & 0.1145  & 0.1119  & 0.1043  & 0.1232  \\
          & FS    & 0.2039  & 0.1619  & 0.1332  & 0.1567  & 0.1457  & 0.1466  & 0.1315  & 0.1428  & 0.1944  \\
          & FAS   & 0.1964  & 0.1723  & 0.1472  & 0.1477  & 0.1520  & 0.1409  & 0.1577  & 0.1531  & 0.2016  \\ \hline
    \end{tabular}%
\label{tab:2}
\end{table}%

\begin{table}[htbp]
\tiny
\caption{\small{The MedSE of coefficients for four settings with $n=160$, where $F_n$ is $N(0,1)$ in Example 1.}}
\center
    \begin{tabular}{rlrrrrrrrrr} \hline \hline
          &       & \multicolumn{9}{c}{$\tau$ } \\
          &       & 0.1   & 0.2   & 0.3   & 0.4   & 0.5   & 0.6   & 0.7   & 0.8   & 0.9 \\ \hline
          &       & \multicolumn{9}{c}{I: $b = 0.5,~c_0 = 0.1,~c_1 = 0.2$} \\
          & RQ    & 0.3437  & 0.2256  & 0.2254  & 0.2385  & 0.1975  & 0.1899  & 0.2270  & 0.2782  & 0.3245  \\
          & FL    & 0.1485  & 0.1281  & 0.1327  & 0.1272  & 0.1364  & 0.1458  & 0.1706  & 0.1746  & 0.2114  \\
          & FAL   & 0.1565  & 0.1333  & 0.1387  & 0.1278  & 0.1326  & 0.1503  & 0.1633  & 0.1740  & 0.1891  \\
          & FS    & 0.1992  & 0.1755  & 0.1734  & 0.1562  & 0.1471  & 0.1532  & 0.1735  & 0.1889  & 0.2498  \\
          & FAS   & 0.1933  & 0.1710  & 0.1622  & 0.1401  & 0.1324  & 0.1523  & 0.1741  & 0.1798  & 0.2157  \\
          &       & \multicolumn{9}{c}{II: $b = 0.5,~c_0 = 0,~c_1 = 0.2$} \\
          & RQ    & 0.1972  & 0.1490  & 0.1259  & 0.1390  & 0.1386  & 0.1400  & 0.1325  & 0.1638  & 0.2194  \\
          & FL    & 0.1102  & 0.0907  & 0.0902  & 0.0893  & 0.0898  & 0.0939  & 0.1025  & 0.1195  & 0.1422  \\
          & FAL   & 0.1016  & 0.0845  & 0.0888  & 0.0842  & 0.0873  & 0.0930  & 0.0952  & 0.1147  & 0.1566  \\
          & FS    & 0.1508  & 0.1098  & 0.1040  & 0.1023  & 0.1077  & 0.1161  & 0.1017  & 0.1336  & 0.1645  \\
          & FAS   & 0.1145  & 0.1010  & 0.1031  & 0.1066  & 0.0977  & 0.1009  & 0.1075  & 0.1230  & 0.1563  \\
    \multicolumn{1}{l}{$\lambda = 0.2$} &       & \multicolumn{9}{c}{III: $b = 0.5,~c_0 = 0.1,~c_1 = 0$} \\
          & RQ    & 0.2865  & 0.2018  & 0.2026  & 0.2024  & 0.1758  & 0.1634  & 0.2064  & 0.2612  & 0.3205  \\
          & FL    & 0.1287  & 0.1203  & 0.1266  & 0.1172  & 0.1135  & 0.1220  & 0.1274  & 0.1327  & 0.1368  \\
          & FAL   & 0.1394  & 0.1326  & 0.1301  & 0.1168  & 0.1169  & 0.1218  & 0.1314  & 0.1504  & 0.1565  \\
          & FS    & 0.1724  & 0.1580  & 0.1524  & 0.1481  & 0.1399  & 0.1348  & 0.1496  & 0.1791  & 0.2182  \\
          & FAS   & 0.1601  & 0.1432  & 0.1307  & 0.1345  & 0.1305  & 0.1290  & 0.1590  & 0.1671  & 0.1804  \\
          &       & \multicolumn{9}{c}{IV: $b = 0.5,~c_0 = 0,~c_1 = 0$} \\
          & RQ    & 0.1925  & 0.1372  & 0.1247  & 0.1156  & 0.1185  & 0.1212  & 0.1156  & 0.1461  & 0.1867  \\
          & FL    & 0.0866  & 0.0764  & 0.0797  & 0.0830  & 0.0767  & 0.0796  & 0.0812  & 0.0841  & 0.0878  \\
          & FAL   & 0.0890  & 0.0820  & 0.0880  & 0.0836  & 0.0796  & 0.0806  & 0.0774  & 0.0808  & 0.0876  \\
          & FS    & 0.1157  & 0.1064  & 0.0972  & 0.0973  & 0.0914  & 0.1027  & 0.0906  & 0.1033  & 0.1302  \\
          & FAS   & 0.1003  & 0.0956  & 0.0954  & 0.0890  & 0.0871  & 0.0965  & 0.0865  & 0.0949  & 0.1173  \\
          &       &       &       &       &       &       &       &       &       &  \\
          &       & \multicolumn{9}{c}{I: $b = 0.5,~c_0 = 0.1,~c_1 = 0.2$} \\
          & RQ    & 0.4352  & 0.2862  & 0.2398  & 0.2818  & 0.3030  & 0.2907  & 0.3564  & 0.3932  & 0.6359  \\
          & FL    & 0.2031  & 0.1631  & 0.1549  & 0.1505  & 0.1808  & 0.1740  & 0.2127  & 0.2080  & 0.2757  \\
          & FAL   & 0.2180  & 0.1648  & 0.1599  & 0.1618  & 0.1817  & 0.1846  & 0.1991  & 0.2043  & 0.2581  \\
          & FS    & 0.2384  & 0.1962  & 0.1954  & 0.1632  & 0.1876  & 0.2101  & 0.2085  & 0.2414  & 0.3055  \\
          & FAS   & 0.2513  & 0.1959  & 0.1942  & 0.1764  & 0.1775  & 0.2047  & 0.2122  & 0.2683  & 0.3134  \\
          &       & \multicolumn{9}{c}{II: $b = 0.5,~c_0 = 0,~c_1 = 0.2$} \\
          & RQ    & 0.2368  & 0.1599  & 0.1364  & 0.1332  & 0.1534  & 0.1362  & 0.1397  & 0.1864  & 0.2586  \\
          & FL    & 0.1005  & 0.0894  & 0.0842  & 0.0866  & 0.0892  & 0.0858  & 0.0842  & 0.0948  & 0.1325  \\
          & FAL   & 0.0977  & 0.0753  & 0.0690  & 0.0852  & 0.0842  & 0.0804  & 0.0828  & 0.0911  & 0.1098  \\
          & FS    & 0.1275  & 0.1158  & 0.0987  & 0.0888  & 0.1006  & 0.1008  & 0.1075  & 0.1076  & 0.1293  \\
          & FAS   & 0.1155  & 0.0959  & 0.0859  & 0.0851  & 0.1004  & 0.0977  & 0.0910  & 0.0969  & 0.1247  \\
    \multicolumn{1}{l}{$\lambda = 0.5$} &       & \multicolumn{9}{c}{III: $b = 0.5,~c_0 = 0.1,~c_1 = 0$} \\
          & RQ    & 0.4004  & 0.2785  & 0.1963  & 0.2508  & 0.2587  & 0.2611  & 0.3359  & 0.3901  & 0.5623  \\
          & FL    & 0.1521  & 0.1527  & 0.1438  & 0.1424  & 0.1374  & 0.1513  & 0.1643  & 0.1609  & 0.1947  \\
          & FAL   & 0.1672  & 0.1516  & 0.1475  & 0.1608  & 0.1610  & 0.1551  & 0.1724  & 0.1859  & 0.2108  \\
          & FS    & 0.2210  & 0.1960  & 0.1694  & 0.1639  & 0.1759  & 0.1808  & 0.2082  & 0.2332  & 0.3037  \\
          & FAS   & 0.2110  & 0.1911  & 0.1730  & 0.1736  & 0.1660  & 0.1744  & 0.2113  & 0.2273  & 0.2946  \\
          &       & \multicolumn{9}{c}{IV: $b = 0.5,~c_0 = 0,~c_1 = 0$} \\
          & RQ    & 0.2159  & 0.1234  & 0.1371  & 0.1334  & 0.1378  & 0.1402  & 0.1153  & 0.1499  & 0.2254  \\
          & FL    & 0.0870  & 0.0872  & 0.0798  & 0.0853  & 0.0901  & 0.0764  & 0.0734  & 0.0724  & 0.0724  \\
          & FAL   & 0.0843  & 0.0810  & 0.0770  & 0.0849  & 0.0862  & 0.0837  & 0.0808  & 0.0783  & 0.0836  \\
          & FS    & 0.1129  & 0.0934  & 0.0909  & 0.0907  & 0.0934  & 0.0934  & 0.0814  & 0.1005  & 0.1218  \\
          & FAS   & 0.0982  & 0.0806  & 0.0804  & 0.0877  & 0.0906  & 0.0851  & 0.0791  & 0.0862  & 0.1001  \\
          &       &       &       &       &       &       &       &       &       &  \\
          &       & \multicolumn{9}{c}{I: $b = 0.5,~c_0 = 0.1,~c_1 = 0.2$} \\
          & RQ    & 1.1799  & 1.2027  & 1.0979  & 0.9894  & 1.4687  & 1.8547  & 1.8186  & 3.1203  & 4.5802  \\
          & FL    & 0.3876  & 0.3864  & 0.3701  & 0.3691  & 0.4234  & 0.4567  & 0.4797  & 0.5398  & 0.6090  \\
          & FAL   & 0.4029  & 0.3517  & 0.3300  & 0.3589  & 0.4064  & 0.4873  & 0.4982  & 0.5990  & 0.7033  \\
          & FS    & 0.4732  & 0.4102  & 0.3580  & 0.4630  & 0.4572  & 0.5216  & 0.5664  & 0.6272  & 0.7437  \\
          & FAS   & 0.6155  & 0.4834  & 0.4407  & 0.5220  & 0.5513  & 0.5772  & 0.7003  & 0.8905  & 1.2072  \\
          &       & \multicolumn{9}{c}{II: $b = 0.5,~c_0 = 0,~c_1 = 0.2$} \\
          & RQ    & 0.4043  & 0.3259  & 0.2624  & 0.2006  & 0.2581  & 0.2527  & 0.2921  & 0.3525  & 0.5165  \\
          & FL    & 0.1111  & 0.1071  & 0.1089  & 0.0884  & 0.0888  & 0.0924  & 0.0985  & 0.0909  & 0.1234  \\
          & FAL   & 0.1027  & 0.0878  & 0.0871  & 0.0870  & 0.0842  & 0.0710  & 0.0909  & 0.0849  & 0.1275  \\
          & FS    & 0.1002  & 0.1102  & 0.1060  & 0.1001  & 0.1057  & 0.1123  & 0.1158  & 0.1090  & 0.1593  \\
          & FAS   & 0.1297  & 0.1064  & 0.1113  & 0.1026  & 0.1086  & 0.1153  & 0.1112  & 0.1208  & 0.1650  \\
    \multicolumn{1}{l}{$\lambda = 0.8$} &       & \multicolumn{9}{c}{III: $b = 0.5,~c_0 = 0.1,~c_1 = 0$} \\
          & RQ    & 1.1699  & 1.2744  & 1.0082  & 0.9725  & 1.4371  & 1.6950  & 1.5545  & 3.0615  & 4.4454  \\
          & FL    & 0.3115  & 0.2944  & 0.2890  & 0.3566  & 0.3889  & 0.4022  & 0.4583  & 0.5093  & 0.5552  \\
          & FAL   & 0.2920  & 0.3035  & 0.2906  & 0.3184  & 0.3559  & 0.3909  & 0.4351  & 0.5183  & 0.6824  \\
          & FS    & 0.3617  & 0.3115  & 0.3572  & 0.4142  & 0.3957  & 0.4945  & 0.5454  & 0.6196  & 0.6528  \\
          & FAS   & 0.4991  & 0.4732  & 0.4463  & 0.4486  & 0.4314  & 0.5287  & 0.5921  & 0.8962  & 1.0627  \\
          &       & \multicolumn{9}{c}{IV: $b = 0.5,~c_0 = 0,~c_1 = 0$} \\
          & RQ    & 0.3704  & 0.2795  & 0.2318  & 0.1936  & 0.2401  & 0.2107  & 0.2850  & 0.3041  & 0.4709  \\
          & FL    & 0.0936  & 0.0876  & 0.0824  & 0.0857  & 0.0801  & 0.0719  & 0.0787  & 0.0726  & 0.0793  \\
          & FAL   & 0.0895  & 0.0812  & 0.0789  & 0.0794  & 0.0750  & 0.0734  & 0.0694  & 0.0584  & 0.0693  \\
          & FS    & 0.1178  & 0.1148  & 0.1109  & 0.0899  & 0.0902  & 0.0788  & 0.0919  & 0.1042  & 0.1324  \\
          & FAS   & 0.1266  & 0.1008  & 0.1055  & 0.0834  & 0.0963  & 0.0997  & 0.0979  & 0.1052  & 0.1416  \\ \\ \hline
    \end{tabular}%
\label{tab:3}
\end{table}%

Table \ref{tab:1}-\ref{tab:3} report the MedSE of coefficients for four settings where $F_n$ are considered as $N(0,1)$. Apparently, FL, FAL, FS and FAS methods have smaller MedSE than the conventional RQ method. Among them, FAL or FL performs best. Higher autocorrelation in the SQAR models brings larger MedSE. And with the increasing of sample size, the MedSE become smaller.
Observing the performance for four settings, when $\lambda$ and $\beta$ are invariant for all quantiles, the penalized estimation methods have the smallest MedSE; when $\lambda$ varies across quantile levels, the MedSE were little affected, no matter whether the $\beta$ is varying or not.

\begin{example}
To explore the influence of the $F_n$ distribution, we consider $F_n$ as $t(3)$ distribution in this example. Other settings are as same
as Example 1.
\end{example}
All the results are shown in Table \ref{tab:4}-\ref{tab:6}. Compared with Example 1, the performance of five methods become worse, especially the conventional RQ method.

\begin{table}[htbp]
\tiny
\caption{\small{The MedSE of coefficients for four settings with $n=80$, where $F_n$ is $t(3)$ in Example 2.}}
\center
    \begin{tabular}{rlrrrrrrrrr} \hline \hline
          &       & \multicolumn{9}{c}{$\tau$ } \\
          &       & 0.1   & 0.2   & 0.3   & 0.4   & 0.5   & 0.6   & 0.7   & 0.8   & 0.9 \\ \hline
          &       & \multicolumn{9}{c}{I: $b = 0.5,~c_0 = 0.1,~c_1 = 0.2$} \\
          & RQ    & 0.9463  & 0.6905  & 0.6112  & 0.4751  & 0.4608  & 0.5056  & 0.5749  & 0.5150  & 0.8897  \\
          & FL    & 0.3171  & 0.3123  & 0.2972  & 0.2714  & 0.2998  & 0.3217  & 0.3554  & 0.3221  & 0.3635  \\
          & FAL   & 0.3250  & 0.3385  & 0.3135  & 0.2978  & 0.3242  & 0.3289  & 0.3599  & 0.3862  & 0.3840  \\
          & FS    & 0.4135  & 0.4274  & 0.3887  & 0.3657  & 0.3687  & 0.3722  & 0.3765  & 0.3420  & 0.4699  \\
          & FAS   & 0.4240  & 0.4209  & 0.3764  & 0.3665  & 0.3608  & 0.3514  & 0.3730  & 0.3537  & 0.4267  \\
          &       & \multicolumn{9}{c}{II: $b = 0.5,~c_0 = 0,~c_1 = 0.2$} \\
          & RQ    & 0.5580  & 0.3408  & 0.3053  & 0.3099  & 0.2482  & 0.2472  & 0.2909  & 0.3254  & 0.4617  \\
          & FL    & 0.1949  & 0.1621  & 0.1686  & 0.1763  & 0.1823  & 0.1820  & 0.1830  & 0.1881  & 0.2157  \\
          & FAL   & 0.2013  & 0.1549  & 0.1644  & 0.1731  & 0.1754  & 0.1765  & 0.1878  & 0.1843  & 0.1933  \\
          & FS    & 0.2917  & 0.2272  & 0.2267  & 0.2136  & 0.1907  & 0.1982  & 0.1857  & 0.2158  & 0.2823  \\
          & FAS   & 0.2947  & 0.2145  & 0.2173  & 0.2185  & 0.1897  & 0.1941  & 0.2047  & 0.2210  & 0.2737  \\
    \multicolumn{1}{l}{$\lambda = 0.2$} &       & \multicolumn{9}{c}{III: $b = 0.5,~c_0 = 0.1,~c_1 = 0$} \\
          & RQ    & 0.7528  & 0.6191  & 0.5385  & 0.4096  & 0.4338  & 0.4160  & 0.5480  & 0.4863  & 0.6749  \\
          & FL    & 0.3001  & 0.3052  & 0.3056  & 0.2909  & 0.2579  & 0.2538  & 0.2722  & 0.2625  & 0.2995  \\
          & FAL   & 0.2914  & 0.2994  & 0.2899  & 0.2813  & 0.2751  & 0.2657  & 0.2614  & 0.2641  & 0.2733  \\
          & FS    & 0.3966  & 0.3516  & 0.3631  & 0.3254  & 0.3099  & 0.3101  & 0.3136  & 0.2948  & 0.3602  \\
          & FAS   & 0.3884  & 0.3876  & 0.3481  & 0.3098  & 0.2958  & 0.3178  & 0.3250  & 0.3269  & 0.3552  \\
          &       & \multicolumn{9}{c}{IV: $b = 0.5,~c_0 = 0,~c_1 = 0$} \\
          & RQ    & 0.5243  & 0.2756  & 0.2604  & 0.2807  & 0.2193  & 0.2225  & 0.2736  & 0.3078  & 0.4781  \\
          & FL    & 0.1504  & 0.1452  & 0.1547  & 0.1565  & 0.1535  & 0.1378  & 0.1319  & 0.1266  & 0.1257  \\
          & FAL   & 0.1608  & 0.1511  & 0.1552  & 0.1466  & 0.1512  & 0.1512  & 0.1498  & 0.1434  & 0.1590  \\
          & FS    & 0.2565  & 0.1778  & 0.1978  & 0.2169  & 0.1954  & 0.1693  & 0.1648  & 0.1711  & 0.2379  \\
          & FAS   & 0.2359  & 0.1791  & 0.1969  & 0.2011  & 0.1849  & 0.1588  & 0.1381  & 0.1776  & 0.2206  \\
          &       &       &       &       &       &       &       &       &       &  \\
          &       & \multicolumn{9}{c}{I: $b = 0.5,~c_0 = 0.1,~c_1 = 0.2$} \\
          & RQ    & 1.1589  & 0.9232  & 0.8547  & 0.8133  & 0.7533  & 0.8164  & 1.0451  & 1.2126  & 1.9087  \\
          & FL    & 0.4223  & 0.4224  & 0.4194  & 0.4478  & 0.4443  & 0.4732  & 0.4551  & 0.4439  & 0.4976  \\
          & FAL   & 0.4094  & 0.4163  & 0.4104  & 0.4762  & 0.4448  & 0.5036  & 0.4843  & 0.4590  & 0.5056  \\
          & FS    & 0.4645  & 0.4472  & 0.4872  & 0.4819  & 0.4737  & 0.4915  & 0.5183  & 0.5250  & 0.5645  \\
          & FAS   & 0.5026  & 0.4251  & 0.4910  & 0.5114  & 0.4871  & 0.5441  & 0.5629  & 0.5339  & 0.6779  \\
          &       & \multicolumn{9}{c}{II: $b = 0.5,~c_0 = 0,~c_1 = 0.2$} \\
          & RQ    & 0.5790  & 0.3768  & 0.3172  & 0.2952  & 0.2585  & 0.2803  & 0.2697  & 0.3597  & 0.5286  \\
          & FL    & 0.2036  & 0.1673  & 0.1550  & 0.1606  & 0.1552  & 0.1471  & 0.1521  & 0.1645  & 0.1947  \\
          & FAL   & 0.1784  & 0.1540  & 0.1587  & 0.1504  & 0.1515  & 0.1423  & 0.1573  & 0.1847  & 0.2267  \\
          & FS    & 0.2299  & 0.2068  & 0.2110  & 0.1969  & 0.1785  & 0.1827  & 0.1720  & 0.2095  & 0.2617  \\
          & FAS   & 0.2409  & 0.1891  & 0.2126  & 0.1800  & 0.1895  & 0.1755  & 0.1746  & 0.2045  & 0.2548  \\
    \multicolumn{1}{l}{$\lambda = 0.5$} &       & \multicolumn{9}{c}{III: $b = 0.5,~c_0 = 0.1,~c_1 = 0$} \\
          & RQ    & 0.9851  & 0.7831  & 0.7467  & 0.6325  & 0.6333  & 0.7397  & 0.9280  & 1.0139  & 1.7633  \\
          & FL    & 0.3679  & 0.3601  & 0.3656  & 0.3699  & 0.3969  & 0.4188  & 0.3714  & 0.3458  & 0.3680  \\
          & FAL   & 0.4052  & 0.3703  & 0.3872  & 0.3933  & 0.3893  & 0.4176  & 0.3914  & 0.4079  & 0.4725  \\
          & FS    & 0.4579  & 0.3823  & 0.4248  & 0.4117  & 0.4532  & 0.4551  & 0.5116  & 0.4774  & 0.4883  \\
          & FAS   & 0.4613  & 0.4257  & 0.4425  & 0.4378  & 0.4305  & 0.4906  & 0.5226  & 0.5170  & 0.5758  \\
          &       & \multicolumn{9}{c}{IV: $b = 0.5,~c_0 = 0,~c_1 = 0$} \\
          & RQ    & 0.4993  & 0.3060  & 0.2737  & 0.2614  & 0.2150  & 0.2355  & 0.2818  & 0.3435  & 0.5247  \\
          & FL    & 0.1549  & 0.1537  & 0.1409  & 0.1399  & 0.1343  & 0.1261  & 0.1197  & 0.1343  & 0.1408  \\
          & FAL   & 0.1352  & 0.1357  & 0.1387  & 0.1372  & 0.1350  & 0.1240  & 0.1226  & 0.1238  & 0.1394  \\
          & FS    & 0.1957  & 0.1678  & 0.1721  & 0.1798  & 0.1433  & 0.1365  & 0.1442  & 0.1754  & 0.1941  \\
          & FAS   & 0.1831  & 0.1583  & 0.1653  & 0.1698  & 0.1297  & 0.1379  & 0.1398  & 0.1424  & 0.1728  \\
          &       &       &       &       &       &       &       &       &       &  \\
          &       & \multicolumn{9}{c}{I: $b = 0.5,~c_0 = 0.1,~c_1 = 0.2$} \\
          & RQ    & 3.1306  & 2.6236  & 3.5506  & 3.7609  & 3.7372  & 4.0245  & 4.7277  & 8.1941  & 13.3519  \\
          & FL    & 1.1891  & 1.0885  & 1.1833  & 1.2415  & 1.2536  & 1.1719  & 1.3045  & 1.3128  & 1.4532  \\
          & FAL   & 1.2363  & 1.1497  & 1.2735  & 1.3569  & 1.3353  & 1.3351  & 1.6303  & 1.6757  & 2.0037  \\
          & FS    & 1.1303  & 1.0580  & 1.1879  & 1.2767  & 1.2881  & 1.3063  & 1.4461  & 1.4823  & 1.8388  \\
          & FAS   & 1.5815  & 1.3213  & 1.6268  & 1.6415  & 1.7149  & 1.7599  & 2.1459  & 2.2368  & 3.4451  \\
          &       & \multicolumn{9}{c}{II: $b = 0.5,~c_0 = 0,~c_1 = 0.2$} \\
          & RQ    & 0.9571  & 0.8117  & 0.5327  & 0.5378  & 0.4504  & 0.4627  & 0.5944  & 0.8268  & 0.7881  \\
          & FL    & 0.1939  & 0.1626  & 0.1721  & 0.1939  & 0.1980  & 0.1942  & 0.2120  & 0.2194  & 0.2460  \\
          & FAL   & 0.1717  & 0.1677  & 0.1906  & 0.1963  & 0.1926  & 0.1911  & 0.2072  & 0.1898  & 0.2068  \\
          & FS    & 0.2135  & 0.2110  & 0.1916  & 0.1872  & 0.2216  & 0.2128  & 0.1993  & 0.2416  & 0.2651  \\
          & FAS   & 0.2907  & 0.2728  & 0.2332  & 0.2454  & 0.2388  & 0.2136  & 0.2033  & 0.2381  & 0.3052  \\
    \multicolumn{1}{l}{$\lambda = 0.8$} &       & \multicolumn{9}{c}{III: $b = 0.5,~c_0 = 0.1,~c_1 = 0$} \\
          & RQ    & 3.3983  & 2.8706  & 3.6200  & 3.4432  & 3.4382  & 4.3822  & 4.4855  & 6.9316  & 13.4683  \\
          & FL    & 1.0305  & 0.9830  & 1.0662  & 1.1598  & 1.2005  & 1.1905  & 1.2710  & 1.2058  & 1.3128  \\
          & FAL   & 1.0790  & 1.0557  & 1.0053  & 1.1204  & 1.1138  & 1.1969  & 1.5089  & 1.6264  & 1.8188  \\
          & FS    & 1.0290  & 1.0074  & 1.0765  & 1.2595  & 1.2421  & 1.3246  & 1.4660  & 1.4274  & 1.7690  \\
          & FAS   & 1.2692  & 1.0785  & 1.1312  & 1.5628  & 1.5772  & 1.8234  & 2.0001  & 2.3276  & 3.0492  \\
          &       & \multicolumn{9}{c}{IV: $b = 0.5,~c_0 = 0,~c_1 = 0$} \\
          & RQ    & 0.8900  & 0.6533  & 0.5951  & 0.5080  & 0.4479  & 0.4311  & 0.4449  & 0.6375  & 0.8158  \\
          & FL    & 0.1756  & 0.1642  & 0.1708  & 0.1635  & 0.1642  & 0.1618  & 0.1384  & 0.1416  & 0.1529  \\
          & FAL   & 0.1676  & 0.1656  & 0.1673  & 0.1656  & 0.1640  & 0.1435  & 0.1501  & 0.1462  & 0.1544  \\
          & FS    & 0.2049  & 0.1955  & 0.1875  & 0.1706  & 0.1860  & 0.1656  & 0.1288  & 0.1680  & 0.1683  \\
          & FAS   & 0.2388  & 0.2419  & 0.2177  & 0.2168  & 0.2263  & 0.1724  & 0.1408  & 0.2111  & 0.2362  \\ \\ \hline
    \end{tabular}%
\label{tab:4}
\end{table}%

\begin{table}[htbp]
\tiny
\caption{\small{The MedSE of coefficients for four settings with $n=120$, where $F_n$ is $t(3)$ in Example 2.}}
\center
    \begin{tabular}{rlrrrrrrrrr} \hline \hline
          &       & \multicolumn{9}{c}{$\tau$ } \\
          &       & 0.1   & 0.2   & 0.3   & 0.4   & 0.5   & 0.6   & 0.7   & 0.8   & 0.9 \\ \hline
          &       & \multicolumn{9}{c}{I: $b = 0.5,~c_0 = 0.1,~c_1 = 0.2$} \\
          & RQ    & 0.4104  & 0.4361  & 0.3280  & 0.2880  & 0.2693  & 0.3108  & 0.3802  & 0.4710  & 0.5317  \\
          & FL    & 0.2118  & 0.2002  & 0.1890  & 0.1844  & 0.2010  & 0.2178  & 0.2423  & 0.2572  & 0.2688  \\
          & FAL   & 0.2145  & 0.1899  & 0.1834  & 0.1953  & 0.2002  & 0.2130  & 0.2350  & 0.2624  & 0.2851  \\
          & FS    & 0.2874  & 0.2706  & 0.2321  & 0.2404  & 0.2121  & 0.2253  & 0.2664  & 0.2877  & 0.3488  \\
          & FAS   & 0.2837  & 0.2610  & 0.2123  & 0.2263  & 0.2047  & 0.2371  & 0.2522  & 0.3010  & 0.3334  \\
          &       & \multicolumn{9}{c}{II: $b = 0.5,~c_0 = 0,~c_1 = 0.2$} \\
          & RQ    & 0.3257  & 0.2137  & 0.2095  & 0.1974  & 0.1949  & 0.1898  & 0.2230  & 0.2040  & 0.3143  \\
          & FL    & 0.1469  & 0.1153  & 0.1183  & 0.1406  & 0.1267  & 0.1240  & 0.1310  & 0.1430  & 0.1866  \\
          & FAL   & 0.1398  & 0.1193  & 0.1151  & 0.1247  & 0.1317  & 0.1251  & 0.1385  & 0.1520  & 0.1875  \\
          & FS    & 0.2049  & 0.1729  & 0.1602  & 0.1560  & 0.1443  & 0.1592  & 0.1576  & 0.1688  & 0.1964  \\
          & FAS   & 0.1735  & 0.1372  & 0.1378  & 0.1457  & 0.1309  & 0.1553  & 0.1411  & 0.1521  & 0.1780  \\
    \multicolumn{1}{l}{$\lambda = 0.2$} &       & \multicolumn{9}{c}{III: $b = 0.5,~c_0 = 0.1,~c_1 = 0$} \\
          & RQ    & 0.3891  & 0.3996  & 0.2927  & 0.3027  & 0.2430  & 0.2814  & 0.3272  & 0.3952  & 0.5016  \\
          & FL    & 0.1793  & 0.1876  & 0.1789  & 0.1872  & 0.1806  & 0.1835  & 0.1915  & 0.2071  & 0.2162  \\
          & FAL   & 0.1837  & 0.1933  & 0.1750  & 0.1854  & 0.1741  & 0.1697  & 0.1799  & 0.1929  & 0.2057  \\
          & FS    & 0.2574  & 0.2589  & 0.2316  & 0.2193  & 0.2038  & 0.1993  & 0.2545  & 0.2660  & 0.2834  \\
          & FAS   & 0.2414  & 0.2387  & 0.2158  & 0.2106  & 0.1836  & 0.2029  & 0.2554  & 0.2647  & 0.2742  \\
          &       & \multicolumn{9}{c}{IV: $b = 0.5,~c_0 = 0,~c_1 = 0$} \\
          & RQ    & 0.3383  & 0.2078  & 0.1946  & 0.1637  & 0.1712  & 0.1711  & 0.1808  & 0.1960  & 0.2987  \\
          & FL    & 0.1239  & 0.1209  & 0.1113  & 0.1206  & 0.1195  & 0.1216  & 0.1181  & 0.1127  & 0.1149  \\
          & FAL   & 0.1215  & 0.1188  & 0.1069  & 0.1170  & 0.1205  & 0.1138  & 0.1105  & 0.1115  & 0.1226  \\
          & FS    & 0.2001  & 0.1548  & 0.1514  & 0.1427  & 0.1300  & 0.1431  & 0.1400  & 0.1480  & 0.1654  \\
          & FAS   & 0.1713  & 0.1367  & 0.1394  & 0.1187  & 0.1238  & 0.1266  & 0.1301  & 0.1247  & 0.1546  \\
          &       &       &       &       &       &       &       &       &       &  \\
          &       & \multicolumn{9}{c}{I: $b = 0.5,~c_0 = 0.1,~c_1 = 0.2$} \\
          & RQ    & 0.5901  & 0.5134  & 0.4643  & 0.4246  & 0.4170  & 0.5018  & 0.5817  & 0.7496  & 1.1272  \\
          & FL    & 0.2184  & 0.2345  & 0.2229  & 0.2336  & 0.2298  & 0.2711  & 0.2609  & 0.2920  & 0.3631  \\
          & FAL   & 0.2472  & 0.2291  & 0.2379  & 0.2334  & 0.2395  & 0.2688  & 0.2763  & 0.3251  & 0.4019  \\
          & FS    & 0.2638  & 0.2550  & 0.2701  & 0.2796  & 0.2552  & 0.2941  & 0.3119  & 0.3239  & 0.4555  \\
          & FAS   & 0.2590  & 0.2507  & 0.2823  & 0.2792  & 0.2652  & 0.3461  & 0.3278  & 0.3574  & 0.4821  \\
          &       & \multicolumn{9}{c}{II: $b = 0.5,~c_0 = 0,~c_1 = 0.2$} \\
          & RQ    & 0.3394  & 0.2319  & 0.2119  & 0.1840  & 0.1903  & 0.2145  & 0.2228  & 0.2425  & 0.3510  \\
          & FL    & 0.1275  & 0.1252  & 0.1125  & 0.1251  & 0.1054  & 0.1091  & 0.1145  & 0.1208  & 0.1573  \\
          & FAL   & 0.1438  & 0.1282  & 0.1129  & 0.1090  & 0.1043  & 0.1003  & 0.1189  & 0.1292  & 0.1515  \\
          & FS    & 0.1633  & 0.1449  & 0.1374  & 0.1338  & 0.1344  & 0.1327  & 0.1474  & 0.1544  & 0.1956  \\
          & FAS   & 0.1854  & 0.1533  & 0.1290  & 0.1339  & 0.1320  & 0.1428  & 0.1270  & 0.1499  & 0.1784  \\
    \multicolumn{1}{l}{$\lambda = 0.5$} &       & \multicolumn{9}{c}{III: $b = 0.5,~c_0 = 0.1,~c_1 = 0$} \\
          & RQ    & 0.5334  & 0.4258  & 0.4504  & 0.3717  & 0.4155  & 0.4792  & 0.5162  & 0.6462  & 1.0090   \\
          & FL    & 0.1969  & 0.1839  & 0.1948  & 0.2016  & 0.2116  & 0.2397  & 0.2340  & 0.2543  & 0.2858  \\
          & FAL   & 0.2025  & 0.2137  & 0.2085  & 0.2112  & 0.2162  & 0.2341  & 0.2414  & 0.2756  & 0.3353  \\
          & FS    & 0.2420  & 0.2316  & 0.2592  & 0.2424  & 0.2461  & 0.2682  & 0.2858  & 0.3090  & 0.4055  \\
          & FAS   & 0.2816  & 0.2330  & 0.2612  & 0.2325  & 0.2627  & 0.3131  & 0.3110  & 0.3458  & 0.4095  \\
          &       & \multicolumn{9}{c}{IV: $b = 0.5,~c_0 = 0,~c_1 = 0$} \\
          & RQ    & 0.3559  & 0.2108  & 0.1768  & 0.1839  & 0.1835  & 0.1894  & 0.1935  & 0.2421  & 0.3322  \\
          & FL    & 0.1167  & 0.1079  & 0.1190  & 0.1205  & 0.1107  & 0.1048  & 0.1017  & 0.1025  & 0.1039  \\
          & FAL   & 0.1112  & 0.1032  & 0.0995  & 0.1025  & 0.0993  & 0.0907  & 0.0897  & 0.0930  & 0.1021  \\
          & FS    & 0.1552  & 0.1367  & 0.1375  & 0.1361  & 0.1279  & 0.1216  & 0.1249  & 0.1362  & 0.1702  \\
          & FAS   & 0.1428  & 0.1324  & 0.1218  & 0.1245  & 0.1253  & 0.1143  & 0.1166  & 0.1220  & 0.1417  \\
          &       &       &       &       &       &       &       &       &       &  \\
          &       & \multicolumn{9}{c}{I: $b = 0.5,~c_0 = 0.1,~c_1 = 0.2$} \\
          & RQ    & 2.1663  & 1.9115  & 2.0349  & 2.1766  & 2.5195  & 2.7271  & 3.3710  & 5.1528  & 12.2799   \\
          & FL    & 0.7338  & 0.6356  & 0.6210  & 0.6897  & 0.7421  & 0.8004  & 0.8878  & 0.9682  & 0.9244  \\
          & FAL   & 0.6572  & 0.5917  & 0.6444  & 0.7348  & 0.7966  & 0.8595  & 0.9245  & 1.0357  & 1.1948  \\
          & FS    & 0.7174  & 0.6077  & 0.6251  & 0.8264  & 0.8702  & 0.9933  & 1.1399  & 1.1942  & 1.3734  \\
          & FAS   & 1.0481  & 0.6457  & 0.8442  & 1.1270  & 1.2220  & 1.0712  & 1.4788  & 1.7912  & 2.3435  \\
          &       & \multicolumn{9}{c}{II: $b = 0.5,~c_0 = 0,~c_1 = 0.2$} \\
          & RQ    & 0.7180  & 0.4630  & 0.3765  & 0.3390  & 0.3006  & 0.3649  & 0.4696  & 0.4130  & 0.6152  \\
          & FL    & 0.1675  & 0.1438  & 0.1344  & 0.1388  & 0.1256  & 0.1397  & 0.1171  & 0.1544  & 0.1574  \\
          & FAL   & 0.1733  & 0.1293  & 0.1180  & 0.1104  & 0.1213  & 0.1145  & 0.1018  & 0.1131  & 0.1565  \\
          & FS    & 0.2060  & 0.1736  & 0.1506  & 0.1393  & 0.1204  & 0.1343  & 0.1171  & 0.1433  & 0.1780  \\
          & FAS   & 0.2504  & 0.1897  & 0.1623  & 0.1561  & 0.1329  & 0.1400  & 0.1463  & 0.1602  & 0.2423  \\
    \multicolumn{1}{l}{$\lambda = 0.8$} &       & \multicolumn{9}{c}{III: $b = 0.5,~c_0 = 0.1,~c_1 = 0$} \\
          & RQ    & 1.9266  & 1.5241  & 1.7587  & 2.1651  & 2.5707  & 2.6144  & 3.4843  & 4.9690  & 11.3622  \\
          & FL    & 0.6688  & 0.6407  & 0.6128  & 0.6248  & 0.6813  & 0.6823  & 0.8301  & 0.9002  & 0.9075  \\
          & FAL   & 0.5757  & 0.5562  & 0.5519  & 0.6574  & 0.7581  & 0.7893  & 0.8592  & 0.9548  & 1.1198  \\
          & FS    & 0.7504  & 0.6530  & 0.6139  & 0.6783  & 0.7094  & 0.7793  & 0.9972  & 1.0546  & 1.3255  \\
          & FAS   & 0.8747  & 0.6562  & 0.8183  & 0.8250  & 1.1498  & 0.8124  & 1.2141  & 1.5672  & 2.2539  \\
          &       & \multicolumn{9}{c}{IV: $b = 0.5,~c_0 = 0,~c_1 = 0$} \\
          & RQ    & 0.6024  & 0.3952  & 0.3613  & 0.2903  & 0.3242  & 0.3734  & 0.3185  & 0.4474  & 0.6622  \\
          & FL    & 0.1121  & 0.1080  & 0.1041  & 0.0977  & 0.0938  & 0.1091  & 0.0919  & 0.1003  & 0.1284  \\
          & FAL   & 0.1098  & 0.1137  & 0.1053  & 0.1063  & 0.1063  & 0.1060  & 0.0913  & 0.1171  & 0.1262  \\
          & FS    & 0.1748  & 0.1766  & 0.1549  & 0.1296  & 0.1114  & 0.1114  & 0.1377  & 0.1425  & 0.1666  \\
          & FAS   & 0.2112  & 0.1921  & 0.1645  & 0.1435  & 0.1221  & 0.1410  & 0.1423  & 0.1428  & 0.2123  \\ \\ \hline
    \end{tabular}%
\label{tab:5}
\end{table}%

\begin{table}[htbp]
\tiny
\caption{\small{The MedSE of coefficients for four settings with $n=160$, where $F_n$ is $t(3)$ in Example 2.}}
\center
    \begin{tabular}{rlrrrrrrrrr} \hline \hline
          &       & \multicolumn{9}{c}{$\tau$ } \\
          &       & 0.1   & 0.2   & 0.3   & 0.4   & 0.5   & 0.6   & 0.7   & 0.8   & 0.9 \\ \hline
          &       & \multicolumn{9}{c}{I: $b = 0.5,~c_0 = 0.1,~c_1 = 0.2$} \\
          & RQ    & 0.3097  & 0.2284  & 0.2403  & 0.2530  & 0.2038  & 0.2071  & 0.2702  & 0.3022  & 0.3597  \\
          & FL    & 0.1474  & 0.1383  & 0.1404  & 0.1374  & 0.1436  & 0.1441  & 0.1551  & 0.1716  & 0.1653  \\
          & FAL   & 0.1621  & 0.1475  & 0.1513  & 0.1484  & 0.1439  & 0.1490  & 0.1544  & 0.1668  & 0.1702  \\
          & FS    & 0.1972  & 0.1790  & 0.1714  & 0.1590  & 0.1735  & 0.1553  & 0.1816  & 0.2043  & 0.2582  \\
          & FAS   & 0.1854  & 0.1781  & 0.1664  & 0.1523  & 0.1541  & 0.1448  & 0.1640  & 0.1929  & 0.2283  \\
          &       & \multicolumn{9}{c}{II: $b = 0.5,~c_0 = 0,~c_1 = 0.2$} \\
          & RQ    & 0.2424  & 0.1734  & 0.1488  & 0.1556  & 0.1471  & 0.1424  & 0.1439  & 0.1996  & 0.2201  \\
          & FL    & 0.1145  & 0.1013  & 0.0997  & 0.1001  & 0.1008  & 0.1045  & 0.1096  & 0.1116  & 0.1390  \\
          & FAL   & 0.1218  & 0.1059  & 0.1063  & 0.0957  & 0.0927  & 0.1004  & 0.1077  & 0.1195  & 0.1434  \\
          & FS    & 0.1449  & 0.1158  & 0.1160  & 0.1211  & 0.1117  & 0.1156  & 0.1171  & 0.1319  & 0.1591  \\
          & FAS   & 0.1227  & 0.1120  & 0.1067  & 0.1048  & 0.0996  & 0.1166  & 0.1148  & 0.1203  & 0.1582  \\
    \multicolumn{1}{l}{$\lambda = 0.2$} &       & \multicolumn{9}{c}{III: $b = 0.5,~c_0 = 0.1,~c_1 = 0$} \\
          & RQ    & 0.3097  & 0.2284  & 0.2403  & 0.2530  & 0.2038  & 0.2071  & 0.2702  & 0.3022  & 0.3597  \\
          & FL    & 0.1474  & 0.1383  & 0.1404  & 0.1374  & 0.1436  & 0.1441  & 0.1551  & 0.1716  & 0.1653  \\
          & FAL   & 0.1621  & 0.1475  & 0.1513  & 0.1484  & 0.1439  & 0.1490  & 0.1544  & 0.1668  & 0.1702  \\
          & FS    & 0.1972  & 0.1790  & 0.1714  & 0.1590  & 0.1735  & 0.1553  & 0.1816  & 0.2043  & 0.2582  \\
          & FAS   & 0.1854  & 0.1781  & 0.1664  & 0.1523  & 0.1541  & 0.1448  & 0.1640  & 0.1929  & 0.2283  \\
          &       & \multicolumn{9}{c}{IV: $b = 0.5,~c_0 = 0,~c_1 = 0$} \\
          & RQ    & 0.1957  & 0.1516  & 0.1328  & 0.1292  & 0.1333  & 0.1230  & 0.1401  & 0.1515  & 0.1895  \\
          & FL    & 0.0915  & 0.0907  & 0.0914  & 0.0924  & 0.0918  & 0.0924  & 0.0896  & 0.0870  & 0.0870  \\
          & FAL   & 0.0960  & 0.0930  & 0.0944  & 0.0950  & 0.0953  & 0.0895  & 0.0858  & 0.0926  & 0.0989  \\
          & FS    & 0.1209  & 0.1034  & 0.1095  & 0.0959  & 0.1055  & 0.1017  & 0.0984  & 0.1125  & 0.1398  \\
          & FAS   & 0.1032  & 0.0953  & 0.0941  & 0.0918  & 0.0975  & 0.0997  & 0.0930  & 0.0995  & 0.1184  \\
          &       &       &       &       &       &       &       &       &       &  \\
          &       & \multicolumn{9}{c}{I: $b = 0.5,~c_0 = 0.1,~c_1 = 0.2$} \\
          & RQ    & 0.4755  & 0.3241  & 0.2922  & 0.3561  & 0.3093  & 0.4004  & 0.4312  & 0.4713  & 1.0290  \\
          & FL    & 0.1856  & 0.1806  & 0.1864  & 0.2019  & 0.2045  & 0.1950  & 0.2251  & 0.2278  & 0.2455  \\
          & FAL   & 0.2224  & 0.2011  & 0.1955  & 0.1996  & 0.2088  & 0.2056  & 0.2225  & 0.2238  & 0.2662  \\
          & FS    & 0.2623  & 0.2011  & 0.2148  & 0.2035  & 0.1958  & 0.2237  & 0.2523  & 0.2772  & 0.3963  \\
          & FAS   & 0.2528  & 0.1916  & 0.2237  & 0.2100  & 0.2110  & 0.2437  & 0.2758  & 0.2753  & 0.4243  \\
          &       & \multicolumn{9}{c}{II: $b = 0.5,~c_0 = 0,~c_1 = 0.2$} \\
          & RQ    & 0.2653  & 0.1738  & 0.1485  & 0.1425  & 0.1562  & 0.1677  & 0.1429  & 0.2087  & 0.2808  \\
          & FL    & 0.1164  & 0.0940  & 0.0938  & 0.0980  & 0.0969  & 0.0917  & 0.0997  & 0.1071  & 0.1406  \\
          & FAL   & 0.1060  & 0.0913  & 0.0917  & 0.0900  & 0.0914  & 0.0907  & 0.0954  & 0.0924  & 0.1196  \\
          & FS    & 0.1218  & 0.1121  & 0.1074  & 0.1087  & 0.1085  & 0.0972  & 0.1017  & 0.1233  & 0.1509  \\
          & FAS   & 0.1452  & 0.1062  & 0.1008  & 0.1032  & 0.1030  & 0.0976  & 0.0954  & 0.1110  & 0.1409  \\
    \multicolumn{1}{l}{$\lambda = 0.5$} &       & \multicolumn{9}{c}{III: $b = 0.5,~c_0 = 0.1,~c_1 = 0$} \\
          & RQ    & 0.4755  & 0.3241  & 0.2922  & 0.3561  & 0.3093  & 0.4004  & 0.4312  & 0.4713  & 1.0290  \\
          & FL    & 0.1856  & 0.1806  & 0.1864  & 0.2019  & 0.2045  & 0.1950  & 0.2251  & 0.2278  & 0.2455  \\
          & FAL   & 0.2224  & 0.2011  & 0.1955  & 0.1996  & 0.2088  & 0.2056  & 0.2225  & 0.2238  & 0.2662  \\
          & FS    & 0.2623  & 0.2011  & 0.2148  & 0.2035  & 0.1958  & 0.2237  & 0.2523  & 0.2772  & 0.3963  \\
          & FAS   & 0.2528  & 0.1916  & 0.2237  & 0.2100  & 0.2110  & 0.2437  & 0.2758  & 0.2753  & 0.4243  \\
          &       & \multicolumn{9}{c}{IV: $b = 0.5,~c_0 = 0,~c_1 = 0$} \\
          & RQ    & 0.2340  & 0.1537  & 0.1398  & 0.1379  & 0.1483  & 0.1234  & 0.1301  & 0.1545  & 0.2085  \\
          & FL    & 0.0888  & 0.0919  & 0.0877  & 0.0963  & 0.0907  & 0.0885  & 0.0900  & 0.0782  & 0.0786  \\
          & FAL   & 0.0918  & 0.0820  & 0.0837  & 0.0882  & 0.0910  & 0.0868  & 0.0807  & 0.0766  & 0.0785  \\
          & FS    & 0.1192  & 0.0935  & 0.1000  & 0.0989  & 0.1101  & 0.0919  & 0.0926  & 0.0948  & 0.1179  \\
          & FAS   & 0.1102  & 0.0912  & 0.0955  & 0.0904  & 0.0994  & 0.0890  & 0.0814  & 0.0888  & 0.1075  \\
          &       &       &       &       &       &       &       &       &       &  \\
          &       & \multicolumn{9}{c}{I: $b = 0.5,~c_0 = 0.1,~c_1 = 0.2$} \\
          & RQ    & 1.7113  & 1.4183  & 1.5008  & 1.1765  & 1.8837  & 2.0089  & 2.8445  & 5.0104  & 7.0400  \\
          & FL    & 0.5696  & 0.5669  & 0.5263  & 0.5340  & 0.5555  & 0.6361  & 0.7080  & 0.7815  & 0.8044  \\
          & FAL   & 0.7013  & 0.6132  & 0.5920  & 0.5374  & 0.5831  & 0.7161  & 0.7684  & 0.8792  & 1.2309  \\
          & FS    & 0.6657  & 0.5929  & 0.5183  & 0.5201  & 0.5871  & 0.7161  & 0.8414  & 0.9373  & 1.1181  \\
          & FAS   & 0.8514  & 0.6363  & 0.5727  & 0.6124  & 0.6867  & 1.0546  & 1.1254  & 1.4436  & 1.6967  \\
          &       & \multicolumn{9}{c}{II: $b = 0.5,~c_0 = 0,~c_1 = 0.2$} \\
          & RQ    & 0.4399  & 0.3142  & 0.2862  & 0.3161  & 0.2716  & 0.2226  & 0.2899  & 0.3135  & 0.5160  \\
          & FL    & 0.1203  & 0.1127  & 0.1095  & 0.1012  & 0.1013  & 0.0880  & 0.0893  & 0.0857  & 0.1196  \\
          & FAL   & 0.1162  & 0.1080  & 0.1089  & 0.0981  & 0.0947  & 0.0879  & 0.0958  & 0.0881  & 0.1352  \\
          & FS    & 0.1333  & 0.1100  & 0.1102  & 0.0995  & 0.1100  & 0.0987  & 0.1106  & 0.1075  & 0.1406  \\
          & FAS   & 0.1508  & 0.1290  & 0.1051  & 0.1065  & 0.1119  & 0.1207  & 0.1229  & 0.1097  & 0.1620  \\
    \multicolumn{1}{l}{$\lambda = 0.8$} &       & \multicolumn{9}{c}{III: $b = 0.5,~c_0 = 0.1,~c_1 = 0$} \\
          & RQ    & 1.7113  & 1.4183  & 1.5008  & 1.1765  & 1.8837  & 2.0089  & 2.8445  & 5.0104  & 7.0400  \\
          & FL    & 0.5696  & 0.5669  & 0.5263  & 0.5340  & 0.5555  & 0.6361  & 0.7080  & 0.7815  & 0.8044  \\
          & FAL   & 0.7013  & 0.6132  & 0.5920  & 0.5374  & 0.5831  & 0.7161  & 0.7684  & 0.8792  & 1.2309  \\
          & FS    & 0.6657  & 0.5929  & 0.5183  & 0.5201  & 0.5871  & 0.7161  & 0.8414  & 0.9373  & 1.1181  \\
          & FAS   & 0.8514  & 0.6363  & 0.5727  & 0.6124  & 0.6867  & 1.0546  & 1.1254  & 1.4436  & 1.6967  \\
          &       & \multicolumn{9}{c}{IV: $b = 0.5,~c_0 = 0,~c_1 = 0$} \\
          & RQ    & 0.3672  & 0.2670  & 0.2461  & 0.2261  & 0.2610  & 0.1985  & 0.3051  & 0.3077  & 0.4630  \\
          & FL    & 0.0984  & 0.0905  & 0.0988  & 0.0983  & 0.0929  & 0.0890  & 0.0864  & 0.0920  & 0.0962  \\
          & FAL   & 0.0765  & 0.0773  & 0.0901  & 0.0870  & 0.0922  & 0.0770  & 0.0780  & 0.0771  & 0.0804  \\
          & FS    & 0.1175  & 0.1144  & 0.1225  & 0.1065  & 0.0911  & 0.1123  & 0.1067  & 0.1073  & 0.1323  \\
          & FAS   & 0.1194  & 0.1205  & 0.1127  & 0.0829  & 0.1154  & 0.1137  & 0.1061  & 0.1201  & 0.1606 \\ \\ \hline
    \end{tabular}%
\label{tab:6}
\end{table}%

\begin{example}
To further study the difference between fused LASSO and fused Sup-norm, we consider a more complex example. The data are generated from model
\begin{equation}
\label{exam:3}
Y_i = \alpha(\tau_{n,i}) + \lambda U_i + {{\beta}(\tau_{n,i})} {X}_{i} + e_i, ~~i = 1, \cdots, n,
\end{equation}
where $\tau_{n,i}$ is randomly from $\mathbb{S}_{\tau}$, $\alpha(\tau_{n,i}) = \alpha +  b F_n^{-1}(\tau_{n,i})$,
$$
\beta(\tau_{n,i}) =
\left \{
\begin{aligned}
&\beta + c_1 F_n^{-1}(\tau_{n,i}) && 0 < \tau_{n_i} < 0.49, \\
&\beta && 0.49 \leq \tau_{n_i} <1.
\end{aligned}
\right.
$$
We only consider the situation where $F_n$ is the standard normal distribution $N(0,1)$ , the sample size $n= 120$. The spatial lag parameter $\lambda$ is a constant, and we consider three situations $\lambda = 0.2,~0.5,~0.8$. $\alpha$ is chosen as $\alpha = 3$, the corresponding varying factor $b=0.5$, and $\beta$ is chosen as $\beta = 3$, the corresponding varying factor $c_1 = 0.2$.
Therefore, the model is heteroscedastic, where the intercept term is varying across all the quantile levels, but the slope coefficient of predictor is not.
$\beta(\tau_{n,i}) $ varies for the quantiles $\{0.1,~0.2,~0.3,~0.4\}$, but remains a constant for the quantiles $\{0.5,~0.6,~0.7,~0.8,~0.9\}$.
\end{example}

\begin{table}[htbp]
\small
\caption{\small{The MedSE of coefficients in Example 3.}}
\center
    \begin{tabular}{lrrrrrrrrr} \hline \hline
                & \multicolumn{9}{c}{$\tau$ } \\
                & 0.1   & 0.2   & 0.3   & 0.4   & 0.5   & 0.6   & 0.7   & 0.8   & 0.9 \\ \hline
          & \multicolumn{9}{c}{$\lambda=0.2$} \\
    RQ & 0.3528  & 0.2100  & 0.1968  & 0.1715  & 0.1655  & 0.1730  & 0.2053  & 0.2288  & 0.3421  \\
    FL    & 0.1525  & 0.1176  & 0.1119  & 0.1182  & 0.1159  & 0.1129  & 0.1188  & 0.1419  & 0.1797  \\
    FAL   & 0.1684  & 0.1224  & 0.1202  & 0.1229  & 0.1167  & 0.1106  & 0.1136  & 0.1325  & 0.1809  \\
    FS    & 0.2295  & 0.1713  & 0.1409  & 0.1451  & 0.1230  & 0.1418  & 0.1288  & 0.1752  & 0.2187  \\
    FAS   & 0.2050  & 0.1558  & 0.1378  & 0.1319  & 0.1299  & 0.1333  & 0.1224  & 0.1523  & 0.1975  \\
          &       &       &       &       &       &       &       &       &  \\
          & \multicolumn{9}{c}{$\lambda=0.5$} \\
    RQ & 0.3437  & 0.2450  & 0.1589  & 0.1630  & 0.1605  & 0.1901  & 0.1848  & 0.2535  & 0.4139  \\
    FL    & 0.1418  & 0.1153  & 0.0962  & 0.1095  & 0.1015  & 0.1009  & 0.1062  & 0.1217  & 0.1637  \\
    FAL   & 0.1411  & 0.1023  & 0.0979  & 0.1023  & 0.0999  & 0.1061  & 0.1037  & 0.1193  & 0.1442  \\
    FS    & 0.1796  & 0.1272  & 0.1236  & 0.1328  & 0.1257  & 0.1353  & 0.1270  & 0.1587  & 0.2367  \\
    FAS   & 0.1878  & 0.1241  & 0.1213  & 0.1247  & 0.1220  & 0.1225  & 0.1141  & 0.1611  & 0.2062  \\
          &       &       &       &       &       &       &       &       &  \\
          & \multicolumn{9}{c}{$\lambda=0.8$} \\
    RQ & 0.5329  & 0.3991  & 0.3059  & 0.3229  & 0.3667  & 0.3626  & 0.3682  & 0.4810  & 0.6016  \\
    FL    & 0.1641  & 0.1564  & 0.1192  & 0.1224  & 0.1148  & 0.1263  & 0.1314  & 0.1298  & 0.1518  \\
    FAL   & 0.1678  & 0.1239  & 0.1163  & 0.1203  & 0.1291  & 0.1262  & 0.1255  & 0.1329  & 0.1551  \\
    FS    & 0.2170  & 0.1963  & 0.1527  & 0.1546  & 0.1463  & 0.1437  & 0.1345  & 0.1445  & 0.1768  \\
    FAS   & 0.2310  & 0.1957  & 0.1552  & 0.1780  & 0.1628  & 0.1457  & 0.1549  & 0.1719  & 0.2528  \\ \hline
    \end{tabular}%
  \label{tab:7}%
\end{table}%

The results are presented in Table {\ref{tab:7}}. Apparently, FAL and FL methods are better than FS and FAS methods when detecting the
the insignificant of partial-varying coefficients.

\begin{example}
The model in this example has bivariate predictors. The data are generated from
\begin{equation}
\label{exam:4}
Y_i = \alpha(\tau_{n,i}) + \lambda(\tau_{n,i}) U_i + {{\beta_1}}(\tau_{n,i}) {X}_{i1} + {{\beta_2}}(\tau_{n,i}) {X}_{i2} + e_i, ~~i = 1, \cdots, n,
\end{equation}
where $\tau_{n,i}$ is still randomly from $\mathbb{S}_{\tau}$, $\alpha(\tau_{n,i}) = \alpha + bF_n^{-1}(\tau_{n,i}) $, $\lambda(\tau_{n,i}) = \lambda + c_0 F_n^{-1}(\tau_{n,i})$, $\beta_1(\tau_{n,i}) = \beta_1 + c_1 F_n^{-1}(\tau_{n,i})$, $\beta_2(\tau_{n,i}) = \beta_2 + c_2 F_n^{-1}(\tau_{n,i})$. In Example 1, we discussed the case with varying intercept term. Without loss of generality, we consider $\alpha=b = 0$.
$F_n^{-1}(\tau_{n,i})$ is the $\tau_{n,i}$ quantile of distribution $N(0,1)$. $e_i \stackrel{i.i.d}\sim N(0, 1) $ and $X_{i1},~X_{i2}$ are generated from $U(0, 1)$ independently.
If $c_0 = c_1 = c_2 = 0$, (\ref{exam:4}) is a homoscedastic model with the constant spatial parameter $\lambda(\tau_{n,i}) = \lambda$ and constant quantile slope $\beta_1(\tau_{n,i}) = \beta_1$, $\beta_2(\tau_{n,i}) = \beta_2$. However, if at least one of $c_0,~c_1,~c_2$ is not zero,
(\ref{exam:4}) becomes a heteroscedastic model with $\lambda(\tau_{n,i})$, $\beta_1(\tau_{n,i})$, or $\beta_2(\tau_{n,i})$ varying in $\tau_{n,i}$. The sample size $n$ is chosen as $n = 120$. $\lambda$ is chosen as $\lambda = 0.2,~0.5,~0.8$. $\beta_1$ is chosen as $\beta_1 = 2$ and $\beta_2$ is chosen as $\beta_2 = 3$.
 We consider the following five settings.
\begin{itemize}
\item[] {\textrm{I}:} ~~$c_0 = 0.1, ~c_1 = 0.3, ~c_2 = 0.5$. In this case, all the parameters are dependent on quantiles.
\item[] {\textrm{II}:} ~$c_0 = 0, ~c_1 = 0.3, ~c_2 = 0.5$. In this case,  spatial lag parameter is constant and the impact $\bm{\beta}$ of the covariate $\bm{X}$ on $\bm{Y}$ is different at different quantiles.
\item[] {\textrm{III}:} $c_0 = 0.1, ~c_1 = 0, ~c_2 = 0.5$. In this case, $\beta_1(\tau_{n,i})$ stays invariant for all quantiles. The spatial lag parameter and $\beta_2(\tau_{n,i})$ vary across quantiles.
\item[] {\textrm{IV}:} $c_0 = 0, ~c_1 = 0, ~c_2 = 0.5$. Only $\beta_2(\tau_{n,i})$ varies across quantiles. The spatial lag parameter and $\beta_1(\tau_{n,i})$ stay invariant for all quantiles.
\item[] {\textrm{V}:} $c_0 = 0, ~c_1 = 0, ~c_2 = 0$. In this case, all the parameters stay invariant for all quantiles.
\end{itemize}

\end{example}

\begin{table}[htbp]
\small
\caption{\small{The MedSE of coefficients for five settings with $n=120$ and $\lambda = 0.2$ in Example 4.}}
\center
    \begin{tabular}{lrrrrrrrrr} \hline \hline
                & \multicolumn{9}{c}{$\tau$ } \\
                & 0.1   & 0.2   & 0.3   & 0.4   & 0.5   & 0.6   & 0.7   & 0.8   & 0.9 \\ \hline
          & \multicolumn{9}{c}{I: $c_0 = 0.1, ~c_1 = 0.3, ~c_2 = 0.5$} \\
    RQ    & 0.8033  & 0.5007  & 0.4440  & 0.3441  & 0.3685  & 0.4010  & 0.4801  & 0.5829  & 0.9362  \\
    FL    & 0.5635  & 0.3454  & 0.2809  & 0.2680  & 0.2609  & 0.3115  & 0.3901  & 0.5320  & 0.9175  \\
    FAL   & 0.5400  & 0.3540  & 0.2742  & 0.2517  & 0.2673  & 0.3117  & 0.3973  & 0.5431  & 0.8370  \\
    FS    & 0.6291  & 0.3911  & 0.3269  & 0.2983  & 0.3160  & 0.3563  & 0.4243  & 0.5442  & 0.8868  \\
    FAS   & 0.5681  & 0.3778  & 0.2781  & 0.2921  & 0.2880  & 0.3439  & 0.3998  & 0.5227  & 0.8777  \\
          & \multicolumn{9}{c}{II: $c_0 = 0, ~c_1 = 0.3, ~c_2 = 0.5$} \\
    RQ    & 0.7965  & 0.4997  & 0.3366  & 0.3215  & 0.3160  & 0.3387  & 0.3817  & 0.5284  & 0.9532  \\
    FL    & 0.5703  & 0.3350  & 0.2513  & 0.2328  & 0.2434  & 0.2684  & 0.3385  & 0.5158  & 0.8735  \\
    FAL   & 0.5484  & 0.3407  & 0.2488  & 0.2115  & 0.2301  & 0.2811  & 0.3452  & 0.5295  & 0.8813  \\
    FS    & 0.6624  & 0.4086  & 0.2847  & 0.2547  & 0.2562  & 0.3110  & 0.3489  & 0.4928  & 0.8329  \\
    FAS   & 0.6230  & 0.3879  & 0.2687  & 0.2352  & 0.2467  & 0.2891  & 0.3541  & 0.5090  & 0.8444  \\
          & \multicolumn{9}{c}{III: $c_0 = 0.1, ~c_1 = 0, ~c_2 = 0.5$} \\
    RQ    & 0.7395  & 0.4766  & 0.4014  & 0.3721  & 0.3785  & 0.4070  & 0.4493  & 0.4778  & 0.7583  \\
    FL    & 0.4862  & 0.3034  & 0.2887  & 0.2640  & 0.2823  & 0.2970  & 0.3362  & 0.4632  & 0.7052  \\
    FAL   & 0.4786  & 0.3280  & 0.2777  & 0.2664  & 0.2713  & 0.2903  & 0.3579  & 0.4379  & 0.6769  \\
    FS    & 0.5622  & 0.3996  & 0.3191  & 0.3172  & 0.3189  & 0.3366  & 0.3870  & 0.4560  & 0.7231  \\
    FAS   & 0.5506  & 0.3748  & 0.2904  & 0.2898  & 0.2849  & 0.3069  & 0.3847  & 0.4518  & 0.7210  \\
          & \multicolumn{9}{c}{IV: $c_0 = 0, ~c_1 = 0, ~c_2 = 0.5$} \\
    RQ    & 0.6641  & 0.4541  & 0.3526  & 0.3057  & 0.2977  & 0.3297  & 0.3825  & 0.4640  & 0.8274  \\
    FL    & 0.4838  & 0.3007  & 0.2290  & 0.2001  & 0.2042  & 0.2490  & 0.2973  & 0.4145  & 0.6819  \\
    FAL   & 0.5098  & 0.3284  & 0.2502  & 0.2224  & 0.2131  & 0.2480  & 0.3072  & 0.4176  & 0.6696  \\
    FS    & 0.5654  & 0.3640  & 0.2986  & 0.2500  & 0.2454  & 0.2688  & 0.3275  & 0.4427  & 0.7348  \\
    FAS   & 0.5460  & 0.3213  & 0.2666  & 0.2291  & 0.2283  & 0.2562  & 0.3048  & 0.4378  & 0.7472  \\
          & \multicolumn{9}{c}{V: $c_0 = 0, ~c_1 = 0, ~c_2 = 0$} \\
    RQ    & 0.5055  & 0.3710  & 0.3219  & 0.2688  & 0.2662  & 0.2597  & 0.3078  & 0.3604  & 0.5291  \\
    FL    & 0.2027  & 0.2005  & 0.1919  & 0.1812  & 0.1860  & 0.1935  & 0.1886  & 0.1922  & 0.1933  \\
    FAL   & 0.2253  & 0.2043  & 0.2025  & 0.1890  & 0.1890  & 0.1816  & 0.1919  & 0.1897  & 0.2037  \\
    FS    & 0.3594  & 0.2992  & 0.2432  & 0.2233  & 0.2109  & 0.2282  & 0.2352  & 0.2740  & 0.3784  \\
    FAS   & 0.2953  & 0.2764  & 0.2317  & 0.2079  & 0.1943  & 0.1992  & 0.2121  & 0.2371  & 0.2996  \\ \hline
    \end{tabular}%
  \label{tab:8}%
\end{table}%

\begin{table}[htbp]
\small
\caption{\small{The MedSE of coefficients for five settings with $n=120$ and $\lambda = 0.5$ in Example 4.}}
\center
    \begin{tabular}{lrrrrrrrrr} \hline \hline
                & \multicolumn{9}{c}{$\tau$ } \\
                & 0.1   & 0.2   & 0.3   & 0.4   & 0.5   & 0.6   & 0.7   & 0.8   & 0.9 \\ \hline
          & \multicolumn{9}{c}{I: $c_0 = 0.1, ~c_1 = 0.3, ~c_2 = 0.5$} \\
    RQ    & 0.8961  & 0.6095  & 0.4799  & 0.5328  & 0.5204  & 0.5220  & 0.6176  & 0.8057  & 1.0785  \\
    FL    & 0.5607  & 0.3738  & 0.2925  & 0.3015  & 0.3251  & 0.3823  & 0.4554  & 0.6125  & 0.9551  \\
    FAL   & 0.5350  & 0.3600  & 0.2950  & 0.3014  & 0.3596  & 0.3886  & 0.4531  & 0.6334  & 0.9590  \\
    FS    & 0.6476  & 0.4382  & 0.3265  & 0.3375  & 0.3315  & 0.4132  & 0.5337  & 0.6160  & 0.9562  \\
    FAS   & 0.6438  & 0.4060  & 0.3355  & 0.3262  & 0.3330  & 0.3882  & 0.5112  & 0.5784  & 0.8777  \\
          & \multicolumn{9}{c}{II: $c_0 = 0, ~c_1 = 0.3, ~c_2 = 0.5$} \\
    RQ    & 0.8822  & 0.5761  & 0.3761  & 0.3572  & 0.3865  & 0.3471  & 0.4494  & 0.6523  & 1.1157  \\
    FL    & 0.5643  & 0.3087  & 0.2390  & 0.2412  & 0.2531  & 0.2639  & 0.3508  & 0.5010  & 0.8318  \\
    FAL   & 0.5683  & 0.3172  & 0.2429  & 0.2402  & 0.2458  & 0.2643  & 0.3314  & 0.4850  & 0.8504  \\
    FS    & 0.6297  & 0.3754  & 0.2775  & 0.2747  & 0.2771  & 0.3062  & 0.3744  & 0.5050  & 0.8529  \\
    FAS   & 0.6531  & 0.3594  & 0.2730  & 0.2630  & 0.2814  & 0.2854  & 0.3809  & 0.5090  & 0.8954  \\
          & \multicolumn{9}{c}{III: $c_0 = 0.1, ~c_1 = 0, ~c_2 = 0.5$} \\
    RQ    & 0.8026  & 0.5894  & 0.5010  & 0.5291  & 0.5281  & 0.5107  & 0.6593  & 0.7243  & 1.0619  \\
    FL    & 0.4746  & 0.3600  & 0.3301  & 0.3232  & 0.3280  & 0.3695  & 0.4449  & 0.5231  & 0.7656  \\
    FAL   & 0.4541  & 0.3330  & 0.3369  & 0.3434  & 0.3431  & 0.3491  & 0.4395  & 0.5630  & 0.7463  \\
    FS    & 0.5348  & 0.3928  & 0.3539  & 0.3607  & 0.3661  & 0.3705  & 0.4764  & 0.5533  & 0.7249  \\
    FAS   & 0.5453  & 0.4206  & 0.3537  & 0.3585  & 0.3572  & 0.3793  & 0.5099  & 0.5410  & 0.7075  \\
          & \multicolumn{9}{c}{IV: $c_0 = 0, ~c_1 = 0, ~c_2 = 0.5$} \\
    RQ    & 0.8468  & 0.5404  & 0.4083  & 0.3515  & 0.3494  & 0.3177  & 0.4899  & 0.6032  & 1.0036  \\
    FL    & 0.4715  & 0.3083  & 0.2475  & 0.2147  & 0.2432  & 0.2602  & 0.3104  & 0.4143  & 0.7018  \\
    FAL   & 0.4920  & 0.3237  & 0.2312  & 0.2187  & 0.2350  & 0.2416  & 0.3092  & 0.4243  & 0.6637  \\
    FS    & 0.5684  & 0.3638  & 0.3157  & 0.2682  & 0.2590  & 0.2714  & 0.3340  & 0.4387  & 0.7120  \\
    FAS   & 0.5145  & 0.3572  & 0.2896  & 0.2353  & 0.2584  & 0.2719  & 0.3260  & 0.4204  & 0.6686  \\
          & \multicolumn{9}{c}{V: $c_0 = 0, ~c_1 = 0, ~c_2 = 0$} \\
    RQ    & 0.5819  & 0.3920  & 0.3408  & 0.3645  & 0.3314  & 0.3042  & 0.3463  & 0.472054 & 0.706793 \\
    FL    & 0.2112  & 0.2069  & 0.1998  & 0.2012  & 0.1990  & 0.2086  & 0.2105  & 0.202647 & 0.212494 \\
    FAL   & 0.2280  & 0.2116  & 0.2059  & 0.2092  & 0.2079  & 0.2020  & 0.2082  & 0.1997  & 0.2219  \\
    FS    & 0.3151  & 0.2753  & 0.2349  & 0.2382  & 0.2285  & 0.2302  & 0.2544  & 0.2795  & 0.3711  \\
    FAS   & 0.2996  & 0.2467  & 0.2249  & 0.2117  & 0.2164  & 0.2173  & 0.2382  & 0.2670  & 0.3311  \\ \hline
    \end{tabular}%
  \label{tab:9}%
\end{table}%

\begin{table}[htbp]
\small
\caption{\small{The MedSE of coefficients for five settings with $n=120$ and $\lambda = 0.8$ in Example 4.}}
\center
    \begin{tabular}{lrrrrrrrrr} \hline \hline
                & \multicolumn{9}{c}{$\tau$ } \\
                & 0.1   & 0.2   & 0.3   & 0.4   & 0.5   & 0.6   & 0.7   & 0.8   & 0.9 \\ \hline
          & \multicolumn{9}{c}{I: $c_0 = 0.1, ~c_1 = 0.3, ~c_2 = 0.5$} \\
    RQ    & 1.9994  & 1.6943  & 1.6610  & 1.7097  & 1.8468  & 2.2542  & 2.2902  & 3.1073  & 5.2059  \\
    FL    & 0.8370  & 0.6886  & 0.6897  & 0.6765  & 0.7409  & 0.8420  & 0.8532  & 1.0659  & 1.4000  \\
    FAL   & 0.8044  & 0.6621  & 0.6261  & 0.6819  & 0.7391  & 0.8003  & 0.9006  & 1.0408  & 1.4038  \\
    FS    & 0.8356  & 0.6922  & 0.6312  & 0.6736  & 0.7508  & 0.8771  & 0.9533  & 1.0263  & 1.3788  \\
    FAS   & 0.9016  & 0.7883  & 0.7651  & 0.7855  & 0.8155  & 0.8308  & 0.9737  & 1.1282  & 1.4742  \\
          & \multicolumn{9}{c}{II: $c_0 = 0, ~c_1 = 0.3, ~c_2 = 0.5$} \\
    RQ    & 1.5138  & 1.0281  & 0.7203  & 0.7340  & 0.6468  & 0.6783  & 0.8807  & 1.0346  & 1.7676  \\
    FL    & 0.5750  & 0.3693  & 0.2583  & 0.2652  & 0.2650  & 0.3129  & 0.4106  & 0.5086  & 0.8215  \\
    FAL   & 0.6591  & 0.3779  & 0.2819  & 0.2785  & 0.2630  & 0.2892  & 0.3832  & 0.5095  & 0.8616  \\
    FS    & 0.6431  & 0.4324  & 0.3196  & 0.2936  & 0.2824  & 0.2995  & 0.3876  & 0.4986  & 0.8399  \\
    FAS   & 0.7495  & 0.4498  & 0.3583  & 0.3119  & 0.2852  & 0.2944  & 0.3884  & 0.5192  & 0.8596  \\
          & \multicolumn{9}{c}{III: $c_0 = 0.1, ~c_1 = 0, ~c_2 = 0.5$} \\
    RQ    & 1.9433  & 1.9610  & 1.7021  & 1.7796  & 1.9362  & 2.2630  & 2.7850  & 3.6951  & 6.2834  \\
    FL    & 0.7960  & 0.7297  & 0.7047  & 0.7242  & 0.7972  & 0.8743  & 0.9738  & 1.1403  & 1.5335  \\
    FAL   & 0.7290  & 0.7189  & 0.6964  & 0.7204  & 0.7888  & 0.8812  & 0.9843  & 1.1594  & 1.5945  \\
    FS    & 0.9143  & 0.7436  & 0.7693  & 0.7740  & 0.8588  & 0.9101  & 1.0490  & 1.1737  & 1.5680  \\
    FAS   & 0.8928  & 0.8439  & 0.8130  & 0.8313  & 0.9437  & 0.9508  & 1.0541  & 1.2682  & 1.7148  \\
          & \multicolumn{9}{c}{IV: $c_0 = 0, ~c_1 = 0, ~c_2 = 0.5$} \\
    RQ    & 1.4731  & 0.9350  & 0.7291  & 0.6092  & 0.6399  & 0.6692  & 0.8357  & 0.8907  & 1.6522  \\
    FL    & 0.4796  & 0.3202  & 0.2616  & 0.2526  & 0.2514  & 0.2806  & 0.3610  & 0.4515  & 0.6561  \\
    FAL   & 0.4962  & 0.3408  & 0.2992  & 0.2399  & 0.2455  & 0.2583  & 0.3529  & 0.4630  & 0.7212  \\
    FS    & 0.5578  & 0.4001  & 0.3301  & 0.2807  & 0.2873  & 0.2827  & 0.3509  & 0.4725  & 0.6851  \\
    FAS   & 0.6292  & 0.4584  & 0.3369  & 0.2831  & 0.2725  & 0.3005  & 0.3639  & 0.4681  & 0.7756  \\
          & \multicolumn{9}{c}{V: $c_0 = 0, ~c_1 = 0, ~c_2 = 0$} \\
    RQ    & 1.0086  & 0.7686  & 0.7119  & 0.5718  & 0.5442  & 0.5357  & 0.6989  & 0.7088  & 1.3848  \\
    FL    & 0.2529  & 0.2472  & 0.2474  & 0.2464  & 0.2251  & 0.2223  & 0.2259  & 0.2436  & 0.2665  \\
    FAL   & 0.2623  & 0.2417  & 0.2337  & 0.2428  & 0.2214  & 0.2143  & 0.2301  & 0.2301  & 0.2563  \\
    FS    & 0.3439  & 0.3035  & 0.2886  & 0.2664  & 0.2389  & 0.2287  & 0.2653  & 0.3195  & 0.4046  \\
    FAS   & 0.3497  & 0.2746  & 0.2906  & 0.2615  & 0.2562  & 0.2268  & 0.2890  & 0.3202  & 0.4469  \\ \hline
    \end{tabular}%
  \label{tab:10}%
\end{table}%

Tabel \ref{tab:8} -  \ref{tab:10} reveals the MedSE of coefficients in the bivariate case with $\lambda= 0.2,~0.5,~0.8$, and $F_n$ is chosen to be standard normal. We consider five settings and the results show that our proposed approaches yield smaller MedSE than RQ method. The MedSE become smaller with more constant coefficients. And with the increasing of spatial lag parameters, the MedSE become greater.

\section{Real data}
We apply our proposed methods to analyze a classical crime dataset, which were originally from Anselin, L.(1988). The dataset contains 49 observations and 22 variables, in which 16 variables are the ID values. The response variable is residential burglaries and vehicle thefts per thousand households in the neighborhood (CRIME). There are 5 covariates: housing value (HOVAL), household income (INC), open space in neighborhood (OPEN), percentage housing units without plumbin (PLUMB), distance to CBD (DISCBD). Our purpose is to investigate the effects of covariates on CRIME in Columbus.


By classical SAR model, only two covariates are significance, that is, HOVAL and INC. The estimated SQAR model is as follows,
\begin{equation}
\label{real data}
{\rm{CRIME}} = {\rm{Constant}}(\tau) + \lambda(\tau)~W \cdot {\rm{CRIME}} + \beta_1(\tau)~{\rm{HOVAL}} +  \beta_2(\tau)~{\rm{INC}} + {\bm \varepsilon}
\end{equation}

We employ the classical quantile regression method(RQ), FAS, FAL to the dataset and $\tau = \{0.1, 0.2, \cdots, 0.9\}$. The results are shown in Figure 1.

\begin{figure}
\centering
{%
\resizebox*{6.5cm}{!}{\includegraphics{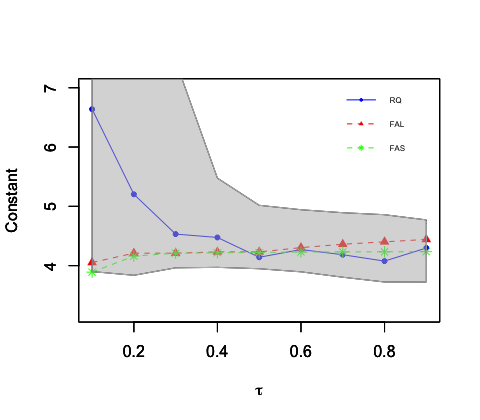}}}\hspace{6pt}
\hfill
{%
\resizebox*{6.5cm}{!}{\includegraphics{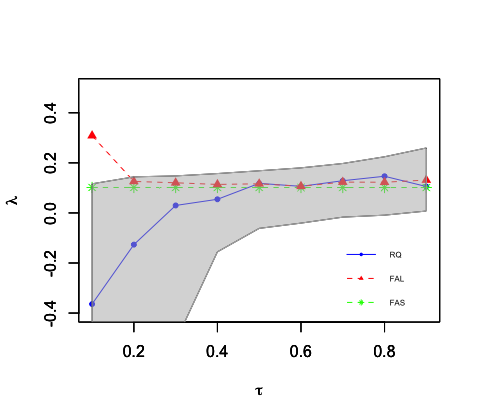}}}\hspace{6pt}
\vfill
{%
\resizebox*{6.5cm}{!}{\includegraphics{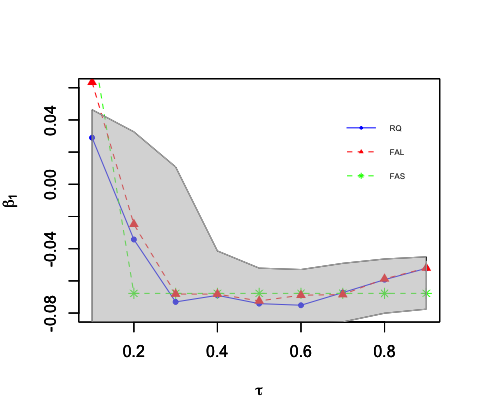}}}\hspace{6pt}
\hfill
{%
\resizebox*{6.5cm}{!}{\includegraphics{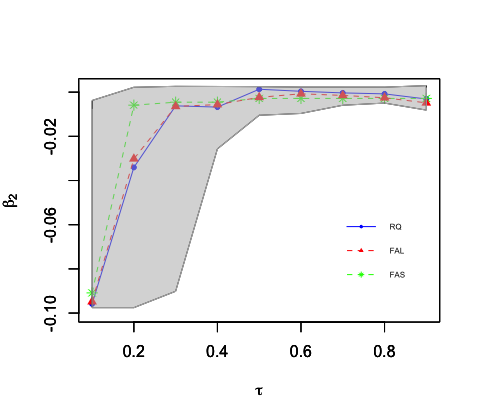}}}\hspace{6pt}
\vfill
\vfill\caption{{Estimated quantile coefficients: $Constant$, $\lambda$, $\beta_1$ and $\beta_2$ from RQ (solid line),
FAL (dashed line with triangle), and FAS (dashed line with stars) methods. Shaded areas are the $90\%$ confident bands from regular RQ method.}}
\label{fig.1}
\end{figure}

It is interesting that the constant estimators by FAL or FAS are almost constant separately on nine quantile levels in Figure but the estimators of traditional RQ method still vary.
The second subfig show that $\lambda$ estimators at nine quantiles by FAL only vary significantly from $\tau = 0.1$ to $\tau = 0.2$, and not very significantly across all the nine quantile levels. To further verify the shrinkage results, we conduct hypothesis tests to check the constancy of slope coefficients. The Wald test for the equality of $\lambda$ on quantile levels $\tau = 0.1, \cdots, 0.9$ returns a $p$-value of 0.01147, and on quantile levels $\tau = 0.2, \cdots, 0.9$ returns a $p$-value of 0.8663. The two $p$-values imply the effect of spatial lag variables vary significantly across the all the nine quantile levels under significance level $\alpha = 0.05$, but even nearly keep constant across $\tau = 0.2, \cdots, 0.9$. This agrees with the results from fused adaptive LASSO. For fused adaptive Lasso, the coefficients differences are penalized individually, so FAL could recognize the situation where the quantile slope coefficients appear constant in certain quantile regions, but vary in others.

On the other hand, for the covariate HOVAL and INC, the equality test on the
quantile coefficients at $\tau = 0.1, \cdots, 0.9$ returns two $p$-values near zero, showing that at a whole the effect of HOVAL and INC vary across the nine quantile levels. The equality test at $\tau = 0.3, \cdots, 0.9$ return $p$-values 0.6438 and 0.2730 separately, which implying the effect of HOVAL and INC almost keep constant across $\tau = 0.3, \cdots, 0.9$. This is consist with the results by FAL. FAS method tends to shrink  all the quantile coefficients to a constant so it couldn't distinguish the single difference of the coefficients when between $\tau = 0.2 $ and $\tau = 0.3$.

\section{Conclusion and discussion}

In this paper, we suggest two new fused penalties methods (FAL and FAS) to shrinkage the interquantile parameters in spatial quantile autoregressive model. Through our method, we realized estimating the parameters and identifying the commonality among parameters in different quantiles at the same time. Lots of simulation experiments and real data analysis also show the greater performance of the new methods than that of traditional estimation method for spatial quantile autoregressive model.

In reality, however, there exist irrelevant variables. Most popular methods to select important predictors are also penalities approaches. How to combine the two penalization methods together in spatial quantile autoregressive model is an interesting research.

\begin{appendix}
\section*{Appendix}
\setcounter{subsection}{0}
\setcounter{equation}{0}

\renewcommand{\thesubsection}{A.\arabic{subsection}}
\renewcommand{\theequation}{A.\arabic{equation}}

\subsection{Proof of Theorem 2.1}
\begin{proof}
Assumption (A1) - (A4) are the basic assmptions for SQAR model and make sure the existence of estimators.

In Assumption (A5), the sequence $\{(u_i, v_i, \mathbb{X}_i)\}$ is independent and identically distributed, that is
\begin{align*}
&P(u_i, v_i) = P(u_i)P( v_i); \\
&P( u_i,  \mathbb{X}_i) = P( u_i)P(  \mathbb{X}_i); \\
&P( v_i,  \mathbb{X}_i) = P( v_i)P( \mathbb{X}_i). \\
\end{align*}
Then
\begin{align*}
P( u_i, {\bm{V}}_i ) &= P( u_i, [1, {\bm{X}}_i, ({\bm{WX}})_i]^{T} ) \\
&= P( u_i, [{\mathbb{X}}_i, {\bm{W}_i}{\bm{X}}]^{T} )  \\
&= P(u_i)P([{\mathbb{X}}_i, {\bm{W}_i}{\bm{X}}]^{T}) \\
&= P( u_i)P( {\bm{V}}_i ) ,
\end{align*}
where ${\bm{W}}_i$ is the $i$th row vector of ${\bm{W}}$.
Similarly, $P( v_i, {\bm{V}}_i ) = P( v_i)P( {\bm{V}}_i ) $. So the sequence $\{(u_i, v_i, {\bm{V}}_i)\}$ is independent and identically distributed. It is satisfied to the Assumption 1 in \cite{Kim-M}.

Denote the $X_{ij}$ as the $(i,j)$th element of $\bm{X}$, ${\bm{X}}_{\cdot j}$ as the $j$th column of of $\bm{X}$.
Then
\begin{align*}
E(||{\bm{V}}_i||^3) &= E(||[1, {\bm{X}}_i, {\bm{W}}_i{\bm{X}}]^T||^3) \\
&= E
\begin{bmatrix}
&1 + {\bm{X}}_i {\bm{X}}_i^T + {\bm{W}}_i{\bm{X}}({\bm{W}}_i{\bm{X}})^T \\
&{\bm{X}}_i^T + {\bm{X}}_i^T{\bm{X}}_i {\bm{X}}_i^T + {\bm{X}}_i^T{\bm{W}}_i{\bm{X}}({\bm{W}}_i{\bm{X}})^T \\
&({\bm{W}}_i{\bm{X}})^T + ({\bm{W}}_i{\bm{X}})^T{\bm{X}}_i {\bm{X}}_i^T + ({\bm{W}}_i{\bm{X}})^T{\bm{W}}_i{\bm{X}}({\bm{W}}_i{\bm{X}})^T \\
\end{bmatrix}
\qquad
\end{align*}
According to Assumption (A6), the second geometric moment of ${\bm X}$ and the third geometric moment of ${\mathbb X}_i$ are finite, so $E(||{\bm{V}}_i||^3) < \infty$, that is the third geometric moment of ${\bm V}_i$ is finite. The Assumption 2(i) in \cite{Kim-M} is proved.

Under Assumption (A8),
\begin{align*}
A_1 &= E\{g_1(0|{\bm{V}}_i){\bm{V}}_i {\bm{V}}_i^{T}\}\\
&= E\{g_1(0|{\bm{V}}_i)\}{\bm{V}}_i {\bm{V}}_i^{T} \\
&= E\{g_1(0|{\bm{V}}_i)\}[1, {\bm{X}}_i, {\bm{W}}_i{\bm{X}}][1, {\bm{X}}_i, {\bm{W}}_i{\bm{X}}]^T \\
&= E\{g_1(0|{\mathbb{X}}_i)\}({\mathbb{X}}_i{\mathbb{X}}_i^T +{\bm{W}}_i{\bm{X}} ({\bm{W}}_i{\bm{X}})^T )\\
& = B_1 + B_1^{'}
\end{align*}
where $B_1^{'} = E\{g_1(0|{\mathbb{X}}_i)\}{\bm{W}}_i{\bm{X}} ({\bm{W}}_i{\bm{X}})^T $.
According to Assumption (A8), $B_1$ and $B_1^{'}$ are finite and positive definite, so $A_1$ is also finite and positive definite.
Similarly,  $A_2$ is finite and positive definite. The Assumption 2(iii) in \cite{Kim-M} is proved.

Under Assumption (A7) and (A9), the Assumption 2(ii) and 2(iv) are easily satisfied.

Thus, according to the Proposition 2 in \cite{Kim-M}, the Theorem 2.1 is proved.
\end{proof}

According to \cite{Jiang-2014}, more general cases ($p >1$) have the similar propositions and theorems. So for ease to illustrate, we just consider $p=1$ in the next subsections.

\subsection{Proof of Proposition 2.1}

\begin{proof}
\noindent 
Without loss of generality, we assume the quantile slopes ${\bm{\beta}= {\beta_1}}$ vary for the first $s_1$ ($s_1 <K$ )quantiles, but remain constant for the remaining $(K-s_1)$ quantile levels.

At first, we consider about minimizing the following formula
$$
    L_{n}(\boldsymbol{\delta})=\sum_{k=1}^{K} \sum_{i=1}^{n}\left[\rho_{\tau_{k}}\left\{Y_{i}-\mathbf{Z}_{k i, \mathcal{A}}^{T}\left(\boldsymbol{\theta}_{\mathcal{A}, 0} + \frac{1}{\sqrt{n}} \boldsymbol{\delta}\right)\right\}-\rho_{\tau_{k}}\left(Y_{i}-\mathbf{Z}_{k i, \mathcal{A}}^{T} \boldsymbol{\theta}_{\mathcal{A}, 0}\right)\right].
$$
The minimizer $\hat{\boldsymbol{\delta}}$ is $n^{1 / 2}\left(\hat{\boldsymbol{\theta}}_{\mathcal{A}}-\boldsymbol{\theta}_{\mathcal{A}, 0}\right)$ and $\boldsymbol{\delta} \in \mathbb{R}^{K+s}$ is a bounded vector. Under the identity in \cite{Knight}, we obtain
$$
    \rho_{\tau}(r-s)-\rho_{\tau}(r)=-s\{\tau-I(r<0)\}+\int_{0}^{s}\{I(r \leq t)-I(r \leq 0)\} d t.
$$\\
Then, $L_{n}$ can be written as
\begin{align*}
    L_{n}(\boldsymbol{\delta}) &=-n^{-1 / 2} \sum_{k=1}^{K} \sum_{i=1}^{n} \mathbf{Z}_{k i, \mathcal{A}}^{T}\left\{\tau_{k}-I\left(y_{i}-\mathbf{Z}_{k i, \mathcal{A}}^{T} \boldsymbol{\theta}_{\mathcal{A}, 0}<0\right)\right\} \boldsymbol{\delta} + \\
    &\sum_{k=1}^{K} \sum_{i=1}^{n} \int_{0}^{n^{-1 / 2} \mathbf{Z}_{k i, \mathcal{A}}^{T} \boldsymbol{\delta}}\left\{I\left(Y_{i}-\mathbf{Z}_{k i, \mathcal{A}}^{T} \boldsymbol{\theta}_{\mathcal{A}, 0} \leq t\right)-I\left(Y_{i}-\mathbf{Z}_{k i, \mathcal{A}}^{T} \boldsymbol{\theta}_{\mathcal{A}, 0} \leq 0\right)\right\} d t \\
    & \triangleq M_{n}^{(1)}(\boldsymbol{\delta}) + \sum_{k=1}^{K} V_{n}^{(k)}.
\end{align*}\\
Denote $M_{n}^{(2)}(\boldsymbol{\delta})=\sum_{k=1}^{K} \frac{1}{2} \boldsymbol{\delta}^{T}\left\{n^{-1} \sum_{i=1}^{n} f_{i}\left(\mathbf{Z}_{k i, \mathcal{A}}^{T} \boldsymbol{\theta}_{\mathcal{A}, 0}\right) \mathbf{Z}_{k i, \mathcal{A}} \mathbf{Z}_{k i, \mathcal{A}}^{T}\right\} \boldsymbol{\delta}.$
we can rewrite $L_{n}$ as
\begin{align*}
    L_{n}(\boldsymbol{\delta}) & = M_{n}^{(1)}(\boldsymbol{\delta}) + M_{n}^{(2)}(\boldsymbol{\delta}) + \left[\sum_{k=1}^{K} V_{n}^{(k)}-M_{n}^{(2)}(\boldsymbol{\delta}) \right]\\
    & \triangleq M_{n}^{(1)}(\boldsymbol{\delta}) + M_{n}^{(2)}(\boldsymbol{\delta}) + M_{n}^{(3)}(\boldsymbol{\delta}).
\end{align*}\\
According to the Assumptions (A10)-(A11), we need to show  $V_{n}^{(k)}$'s mean and variance in detail.
\begin{align*}
E\left\{V_{n}^{(k)}\right\}
&=\sum_{i=1}^{n} \int_{0}^{n^{-1 / 2} \mathbf{Z}_{k i, \mathcal{A}}^{T} \boldsymbol{\delta}}\left\{F_{i}\left(\mathbf{Z}_{k i, \mathcal{A}}^{T} \boldsymbol{\theta}_{\mathcal{A}, 0}+t\right)-F_{i}\left(\mathbf{Z}_{k i, \mathcal{A}}^{T} \boldsymbol{\theta}_{\mathcal{A}, 0}\right)\right\} d t \\
&=\frac{1}{n} \sum_{i=1}^{n} \int_{0}^{\mathbf{Z}_{k i, \mathcal{A}}^{T} {\boldsymbol{\delta}}} n^{1 / 2}\left\{F_{i}\left(\mathbf{Z}_{k i, \mathcal{A}}^{T} \boldsymbol{\theta}_{\mathcal{A}, 0}+n^{-1 / 2} t\right)-F_{i}\left(\mathbf{Z}_{k i, \mathcal{A}}^{T} \boldsymbol{\theta}_{\mathcal{A}, 0}\right)\right\} d t \\
&=\frac{1}{n} \sum_{i=1}^{n} \int_{0}^{\mathbf{Z}_{k i, \mathcal{A}}^{T}{\boldsymbol{\delta}}} f_{i}\left(\mathbf{Z}_{k i, \mathcal{A}}^{T} \boldsymbol{\theta}_{\mathcal{A}, 0}\right) t d t+o_{p}(1) \\
&=\frac{1}{2} \boldsymbol{\delta}^{T}\left\{\frac{1}{n} \sum_{i=1}^{n} f_{i}\left(\mathbf{Z}_{k i, \mathcal{A}}^{T} \boldsymbol{\theta}_{\mathcal{A}, 0}\right) \mathbf{Z}_{k i, \mathcal{A}} \mathbf{Z}_{k i, \mathcal{A}}^{T}\right\} \boldsymbol{\delta}+o_{p}(1).
\end{align*}

\begin{align*}
\operatorname{Var}\left\{V_{n}^{(k)}\right\}
=& E\left[V_{n}^{(k)}-E\left\{V_{n}^{(k)}\right\}\right]^{2} \\
=& E\left[\sum _ { i = 1 } ^ { n } \int _ { 0 } ^ { n ^ { - 1 / 2 } \mathbf{Z}_{k i, \mathcal{A}}^{T} \boldsymbol { \delta } } \left\{I\left(y_{i}-\mathbf{Z}_{k i, \mathcal{A}}^{T} \boldsymbol{\theta}_{\mathcal{A}, 0} \leq t\right)-I\left(y_{i}-\mathbf{Z}_{k i, \mathcal{A}}^{T} \boldsymbol{\theta}_{\mathcal{A}, 0} \leq 0\right)\right.\right.\\
&\left.\left.-F_{i}\left(\mathbf{Z}_{k i, \mathcal{A}}^{T} \boldsymbol{\theta}_{\mathcal{A}, 0}+t\right)+F_{i}\left(\mathbf{Z}_{k i, \mathcal{A}}^{T} \boldsymbol{\theta}_{\mathcal{A}, 0}\right)\right\} d t\right]^{2} \\
\leq & \sum_{i=1}^{n} E[ \mid \int_{0}^{n^{-1 / 2} \mathbf{Z}_{k i, \mathcal{A}}^{T} \delta}\{I\left(y_{i}-\mathbf{Z}_{k i, \mathcal{A}}^{T} \boldsymbol{\theta}_{\mathcal{A}, 0} \leq t\right)-I\left(y_{i}-\mathbf{Z}_{k i, \mathcal{A}}^{T} \boldsymbol{\theta}_{\mathcal{A}, 0} \leq 0\right)\\
&-F_{i}\left(\mathbf{Z}_{k i, \mathcal{A}}^{T} \boldsymbol{\theta}_{\mathcal{A}, 0}+t\right)+F_{i}\left(\mathbf{Z}_{k i, \mathcal{A}}^{T} \boldsymbol{\theta}_{\mathcal{A}, 0}\right)\} d t \mid ] \times 2\left|n^{-1 / 2} \mathbf{Z}_{k i, \mathcal{A}}^{T} \boldsymbol{\delta}\right| \\
\leq & 4 n^{-1 / 2} E\left\{B_{n}^{(k)}\right\} \max _{1 \leq i \leq n}\left\|\mathbf{Z}_{k i, \mathcal{A}}\right\|\|\boldsymbol{\delta}\| .
\end{align*}
It's easy to see that for any fixed $\boldsymbol{\delta}\in \mathbb{R}^{K+s}$,
$$\operatorname{Var}\left\{V_{n}^{(k)}\right\}\rightarrow 0 .$$
Therefore, according to the Lindeberg-Feller Central Limit Theorem and Cramer-Wold Device, the first term and second of $L_{n}(\boldsymbol{\delta})$ have,
\begin{align*}
&M_{n}^{(1)}(\boldsymbol{\delta}) \stackrel{d}{\rightarrow}-\sum_{k=1}^{K} \delta^{T} \boldsymbol{S}_{\boldsymbol{k}}, \\
&M_{n}^{(2)}(\boldsymbol{\delta}) \rightarrow \frac{1}{2} \boldsymbol{\delta}^{T}\left(\sum_{k=1}^{K} {\bm{\Omega}}_{k, \mathcal{A}}\right) \boldsymbol{\delta},
\end{align*}
where $\boldsymbol{S}_{\boldsymbol{k}} \sim N\left(0, \tau_{k}\left(1-\tau_{k}\right) {\bm{\Gamma}}_{k, \mathcal{A}}\right)$.\\
As for the third term $M_{n}^{(3)}(\boldsymbol{\delta})$, we know $M_{n}^{(3)}(\cdot)$ is a convex function and $M_{n}^{(3)}(\boldsymbol{\delta}) \rightarrow 0$ for any fixed $\boldsymbol{\delta}$.
Due to the Convexity Lemma in \cite{Jiang-2013}, this pointwise convergence can be strengthened to the uniform convergence.
In other word, $h_{n}(\boldsymbol{\delta}) \rightarrow 0$ uniformly on any compact subset of $\boldsymbol{R}^{K+s}$.
In conclusion,
$$
L_{n}(\boldsymbol{\delta}) \rightarrow-\sum_{k=1}^{K} \boldsymbol{\delta}^{T} \boldsymbol{S}_{k}+\frac{1}{2} \boldsymbol{\delta}^{T}\left(\sum_{k=1}^{K} {\bm{\Omega}}_{k, \mathcal{A}}\right) \boldsymbol{\delta} .
$$
and the minimizer to $L_{n}(\boldsymbol{\delta})$, defined as $\hat{\boldsymbol{\delta}}$, follows the asymptotic normal distribution $N\left(0, {\bm\Sigma}_{\mathcal{A}}\right)$, where ${\bm{\Sigma}}_{\mathcal{A}}=\left(\sum_{k=1}^{K} {\bm{\Omega}}_{k, \mathcal{A}}\right)^{-1}\left\{\sum_{k=1}^{K} \tau_{k}\left(1-\tau_{k}\right) {\bm{\Gamma}}_{k, \mathcal{A}}\right\}\left(\sum_{k=1}^{K} {\bm{\Omega}}_{k, \mathcal{A}}\right)^{-1}$

\end{proof}

\subsection{Proof of Theorem 2.2}
Without loss of generality, we consider the oracle estimator in the setting with $p=1$. 
Then notation can then be simplified by letting $\theta_{j}$ be the $j$ th element of $\boldsymbol{\theta}$ and $\tilde{\omega}_{j} = \mid \hat{\theta}_{K+j} \mid ^{-1} $  for $j = 2,3,\dots,K$.

\subsubsection{Root-n Consistency Lemma}
\begin{lemm}[Root-n consistency of $\left.\hat{\boldsymbol{\theta}}_{F A L}\right)$]
Assume conditions $A 1-A 11$ hold, if $n^{1 / 2} \tilde{\gamma}_{1 n} \rightarrow 0$, then
 $\hat{\boldsymbol{\theta}}_{F A L}-\boldsymbol{\theta}_{0}=O_{p}\left(n^{-1 / 2}\right) .$
\end{lemm}

\begin{proof}
As in Fan and Li (2001), we only need to show that for any $\epsilon>0$, there exists a sufficiently large constant $\eta$, such that
$$
P\{\inf _{\| \mathbf{\xi}\|=\eta} Q_{1}(\theta_{0}+n^{-1 / 2} \boldsymbol{\xi})>Q_{1}(\boldsymbol{\theta}_{0})\} \geq 1-\epsilon,
$$
where
$Q_{1}(\cdot)$ is defined in (\ref{Objfunc_p1}) when $p=1$.
In section 2.3.1, we assume the slopes $\boldsymbol{\beta}_{k}$ vary for the first $s_{1}<K$ quantiles, but remain constant for the rest $(K-s_{1})$ quantile levels.
Then, we consider
\begin{align}
\label{Q_unequality}
& Q_{1}\left(\boldsymbol{\theta}_{0}+n^{-1 / 2} \boldsymbol{\xi}\right)-Q_{1}\left(\boldsymbol{\theta}_{0}\right) \\ \nonumber
=& L_{n}(\boldsymbol{\xi})+n \tilde{\gamma}_{1 n} \sum_{j=K+2}^{2 K} \tilde{w}_{j}\left(\left|\theta_{j, 0}+n^{-1 / 2} \xi_{j}\right|-\left|\theta_{j, 0}\right|\right) \\ \nonumber
=& L_{n}(\boldsymbol{\xi})+n \tilde{\gamma}_{1 n} \sum_{j=K+2}^{K+s_{1}} \tilde{w}_{j}\left(\left|\theta_{j, 0}+n^{-1 / 2} \xi_{j}\right|-\left|\theta_{j, 0}\right|\right)+n \tilde{\gamma}_{1 n} \sum_{j=K+s_{1}+1}^{2 K} \tilde{w}_{j}\left|n^{-1 / 2} \xi_{j}\right| \\ \nonumber
\geq & L_{n}(\boldsymbol{\xi})+n \tilde{\gamma}_{1 n} \sum_{j=K+2}^{K+s_{1}} \tilde{w}_{j}\left(\left|\theta_{j, 0}+n^{-1 / 2} \xi_{j}\right|-\left|\theta_{j, 0}\right|\right) ,
\end{align}
where $\tilde{w}_{j} \stackrel{p}{\rightarrow}\left|\theta_{j, 0}\right|^{-1} $ for any $\theta_{j, 0} \neq 0$.
According to the assumption $n^{1 / 2} \tilde{\gamma}_{1 n} \rightarrow 0$ and $\|\boldsymbol{\xi}\|= \eta$, we have
$$
    n \tilde{\gamma}_{1 n} \sum_{j=K+2}^{K+s_{1}} \tilde{w}_{j}\left(\left|\theta_{j, 0}+n^{-1 / 2} \xi_{j}\right|-\left|\theta_{j, 0}\right|\right)=n^{1 / 2} \tilde{\gamma}_{1 n} \sum_{j=K+2}^{K+s_{1}} \tilde{w}_{j} \xi_{j} \operatorname{sgn}\left(\theta_{j, 0}\right) \rightarrow 0
$$
uniformly in any compact set of $\mathbb{R}^{2 K}$. \\
Combined with the Proposition 2.1, the right side of inequality (\ref{Q_unequality}) is dominated by the quadratic term $M_{n}^{(2)}(\boldsymbol{\xi})$ when $\mathrm{n}$ is sufficiently large. Based on assumption (A11), we know $Q_{1}\left(\boldsymbol{\theta}_{0}+n^{-1 / 2} \boldsymbol{\xi}\right)-Q_{1}\left(\boldsymbol{\theta}_{0}\right) \geq 0$ is always holds when $\|\boldsymbol{\xi}\|= \eta$ (as $\eta \rightarrow \infty$).

\end{proof}

\subsubsection{Proof of Theorem 2.2}
\begin{proof}
~\\
\noindent \textbf{Sparsity:}\\
Under the model assumption, $\boldsymbol{\theta}$ can decomposed as $\left(\boldsymbol{\theta}_{\mathcal{A}}^{C}, \boldsymbol{\theta}_{\mathcal{A}^{C}}^{T}\right)^{T}$, where $\boldsymbol{\theta}_{\mathcal{A}} \in \mathbb{R}^{K+s_{1}}$, and $\boldsymbol{\theta}_{\mathcal{A}^{C}} \in \mathbb{R}^{K-s_{1}}$. Denote $\hat{\boldsymbol{\theta}}=\left(\hat{\boldsymbol{\theta}}_{\mathcal{A}}^{T}, \hat{\boldsymbol{\theta}}_{\mathcal{A}^{C}}^{T}\right)^{T}$ and by the way of contradiction suppose $\hat{\boldsymbol{\theta}}_{\mathcal{A}^{C}} \neq 0$. Let $\boldsymbol{\theta}^{*}$ be a vector constructed by
replacing $\hat{\boldsymbol{\theta}}_{\mathcal{A}^{C}}$ with $\boldsymbol{0}^{T}$ in $\hat{\boldsymbol{\theta}}$, that is, $\boldsymbol{\theta}^{*} = \left(\hat{\boldsymbol{\theta}}_{\mathcal{A}}^{T}, \boldsymbol{0}_{k-s_{1}}^{T}\right)^{T}$.We have
 \begin{align}
 \label{Q_1_equation}
 &Q_{1}\left(\boldsymbol{\theta}^{*}\right)-Q_{1}(\hat{\boldsymbol{\theta}}) \\ \nonumber
 =& \left\{Q_{1}\left(\boldsymbol{\theta}^{*}\right)-Q_{1}\left(\boldsymbol{\theta}_{0}\right)\right\}-\left\{Q_{1}(\hat{\boldsymbol{\theta}})-Q_{1}\left(\boldsymbol{\theta}_{0}\right)\right\} \\ \nonumber
 =& L_{n}\left\{n^{1 / 2}\left(\hat{\boldsymbol{\theta}}_{\mathcal{A}}^{T}-\boldsymbol{\theta}_{\mathcal{A}, 0}^{T}, \boldsymbol{0}_{K-s}^{T}\right)^{T}\right\}-L_{n}\left\{n^{1 / 2}\left(\hat{\boldsymbol{\theta}}-\boldsymbol{\theta}_{0}\right)\right\}-n \tilde{\gamma}_{1 n} \sum_{j=K+s+1}^{2 K} \tilde{w}_{j}\left|\hat{\theta}_{j}\right| .
 \end{align}\\
 Following the Lemma 1 of root-n consistent of $\hat{\boldsymbol{\theta}}$, the first and the second term in (\ref{Q_1_equation}) are both $O_{p}(1)$. Then, we consider the thrid term in
 \begin{align*}
     &-n \tilde{\gamma}_{1 n} \sum_{j=K+s_{1}+1}^{2 K} \tilde{w}_{j}\left|\hat{\theta}_{j}\right|\\
     =&-n \tilde{\gamma}_{1 n} n^{1 / 2} \sum_{j=K+s_{1}+1}^{2 K}\left(n^{1 / 2} \tilde{\theta}_{j}\right)^{-1}\left|\hat{\theta}_{j}\right| \rightarrow o(n^{1/2})~~~(as~~n\tilde{\gamma}_{1 n} \rightarrow \infty),
 \end{align*}
 where $\tilde{\theta}_{j}=O_{p}(1)$, when $\theta_{j, 0}=0, n^{1 / 2} $.
 Therefore, (\ref{Q_1_equation}) is dominated by the third term and $Q_{1}\left(\boldsymbol{\theta}^{*}\right)<Q_{1}(\hat{\boldsymbol{\theta}})$ always holds, which contradicts that $\hat{\boldsymbol{\theta}}$ is the minimizer of $Q_{1}(\boldsymbol{\theta})$.

\noindent \textbf{Asymptotic normality:}\\
Let
\begin{align}
 \label{Q_1_A}
 & Q_{1}\left\{\left(\boldsymbol{\theta}_{\mathcal{A}, 0}^{T}+n^{-1 / 2} \boldsymbol{\delta}^{T}, \mathbf{0}_{K-s_{1}}^{T}\right)^{T}\right\}-Q_{1}\left\{\left(\boldsymbol{\theta}_{\mathcal{A}, 0}^{T}, \boldsymbol{0}_{K-s_{1}}^{T}\right)^{T}\right\} \\ \nonumber
 =& L_{n}(\boldsymbol{\delta})+n \tilde{\gamma}_{1 n} \sum_{j=K+2}^{K+s_{1}} \tilde{w}_{j}\left(\left|\theta_{j, 0}+n^{-1 / 2} \delta_{j}\right|-\left|\theta_{j, 0}\right|\right) \\ \nonumber
 =& L_{n}(\boldsymbol{\delta})+n^{1 / 2} \tilde{\gamma}_{1 n} \sum_{j=K+2}^{K+s_{1}} \tilde{w}_{j} \delta_{j} \operatorname{sgn}\left(\theta_{j, 0}\right),
 \end{align}
 where $\boldsymbol{\delta} \in \mathbb{R}^{K+s_{1}}$ is a fixed vector. According to the Proposition 2.1, we know
$$
L_{n}(\boldsymbol{\delta}) \rightarrow-\sum_{k=1}^{K} \boldsymbol{\delta}^{T} \boldsymbol{S}_{k}+\frac{1}{2} \boldsymbol{\delta}^{T}\left(\sum_{k=1}^{K} {\bm{\Omega}}_{k, \mathcal{A}}\right) \boldsymbol{\delta},
$$
and
$$
n^{1 / 2} \tilde{\gamma}_{1 n} \sum_{K+2}^{K+s_{1}} \tilde{w}_{j} \delta_{j} \operatorname{sgn}\left(\theta_{j, 0}\right) \stackrel{p}{\longrightarrow} 0~~~(as~~n^{1 / 2} \tilde{\gamma}_{1 n} \rightarrow 0).
$$\\
Therefore, the (\ref{Q_1_A}) can be deduced to the following conclusion.
 $$
 Q_{1}\left\{\left(\boldsymbol{\theta}_{\mathcal{A}, 0}^{T}+n^{-1 / 2} \boldsymbol{\delta}^{T}, \mathbf{0}_{K-s_{1}}^{T}\right)^{T}\right\}-Q_{1}\left\{\left(\boldsymbol{\theta}_{\mathcal{A}, 0}^{T}, \mathbf{0}_{K-s_{1}}^{T}\right)^{T}\right\} \rightarrow-\sum_{k=1}^{K} \boldsymbol{\delta}^{T} \boldsymbol{S}_{k}+\frac{1}{2} \boldsymbol{\delta}^{T}\left(\sum_{k=1}^{K} {\bm{\Omega}}_{k, \mathcal{A}}\right) \boldsymbol{\delta}
 $$
 and the minimizer to (\ref{Q_1_A}), defined as $\hat{\boldsymbol{\delta}}$, follows the asymptotic normal distribution $N\left(0, {\bm{\Sigma}}_{\mathcal{A}}\right)$, where ${\bm{\Sigma}}_{{\mathcal{A}}}$is the covariance matrix of the oracle estimator given in Proposition 2.1.\\
 From the properties of convex functions, we know that the minimizer to (\ref{Q_1_A}) is unique and $n^{1 / 2}\left(\hat{\boldsymbol{\theta}}_{\mathcal{A}, F A L}-\boldsymbol{\theta}_{\mathcal{A}, 0}\right)$ is a minimizer to (\ref{Q_1_A}). Hence, we get $n^{1 / 2}\left(\hat{\boldsymbol{\theta}}_{\mathcal{A}, F A L}-\boldsymbol{\theta}_{\mathcal{A}, 0}\right)=\hat{\boldsymbol{\delta}} \stackrel{d}{\rightarrow} N\left(0, {\bm{\Sigma}}_{\mathcal{A}}\right)$.
\cite{Jiang-2014} stated that more general cases follow the similar exposition, but with more complicated notations. So the Theorem 2.2 is proved.
\end{proof}

\subsection{Proof of Proposition 2.2}
\begin{proof}
The proof is similar to the proof of Proposition 2.1 and thus is skipped.
\end{proof}

\subsection{Proof of Theorem 2.3}
\begin{proof}
Similar to the proof of Theorem 2.2.
\end{proof}
\subsubsection{Root-n Consistency Lemma}
\begin{lemm}[Root-n consistency of $\left.\hat{\boldsymbol{\theta}}_{FAS}\right)$]
Assume conditions $A 1-A 11$ hold, if $n^{1 / 2} \tilde{\gamma}_{2 n} \rightarrow 0$, then $\hat{\boldsymbol{\theta}}_{F A S}-\boldsymbol{\theta}_{0}=O_{p}\left(n^{-1 / 2}\right)$
\end{lemm}
\begin{proof}
As in Lemma 1, we only need to show that for any $\epsilon>0$, there exists a sufficiently large constant $\eta$, such that
$$
P\left\{\inf _{\|\xi\|=\eta} Q_2\left(\boldsymbol{\theta}_{0}+n^{-1 / 2} \boldsymbol{\xi}\right)>Q_{2}\left(\boldsymbol{\theta}_{0}\right)\right\} \geq 1-\epsilon
$$
where $Q_{2}(\cdot)$ is defined in (\ref{Objfunc_p2}) in section
$2.3.2$. Note that
\begin{align}
 \label{Q_2_inequality}
 & Q_{2}\left(\boldsymbol{\theta}_{0}+n^{-1 / 2} \boldsymbol{u}\right)-Q_{2}\left(\boldsymbol{\theta}_{0}\right) \\ \nonumber
 =& L_{n}(\boldsymbol{u})+n \tilde{\gamma}_{2 n} \sum_{l=1}^{g} \tilde{w}_{(l)}\left\{\left\|\boldsymbol{\theta}_{(l), 0}+n^{-1 / 2} \boldsymbol{u}_{(l)}\right\|_{\infty}-\left\|\boldsymbol{\theta}_{(l), 0}\right\|_{\infty}\right\}+n^{1 / 2} \tilde{\gamma}_{2 n} \sum_{l=g+1}^{p} \tilde{w}_{(l)}\left\{\left\|\boldsymbol{u}_{(l)}\right\|_{\infty}\right\} \\ \nonumber
  \geq & L_{n}(\boldsymbol{u})+n \tilde{\gamma}_{2 n} \sum_{l=1}^{g} \tilde{w}_{(l)}\left\{\left\|\boldsymbol{\theta}_{(l), 0}+n^{-1 / 2} \boldsymbol{u}_{(l)}\right\|_{\infty}-\left\|\boldsymbol{\theta}_{(l), 0}\right\|_{\infty}\right\},
\end{align}
where $\tilde{w}_{(l)}=\left[\max \left\{\left|\tilde{d}_{k, l}\right|, k=2, \cdots, K\right\}\right]^{-1}$ is the group-wise weight for the $l^{t h}$ predictor and $\boldsymbol{\theta}_{0}=\left\{\boldsymbol{\theta}_{(l), 0}, l=-1,0,1, \cdots, p\right\}$. Due to the assumption in section 2.3.2, we have
$$
n \tilde{\gamma}_{2 n} \sum_{l=1}^{g} \tilde{w}_{(l)}\left\{\| \boldsymbol{\theta}_{(l), 0}+n^{-1 / 2} \boldsymbol{\xi}_{(l)}\left\|_{\infty}-\right\| \boldsymbol{\theta}_{(l), 0} \|_{\infty}\right\} \stackrel{p}{\longrightarrow} 0~~~(as~~n^{1 / 2} \tilde{\gamma}_{2 n} \rightarrow 0).
$$\\
Combined with the Proposition 2.1, the right side of inequality (\ref{Q_2_inequality}) is dominated bu the quadratic term $M_{n}^{(2)}(\boldsymbol{\xi}) \geq 0$ when n is sufficiently large.
Based on assumption (A11), we know $Q_{2}\left(\boldsymbol{\theta}_{0}+n^{-1 / 2} \boldsymbol{\xi}\right) \geq Q_{2} \left(\boldsymbol{\theta}_{0}\right)$ is also always holds when $\|\boldsymbol{\xi}\|=\eta$.

\end{proof}

\subsubsection{Proof of Theorem 2.3}
\begin{proof}
~\\
\noindent\textbf{Sparsity:}\\
The proof is taken the similar arguments as theorem 2.1.
Under the assumption in section 2.3.2, $\boldsymbol{\theta}$ can decomposed as $ \boldsymbol{\theta}=\left(\boldsymbol{\theta}_{\mathcal{B}}^{T}, \boldsymbol{\theta}_{\mathcal{B}^{C}}^{T}\right)^{T}$, where $\boldsymbol{\theta}_{\mathcal{B}}=\left(\boldsymbol{\theta}_{(-1)}^{T}, \boldsymbol{\theta}_{(0)}^{T}, \ldots, \boldsymbol{\theta}_{(g)}^{T}\right)^{T}$ and $\boldsymbol{\theta}_{\mathcal{B}^{C}}=\left(\boldsymbol{\theta}_{(g+1)}^{T}, \ldots, \boldsymbol{\theta}_{(p)}^{T}\right)^{T}$.
Denote $\hat{\boldsymbol{\theta}}=\left(\hat{\boldsymbol{\theta}}_{\mathcal{B}}^{T}, \hat{\boldsymbol{\theta}}_{\mathcal{B}^{C}}^{T}\right)^{T}$ and by the way of contradiction suppose $\hat{\boldsymbol{\theta}}_{\mathcal{B}^{C}} \neq \mathbf{0}$.
Let $\boldsymbol{\theta}^{*}$ be a vector constructed by replacing $\hat{\boldsymbol{\theta}}_{\mathcal{B}^{C}}^{T}$ with $\boldsymbol{0}^{T}$ in $\hat{\boldsymbol{\theta}}$, that is,
$\boldsymbol{\theta}^{*}=\left(\hat{\boldsymbol{\theta}}_{\mathcal{B}}^{T}, \boldsymbol{0}_{(p-g) K}^{T}\right)^{T}$. 
Note that
\begin{align}
\label{Q_2_equation}
    & Q_{2}\left(\boldsymbol{\theta}^{*}\right)-Q_{2}(\hat{\boldsymbol{\theta}})\\
    =&\left\{Q_{2}\left(\boldsymbol{\theta}^{*}\right)-Q_{2}\left(\boldsymbol{\theta}_{0}\right)\right\}-\left\{Q_{2}(\hat{\boldsymbol{\theta}})-Q_{2}\left(\boldsymbol{\theta}_{0}\right)\right\}\nonumber \\ \nonumber
    =&L_{n}\left\{n^{1/2}\left(\hat{\boldsymbol{\theta}}_{\mathcal{B}}^{T}-\boldsymbol{\theta}_{\mathcal{B}, 0}^{T}, \boldsymbol{0}_{(p-g) K}^{T}\right)^{T}\right\}-L_{n}\left\{n^{1 / 2}\left(\hat{\boldsymbol{\theta}}-\boldsymbol{\theta}_{0}\right)\right\}-n \tilde{\gamma}_{2 n} \sum_{l=g+1}^{p} \tilde{w}_{(l)}\left\|\hat{\boldsymbol{\theta}}_{(l)}\right\|_{\infty}. \nonumber
\end{align}

Following the Lemma 2 of $\hat{\boldsymbol{\theta}}_{F A S}$, the first and second term in (\ref{Q_2_equation}) are both $O_{p}(1)$.
Similarly, the third term
\begin{align*}
    &-n \tilde{\gamma}_{2 n} \sum_{l=g+1}^{p} \tilde{w}_{(l)}\left\|\hat{\boldsymbol{\theta}}_{(l)}\right\|_{\infty}\\
    =&-n \tilde{\gamma}_{2 n} n^{1 / 2} \sum_{l=g+1}^{p}\left[\max \left\{n^{1 / 2}\left|\tilde{\boldsymbol{\theta}}_{(l)}\right|\right\}\right]^{-1}\left\|\hat{\boldsymbol{\theta}}_{(l)}\right\|_{\infty} \rightarrow o(n^{1/2})~~~(as~~n \tilde{\gamma}_{2 n} \rightarrow \infty),
\end{align*}
where $\max \left\{n^{1 / 2}\left|\tilde{\boldsymbol{\theta}}_{(l)}\right|\right\}=O_{p}(1)$ for $l=g+1, \ldots, p$, when $\tilde{\lambda}_{2 n} \rightarrow \infty$. Therefore, (\ref{Q_2_equation}) is dominated by the third term and $Q_{2}\left(\boldsymbol{\theta}^{*}\right)<Q_{2}(\hat{\boldsymbol{\theta}})$ always holds.

\noindent \textbf{Asymptotic normality:}\\
The proof is similar to the proof of Theorem 2.2 and thus is skipped.

\end{proof}

\end{appendix}

\end{document}